\pdfoutput=1

\documentclass[manuscript,screen]{acmart}

\usepackage[noend]{algpseudocode}
\usepackage{algorithm}
\usepackage{textcomp}
\usepackage{pgf,tikz}
\usepackage{pgfplots}
\pgfplotsset{compat=1.18}
\usetikzlibrary{arrows,positioning,calc,fadings,decorations.markings,arrows.meta}
\usepackage[braket,qm]{qcircuit}
\usepackage{makecell}
\usepackage{stackengine}

\newtheorem{problem}{Problem}
\newtheorem{remark}{Remark}

\renewcommand{\imath}{\mathrm{i}}
\newcommand{\BigO}{\mathcal{O}}
\newcommand{\id}{\mathbf{I}}
\newcommand{\I}{\mathcal{I}}
\newcommand{\E}{\mathcal{E}}
\newcommand{\F}{\mathcal{F}}
\newcommand{\EE}{\mathbf{E}}
\newcommand{\FF}{\mathbf{F}}
\newcommand{\M}{\mathbf{M}}
\newcommand{\h}{\mathbb{H}}
\newcommand{\den}{\mathcal{D}}
\newcommand{\her}{\mathcal{H}}
\newcommand{\supp}{\mathrm{supp}}
\newcommand{\spn}{\mathrm{span}}
\newcommand{\Pro}{\mathcal{P}}

\newcommand{\tr}{\mathrm{tr}}
\newcommand{\Prob}{\mathrm{TP}}

\newcommand{\TT}{\mathbf{t}}
\newcommand{\NT}{\mathbf{nt}}
\newcommand{\yes}{\mathbf{true}}
\newcommand{\no}{\mathbf{false}}
\newcommand{\CQ}{\textup{cq}}
\newcommand{\Sch}{\mathit{Sch}}
\newcommand{\B}{\mathbb{S}}
\newcommand{\inv}{\mathbb{I}}
\newcommand{\PP}{\mathbf{P}}
\newcommand\xrowht[2][0]{\addstackgap[.5\dimexpr#2\relax]{\vphantom{#1}}}

\AtBeginDocument{%
  }

\begin{document}

\title{Algorithmic Analysis of Termination Problems for Nondeterministic Quantum Programs}
\renewcommand{\shorttitle}{Termination Problems for Nondeterministic Quantum Programs}

\author{Jianling Fu}
\affiliation{%
  \institution{School of Computer Science and Technology, East China Normal University}
  \city{Shanghai}
  \country{China}}
\email{scsse\_fjl2015@126.com}

\author{Hui Jiang}
  \affiliation{%
    \institution{Shanghai Key Laboratory of Trustworthy Computing, East China Normal University}
    \city{Shanghai}
    \country{China}}
\email{yhq\_jh@126.com}

\author{Ming Xu}
  \affiliation{%
    \institution{Shanghai Key Laboratory of Trustworthy Computing, East China Normal University}
    \city{Shanghai}
    \country{China}}
\email{mxu@cs.ecnu.edu.cn}
\orcid{0000-0002-9906-5677}

\author{Yuxin Deng}
  \affiliation{%
    \institution{Shanghai Key Laboratory of Trustworthy Computing, East China Normal University}
    \city{Shanghai}
    \country{China}}
\email{yxdeng@sei.ecnu.edu.cn}

\author{Zhi-Bin Li}
\affiliation{%
  \institution{School of Computer Science and Technology, East China Normal University}
  \city{Shanghai}
  \country{China}}
\email{lizb@cs.ecnu.edu.cn}

\renewcommand{\shortauthors}{J.~Fu, H.~Jiang, M.~Xu, \textit{et~al.}}

\begin{abstract}
  Verifying quantum programs has attracted a lot of interest in recent years.
    In this paper,
    we consider the following two categories of termination problems of quantum programs with nondeterminism,
    namely:
    \begin{enumerate}
      \item Is an input of a program terminating with probability one under all schedulers?
      If not, how can a scheduler be synthesized to evidence the nontermination?
      \item Are all inputs terminating with probability one under their respective schedulers?
      If yes, a further question asks whether there is a scheduler
      that forces all inputs to be terminating with probability one together with how to synthesize it;
      otherwise, how can an input be provided to refute the universal termination?
    \end{enumerate}
    
    For the effective verification of the first category,
    we over-approximate the reachable set of quantum program states by the reachable subspace,
    whose algebraic structure is a linear space.
    On the other hand, we study the set of divergent states
    from which the program terminates with probability zero under \emph{some} scheduler.
    The divergent set also has an explicit algebraic structure.
    Exploiting these explicit algebraic structures,
    we address the decision problem by a necessary and sufficient condition,
    i.\,e.\@ the disjointness of the reachable subspace and the divergent set.
    Furthermore, the scheduler synthesis is completed in exponential time,
    whose bottleneck lies in computing the divergent set,
    reported for the first time.
    
    For the second category,
    we reduce the decision problem to the existence of invariant subspace,
    from which the program terminates with probability zero under \emph{all} schedulers.
    The invariant subspace is characterized by linear equations
    and thus can be efficiently computed.
    The states on that invariant subspace are evidence of the nontermination.
    Furthermore, the scheduler synthesis is completed by
    seeking a pattern of finite schedulers that forces all inputs to be terminating with positive probability.
    The repetition of that pattern yields the desired universal scheduler
    that forces all inputs to be terminating with probability one.
    All the problems in the second category are shown, also for the first time, to be solved in polynomial time.
    Finally, we demonstrate the aforementioned methods via a running example
    --- the quantum Bernoulli factory protocol.
\end{abstract}




\maketitle

\section{Introduction}
In the field of quantum computing,
physical devices have been rapidly developed in the last decades,
particularly in very recent years.
From the original experimental Deutsch's problem on a working 2-qubit quantum computer in 1998~\cite{CGK98},
to in 2020 Chinese quantum computer JiuZhang's implementing a type of Boson sampling on 76 photons~\cite{ZWD+20},
and in 2021 IBM's releasing its latest 127-bit ``Eagle'' processor~\cite{IBM127},
quantum computers have showed the advantage of computing
and attracted people to explore particular problems at a new level of complexity,
which cannot be achieved by their classical counterparts.

Equally important are quantum algorithms,
which are realized by quantum programming languages.
They build a bridge between the hardware devices and quantum algorithms to harness the power of quantum computers.
The first practical quantum programming language QCL appeared in~\cite{Omer98}.
The quantum guarded command language (qGCL)
was presented to program a ``universal'' quantum computer~\cite{SZ99},
and a nondeterministic structure was embedded into qGCL in the follow-up work~\cite{Zul04}.
Selinger~\cite{Sel04} and Altenkirch and Grattage~\cite{AG05} respectively proposed
functional programming languages QFC and QML with high-level features.
Nowadays, several quantum programming languages,
e.\,g.\@, Q\#~\cite{qsharp} and Cirq~\cite{cirq},
have been proposed for real-world applications,
enabling researchers to develop quantum software more conveniently and efficiently.

However, quantum programs have been found difficult to be tested or analyzed.
It is necessary to develop formal verification methods for quantum programs.
For classical program verification,
Morgan et al. extended the standard weakest precondition calculus~\cite{Dij76} 
into both probabilistic and nondeterministic settings~\cite{MMS96}.
Inspired by probabilistic predicate transformer, 
D'Hondt and Panangaden~\cite{DHP06} directly proposed a quantum analogy of the weakest preconditions
for a particular quantum programming language (QPL),
which was further employed by~\cite{FDJY07} and \cite{YDF10}
in reasoning about the correctness of deterministic quantum programs.
The quantum weakest precondition was designed for reasoning about
the expected runtime of quantum program~\cite{LZBY22}
and the expexted cost of various quantum resources~\cite{AMP+22}.
Meanwhile, Yu and Palsberg~\cite{YuP21} presented a framework of quantum abstract interpretation~\cite{CoC77}
to check assertions for the properties of large-scale quantum programs.
More theoretically,
one can decompose ``total correctness''
into ``partial correctness'' plus ``termination'' as advocated by C.\,A.\,R.\,Hoare~\cite{Hoa69}.
Hence termination analysis plays a central role in program verification,
which especially deserves studying when colliding with features in quantum setting.

In this paper,
we study the termination and the universal termination problems on nondeterministic quantum programs.
The former asks
whether \emph{an input state} of a program terminating with probability one under all schedulers;
the latter asks
whether \emph{all input states} are terminating with probability one under their respective schedulers.
We will give a series of positive results toward solving them.

\paragraph{Expressiveness on Program Models}
First of all,
we consider two models of quantum Markov decision processes with states in finite-dimensional Hilbert spaces
that interpret the operational semantics of nondeterministic quantum programs.
One has finitely many program locations,
so that it is easier to model practical scenarios;
the other has only one.
We show that they are of the same expressiveness,
and thus adopt the latter for ease of verification.

\paragraph{Precise Over-approximation of the Reachable Set}
In general, the set of reachable states of a program
does not exhibit any explicit algebraic structure,
which brings nontrivial hardness in verification.
To overcome it,
we give two definitions of reachable space
that over-approximate the set of reachable states.
The I-reachable space has the type of a subspace of the Hilbert space,
as proposed in~\cite{LYY14},
which is spanned by those vector representations of reachable states;
the II-reachable space is spanned by those Hermitian matrix representations of reachable states.
Both are computable in polynomial time,
specified in terms of the size of program model as usual,
i.\,e.\@ the dimension of Hilbert space.
But the latter is more precise,
which is validated by the running example of quantum Bernoulli factory protocol.

\paragraph{Algorithmic Complexity}
Moreover, we study the set of divergent states
from which the program terminates with probability zero under some scheduler.
By exploiting the algebraic structure of the divergent set,
an effective approach is developed to compute it in exponential time.
Combining the reachable spaces and the divergent set,
the termination is decided by a necessary and sufficient condition,
i.\,e.\@, the reachable subspace and the divergent set are disjoint.
The complexity of the decision procedure is in exponential time,
which is reported for the first time.

\paragraph{Scheduler Synthesis}
If the termination is decided to be false,
we know there are some schedulers that force the program not to terminate with probability one.
Scheduler synthesis is particularly important to resolve the nondeterminism in system design.
To achieve this,
we confine the nontermination scheduler into finitely many $\omega$-regular ones as candidates.
Conditioning on each candidate,
we derive a system of linear equations by Brouwer's fixedpoint theorem,
whose nonzero solutions help us to recognize the candidate as the nontermination scheduler.

We finally attack the universal termination problem.
It is reduced to the existence of invariant subspace,
from which the program terminates with probability zero under any scheduler.
The invariant subspace is also characterized by linear equations
and can be computed in polynomial time,
thus deciding the universal termination.
The states on that invariant subspace are evidence of the nontermination.
If the universal termination is decided to be true,
the scheduler synthesis is completed in polynomial time
by seeking a pattern of finite schedulers
that forces all input states to be terminating with positive probability.
The repetition of that pattern yields the desired universal scheduler
that forces all input states to be terminating with probability one.

\subsection{Related Work}
\paragraph{Verification on Probabilistic Programs}
In contrast to deterministic programs,
probabilistic programs have several syntactic constructors ---
probabilistic choice, nondeterministic choice and observation~\cite{KGJ+15}.
The termination problem yields many variants to be studied,
e.\,g.\@,
\begin{itemize}
	\item \emph{almost-sure termination} ---
	Does a program terminate with probability one?
	\item \emph{positive almost-sure termination} ---
	Is the expected running time of a program finite?
\end{itemize}
Although the almost-sure termination of probabilistic programs was proved to be undecidable in general~\cite{KaK15},
there are many approaches to attack it.
Fioriti and Hermanns~\cite{FiH15} proposed
a framework to prove almost-sure termination
by \emph{ranking super-martingales},
which is analogous to ranking functions on deterministic programs~\cite{BAG14}.
Chakarov and Sankaranarayanan applied constraint-based techniques
to generate linear ranking super-martingales~\cite[]{ChS13}.
Chatterjee \textit{et~al.}~\cite{CFG16} constructed polynomial ranking super-martingales extending linear ones.
A polynomial-time procedure was given in~\cite{ACN18}
to synthesize lexicographic ranking super-martingales
for linear probabilistic programs.
Fu and Chatterjee~\cite{FuC19} applied ranking super-martingales to
study the positive almost-sure termination of nondeterministic probabilistic programs.
McIver and Morgan~\cite{McM06} generalized the \emph{weakest preconditions} of Dijkstra
(an approach to prove correctness)
to the \emph{weakest pre-expectations}
for analyzing properties of
probabilistic guarded command language (pGCL)~\cite{HSM97}
and for establishing almost-sure termination~\cite{MMKK18}.
Kaminski \textit{et~al.}~\cite{KKM+16} presented a calculus of weakest pre-expectation style
for obtaining bounds on the expected running time of probabilistic programs.
Verification tools
like \textsc{Amber}~\cite{MBKK21}
have been released to automatically prove almost-sure and positive almost-sure termination.
However, in the setting of quantum computing,
a program state is no longer simply a probabilistic distribution over program variables;
it is instead a density operator (positive semi-definite matrix with unit trace) on Hilbert space,
which would be considered in the following.
Moreover, the state space admitting the finite-dimensional Hilbert space
could bring extra sugar in our analysis,
breaking the undecidability for probabilistic programs.

\paragraph{Verification on Quantum Programs}
Ying and Feng~\cite{YiF10} first studied the verification of quantum loop programs
by giving some necessary and sufficient conditions to ensure termination and almost-sure termination.
Later on, the classical Floyd--Hoare logic was extended in the quantum setting
to be quantum Floyd--Hoare logic~\cite{Ying11},
and the Sharir--Pnueli--Hart method was also extended
from probabilistic programs to quantum programs~\cite{YYF+13}
toward automatic verification~\cite{Ying19}.
The quantum Markov chain~\cite{FYY13,XHF21} could be a standard model
to interpret the operational semantics of deterministic quantum programs,
and the quantum Markov decision processes~\cite{YiY18}
could interpret the operational semantics of nondeterministic quantum programs.
Yu and Ying~\cite{YuY12} considered concurrent quantum programs,
and reduced the termination problem to the reachability problem of quantum Markov chains.

The work closest to ours is \cite{LYY14} where Li \textit{et~al.} dealt with nondeterministic quantum programs.
The nondeterminism in that program is used to model quantum processes,
and the program execution relies on a scheduler specified by the users.
Given a nondeterministic quantum programs,
the set of reachable states from an input state has no explicit algebraic structure in general,
which yields nontrivial hardness in verification.
The authors of~\cite{LYY14} proposed a polynomial-time method
for computing a linear space as the reachable space,
over-approximating the reachable set.
They also presented an algorithm to compute the set of divergent states
but the time complexity of the algorithm was left unsettled.
When the two sets are disjoint, the termination of a program can be inferred.
However, two remaining issues should be addressed,
i.\,e.\@,
i) how to analyze the complexity of computing the divergent set
and ii) how to synthesize the scheduler for nontermination.
Both will be solved in the current paper.

Recently, Li and Ying~\cite{LiY18} proposed the notions of additive and multiplicative invariants
that are predicates over program states at program locations
for quantum programs with nondeterministic choices, namely angelic and demonic choices.
They further synthesized linear ranking super-martingales
w.\,r.\,t.\@ an additive and/or multiplicative invariant
by semi-definite programming and quantifier elimination over real closed fields.
It was revealed in~\cite{LZBY22} that
almost-sure termination of a deterministic quantum program with an input state
is equivalent to positive almost-sure termination.
There are also many works for verifying various kinds of quantum protocols
and quantum algorithms~\cite{TNB11,AGN14,FHTZ15,LZW+19,XMGY21,XFM+22}.

\subsection{Contribution and Innovation}
The contributions of the current paper are summarized as follows:
\begin{enumerate}
	\item We propose a precise over-approximation of the reachable set,
	which can be computed in polynomial time.
	\item The complexity of computing the set of divergent states is given for the first time,
	thus settling an open problem.
	\item We decide the termination problem in exponential time
	and synthesize a nontermination scheduler provided that it exists.
	\item We decide the universal termination problem in polynomial time
	together with the synthesis of a universal scheduler for termination.
\end{enumerate}
To achive them, our technical innovations lie in:
i) using a tree construction for demonstrating the derivation of divergent states,
whose explicit structure could analyze the complexity of deciding termination;
ii) applying Knaster--Tarski fixedpoint theorem to ensure the efficiency of computing reachable spaces
while applying Brouwer's fixedpoint theorem to ensure the effectivity of scheduler synthesis.

\paragraph{Organization}
The rest of this paper is organized as follows.
Section~\ref{S2} recalls basic notions and notations from quantum computing.
The models of nondeterministic quantum programs are introduced in Section~\ref{S3}
together with their termination problems.
We compute the reachable spaces and the divergent set respectively
in Sections~\ref{S4} \&~\ref{S5}.
Combining them, we are able to decide the termination in Section~\ref{S6}.
We further solve the universal termination problem in Section~\ref{S7}.
Section~\ref{S9} is the conclusion.
For clarity, the implementation details is moved to the appendix.

\section{Preliminaries}\label{S2}
Let $\h$ be a finite Hilbert space
that is a complete vector space over complex numbers $\mathbb{C}$ equipped with an inner product operation,
and $d$ the dimension of $\h$ throughout this paper.
We recall the standard Dirac notations from quantum computing.
Interested readers can refer to~\cite{NC00} for more details.
\begin{itemize}
	\item $\ket{\psi}$ stands for a unit column vector in $\h$ labelled with $\psi$;
	\item $\bra{\psi}:=\ket{\psi}^\dag$ is the Hermitian adjoint
	(transpose and complex conjugate entrywise) of $\ket{\psi}$;
	\item $\ip{\psi_1}{\psi_2}:=\bra{\psi_1} \ket{\psi_2}$
	is the inner product of $\ket{\psi_1}$ and $\ket{\psi_2}$;
	\item $\op{\psi_1}{\psi_2}:=\ket{\psi_1} \otimes \bra{\psi_2}$ is the outer product
	where $\otimes$ denotes tensor product;
	\item $\ket{\psi,\psi'}$ is a shorthand of
	the tensor product $\ket{\psi}\ket{\psi'}=\ket{\psi}\otimes\ket{\psi'}$.
\end{itemize}
Let $\{\ket{i}: i=1,2,\dots,d\}$ be an orthonormal basis of $\h$.
Then any element $\ket{\psi}$, interpreted as a \emph{state}, of $\h$
can be entirely determined as $\ket{\psi}=\sum_{i=1}^d c_i\ket{i}$,
where $c_i \in \mathbb{C}$ ($i=1,2,\dots,d$) satisfy the normalization condition $\sum_{i=1}^d |c_i|^2=1$.
If $\ket{\psi}$ is linearly expressed by two or more elements $\ket{i}$ with nonzero coefficients,
it is said to be a \emph{super-position} of those elements $\ket{i}$.
For two spaces $\B$ and $\B'$,
the join $\B \vee \B'$ is the space spanned by the elements of $\B$ and $\B'$,
i.\,e.\@ $\spn(\B\cup\B')$.
For two quantum systems $\h$ and $\h'$,
the state space of their composite system is given by
the tensor product $\h\otimes \h'$
that is the Hilbert space spanned by the tensor products of elements in $\h$ and $\h'$,
i.\,e.\@ $\spn(\{\ket{\psi,\psi'}: \ket{\psi}\in\h \wedge \ket{\psi'}\in\h'\})$,
equipped with the inner product
$\langle \psi_1,\psi_1'\,|\,\psi_2,\psi_2'\rangle=\ip{\psi_1}{\psi_2}\langle\psi_1'\,|\,\psi_2'\rangle$
for any $\ket{\psi_1},\ket{\psi_2} \in \h$ and $|\psi_1'\rangle,|\psi_2'\rangle \in \h'$.

\paragraph{Linear Operators}
Let $\gamma$ be a linear operator on $\h$.
It is \textit{Hermitian} if $\gamma=\gamma^\dag$.
A Hermitian operator $\gamma$ is entirely determined
by its $d$ diagonal elements $\bra{i}\gamma\ket{i} \in \mathbb{R}$ ($i=1,2,\dots,d$)
and $d(d-1)/2$ off-diagonal elements $\bra{i}\gamma\ket{j} \in \mathbb{C}$ with $1 \le i<j \le d$
for a total of $d^2$ real numbers.
Let $\her(\h)$ be the set of Hermitian operators on $\h$.
For brevity, such a parameter $\h$ in $\her(\h)$ can be omitted if it is clear from the context.
For a Hermitian operator $\gamma$, we have the spectral decomposition
$\gamma=\sum_{i=1}^d \lambda_i\op{\lambda_i}{\lambda_i}$
where $\lambda_i \in \mathbb{R}$ ($i=1,2,\dots,d$) are the eigenvalues of $\gamma$
and $\ket{\lambda_i}$ (taking the meaningful labels $\lambda_i$) are the corresponding eigenvectors.
The \emph{support} of $\gamma$ is the subspace of $\h$ spanned by
all eigenvectors associated with nonzero eigenvalues,
i.\,e.\@ $\supp(\gamma):=\spn(\{\ket{\lambda_i}:i=1,2,\dots,d \wedge \lambda_i\ne 0\})$.
Although the spectral decomposition of $\gamma$ may be not unique,
the support of $\gamma$ must be unique,
since it is the orthocomplement of
the null space $\spn(\{\ket{\psi}\in\h:\gamma\ket{\psi}=0\})$ of $\gamma$.
So the notion of support is well defined.
A Hermitian operator $\gamma$ is \emph{positive}
if $\bra{\psi}\gamma\ket{\psi} \ge 0$ holds for any $\ket{\psi}\in\h$.
A \emph{projector} $\PP$ is a positive operator of
the form $\sum_{i=1}^m \op{\psi_i}{\psi_i}$ with $m\le d$,
where $\ket{\psi_i}$ ($i=1,2,\dots,m$) are orthonormal.
It implies that all eigenvalues of $\PP$ are in $\{0,1\}$.

\paragraph{Quantum States}
The \emph{trace} of a linear operator $\gamma$ is defined as
$\tr(\gamma):=\sum_{i=1}^d \bra{\psi_i} \gamma \ket{\psi_i}$
for any orthonormal basis $\{\ket{\psi_i}:i=1,2,\dots,d\}$.
It is unique as it equals the sum of all eigenvalues of $\gamma$.
A \emph{density} operator $\rho$ is a positive operator with unit trace;
a partial density operator $\rho$ is a positive operator with trace not greater than $1$.
Let $\den$ be the set of density operators,
and $\den^{\le 1}$ the set of partial density operators.
For a density operator $\rho$,
we have the spectral decomposition
$\rho=\sum_{i=1}^m \lambda_i \op{\lambda_i}{\lambda_i}$
where $\lambda_i$ ($i=1,2,\dots,m$) are positive eigenvalues.
We call such eigenvectors $\ket{\lambda_i}$ \emph{eigenstates} of $\rho$ explained below.
The density operators are usually used to describe quantum states.
Under that decomposition,
it means that the quantum system is in state $\ket{\lambda_i}$ with probability $\lambda_i$.
When $m=1$, we know that the system is surely in state $\ket{\lambda_1}$ (with probability one),
which is the so-called \emph{pure} state;
otherwise the state is \emph{mixed}.
Both the vector notation $\ket{\lambda_i}$ and the outer product notation $\op{\lambda_i}{\lambda_i}$
could be employed to denote pure states.
An alternative way to describe quantum states is
the \emph{probabilistic ensemble} $\{(\ket{\psi_k},p_k):k=1,2,\ldots\}$ with $\sum_k p_k=1$.
It means that the system is the mixture of being in state $\ket{\psi_k}$ with probability weight $p_k$.
Here $\ket{\psi_k}$ ($k=1,2,\ldots$) are not necessarily orthogonal.

\paragraph{Quantum Operations}
Super-operators $\E$ on $\h$ are linear operators on the (ground) linear operators on $\h$.
Particularly, (completely-positive) super-operators are employed to describe quantum operations.
It is usually described by the Kraus representation $\E=\{\EE_i:i=1,2,\dots,m\}$,
entailing that for a given density operator $\rho$,
we have $\E(\rho)=\sum_{i=1}^m \EE_i \rho \EE_i^\dag$.
Here the number $m$ of \emph{Kraus operators} $\EE_i$ could be bounded by $d^2$
without loss of generality,
since the (ground) linear operators $\rho$ are Hermitian,
i.\,e.\@ a linear space of dimension $d^2$,
and thereby there are at most $d^2$ linearly independent linear operators on that space.
We will use the bracket notation $\{\EE_i:i=1,2,\dots,m\}$
to denote (the Kraus representation of) a super-operator $\E$.
For two super-operators $\E=\{\EE_i:i=1,2,\dots,m\}$ and $\E'=\{\EE_j':j=1,2,\dots,m'\}$,
their sum $\E+\E'$ is given by $\{\EE_i:i=1,2,\dots,m\} \cup \{\EE_j':j=1,2,\dots,m'\}$;
their composition $\E \circ \E'$
with associative law $\E \circ \E'(\rho)=\E(\E'(\rho))$
is $\{\EE_i\EE_j':i=1,2,\dots,m \wedge j=1,2,\dots,m'\}$.
Let $\I$ be the identity super-operator, and $\id$ the identity operator.
A super-operator $\E$ is \emph{trace-preserving}, denoted $\E \eqsim \I$,
if $\sum_{i=1}^m \EE_i^\dag\EE_i = \id$,
due to
\[
	\tr(\E(\rho)) =\tr\left(\sum_{i=1}^m \EE_i \rho \EE_i^\dag\right)
	=\tr\left(\sum_{i=1}^m \EE_i^\dag \EE_i \rho\right)
	=\tr(\id \rho) = \tr(\rho);
\]
it is \emph{trace-nonincreasing}, denoted $\E \lesssim \I$,
if $\id - \sum_{i=1}^m \EE_i^\dag\EE_i$ is positive,
due to
\[
\begin{aligned}
	\tr(\op{\psi}{\psi}) - \tr(\E(\op{\psi}{\psi}))
    & = \tr(\id \op{\psi}{\psi}) - \tr\left(\sum_{i=1}^m \EE_i^\dag \EE_i \op{\psi}{\psi}\right) \\
    & = \tr\left(\left(\id-\sum_{i=1}^m \EE_i^\dag \EE_i\right) \op{\psi}{\psi}\right)
	= \bra{\psi} \left(\id-\sum_{i=1}^m \EE_i^\dag \EE_i\right) \ket{\psi} \ge 0.
\end{aligned}
\]
Let $\mathcal{S}$ be the set of super-operators,
$\mathcal{S}^{\eqsim \I}$ the set of trace-preserving ones,
and $\mathcal{S}^{\lesssim \I}$ the set of trace-nonincreasing ones.

\paragraph{Quantum Measurement}
A finite set of projectors $\PP_i$ with index $i$ ranging over $\mathit{IDX}$ forms a \emph{projective measurement}
if $\sum_{i \in \mathit{IDX}} \PP_i = \id$.
The measurement aims to extract classical information from quantum states,
but it may destroy the quantum state.
Specifically, given a quantum state $\rho$,
after the above projective measurement,
we will get an outcome
$i \in \mathit{IDX}$ with probability $p_i=\tr(\PP_i \rho)$;
when the outcome is $i$, the final state would be \emph{collapsed} to $\PP_i \rho\PP_i/p_i$.
The measurement process is not reversible.
For a projector $\PP_i$ and a quantum state $\rho$,
$\tr(\PP_i \rho)=0$ implies that the outcome $i$ does not occur,
which holds if and only if $\supp(\PP_i)$ is orthogonal to $\supp(\rho)$.
For a super-operator $\E=\{\EE_i:i=1,2,\ldots,m\}$ and a pure state $\op{\psi}{\psi}$,
we have $\supp(\E(\op{\psi}{\psi}))=\spn(\{\EE_i\ket{\psi}:i=1,2,\ldots,m\})$.
Finally we would mention a useful inclusion:
\begin{equation}\label{eq:support}
	\supp(\E(\op{\psi}{\psi})) \subseteq \bigvee_{k=1}^K \supp(\E(\op{\psi_k}{\psi_k}))
\end{equation}
holds for any $\ket{\psi} \in \spn(\{\ket{\psi_k}:k=1,2,\ldots,K\})$.
It follows from the fact:
Assume the RHS of~\eqref{eq:support} is a proper subspace of $\h$;
otherwise the inclusion follows trivially.
Let $|\psi^\perp\rangle$ be an element of $\h$
orthogonal to all $\supp(\E(\op{\psi_k}{\psi_k}))$ with $k=1,2,\ldots,K$.
It is also orthogonal to $\EE_i \ket{\psi_k}$
for each $i=1,2,\ldots,m$ and $k=1,2,\ldots,K$,
i.\,e.\@ $\langle\psi^\perp| \EE_i \ket{\psi_k}=0$.
It implies that $\langle\psi^\perp| \EE_i \ket{\psi}=0$ holds for each $i=1,2,\ldots,m$,
and thus $|\psi^\perp\rangle$ is orthogonal to $\supp(\E(\op{\psi}{\psi}))$.

\section{Program Model and Termination Problems}\label{S3}
In this section,
we introduce a nondeterministic extension of a quantum \emph{while}-language,
and interpret its operational semantics by two models of quantum Markov decision processes (quantum MDPs).
One model is more complicated but easier to model practical scenarios
while the other is simpler and easier to be verified.
They are shown to have the same expressiveness.
For ease of verification,
we will adopt the latter to represent nondeterministic quantum programs later on.
Finally, we introduce the termination problems of nondeterministic quantum programs considered in the paper.

\subsection{Program Model}
\begin{definition}[{\cite[Chapter~6]{Yin16}}]\label{NQPL}
    A nondeterministic quantum program is generated by the following syntax:
    \[
    \begin{aligned}
        S \ \triangleq \
        q := \ket{0} \,|\, & \Bar{q} := U[\Bar{q}]
		\,|\, S_1;S_2
        \,|\, \square_{j=1}^m \ S_j \\
        \,|\, & \mathbf{if}\ \M[\Bar{q}] = \yes\ \mathbf{then} \ S_1 \\
        \,|\, & \mathbf{while}\ \M[\Bar{q}] = \yes\ \mathbf{do} \ S_1 \enspace.
    \end{aligned}
    \]
\end{definition}

We briefly explain the syntax as follows:
\begin{itemize}
    \item The initialization ``$q := \ket{0}$'' sets quantum variable $q$ to the basis state $\ket{0}$ on $\h_q$,
    where the Hilbert space $\h_q$ of $q$ is supposed to
    have basis states $\ket{0},\ket{1}$ for $q$ being Boolean
    and $\ket{0},\ket{1},\ldots,\ket{k-1}$ for $q$ being integer.
    For any appointed pure state $\ket{\psi}\in\h$,
    there is a unitary operation $\mathbf{U}$
    such that the state $\ket{\psi} = \mathbf{U} \ket{0}$ can be prepared.
    \item The assignment ``$\Bar{q} := U[\Bar{q}]$'' performs the unitary transformation $U$ on the register $\Bar{q}$.
    For operations in classical programs,
    we can implement the quantum analogue by constructing a corresponding unitary operator.
    For example,
    we can take the unitary operation $\mathbf{U}_{+1} = \sum_{i=0}^{2^{32}-1}|i+1 \mod 2^{32}\rangle\langle i|$
    as the quantum counterpart to implement the classical increment assignment $x := x + 1$
    on a $32$-bit integer $x$.
    \item The statement ``$S_1;S_2$'' represents the sequential composition.
    \item The nondeterministic choice ``$\square_{j=1}^m\ S_j$'' means that
    a subprogram $S_j$ ($j=1,2,\ldots,m$) is nondeterministically chosen to execute.
    The nondeterminism will be resolved by
    some action $\alpha_j$ from the set $\mathbf{ACT}=\{\alpha_1,\alpha_2,\ldots,\alpha_m\}$,
    indicating the subprogram $S_j$ to be executed.
    \item The statement
    \[
        \mathbf{if}\ \M[\Bar{q}] = \yes\ \mathbf{then} \ S_1 
    \]
    is a quantum analogue of the classical condition statement.
    As the guard condition,
    a quantum measurement $\M= \{ \M_\yes, \M_\no\}$
    is performed on the register $\Bar{q}$.
    If the outcome of the measurement is $\yes$ whose probability is $p_\yes = \tr(\M_\yes[\Bar{q}])$,
    the state in register $\Bar{q}$ is collapsed into $\Bar{q}_\yes := \M_\yes[\Bar{q}]/p_\yes$,
    and the subprogram $S_1$ will be applied to $\Bar{q}_\yes$,
    resulting in $S_1[\Bar{q}_\yes]$;
    otherwise the state in register $\Bar{q}$ is collapsed into $\Bar{q}_\no := \M_\no[\Bar{q}]/p_\no$
    where $p_\no = \tr(\M_\no[\Bar{q}])$
    and the subprogram $S_1$ will not be applied.
    It is easy to see $p_\yes+p_\no=1$.
    Note that if $p_\yes=0$ (resp.~$p_\no=0$),
    meaning the outcome $\yes$ (resp.~$\no$) cannot be observed,
    we do not need to further consider this branch.  
    In the viewpoint of static analysis,
    the resulting state would be the mixture of $S_1[\Bar{q}_\yes]$
    and $\Bar{q}_\no$ with probability $p_\yes$ and $p_\no$, respectively.
    \item The loop statement 
    \[
        \mathbf{while}\ \M[\Bar{q}] = \yes\ \mathbf{do} \ S_1 
    \]
    admits a projective measurement $\M = \{\M_\yes,\M_\no\}$ as the guard condition,
    so that exactly one of the two outcome $\yes$ and $\no$ would occur
    after performing that measurement $\M$.
    If the outcome $\no$ is observed, 
    the program ends in the collapsed state $\Bar{q}_\no$;
    if the outcome $\yes$ occurs,
    the subprogram $S_1$ will be applied to $\Bar{q}_\yes$ and repeat the above process.
    All in all, the resulting state would be a mixture of countably many states in static analysis.
\end{itemize}

A nondeterministic quantum program is a \emph{finite} statement
generated by the syntax in Definition~\ref{NQPL}.
For a program $\Pro$,
we could assign each statement with a program location
as specified in the left column of Table~\ref{tab:semantics},
in which the nondeterministic choice, the condition and the loop statements are of appropriate wrapping.
All of such locations are collected into the finite set $\mathbf{LOC} = \{\ell_i: i = 1,2,\ldots,n\}$,
in which $\ell_n$ is the end location, indicating no statement to be executed.
Let $\mathbf{VAR}$ be the set of program variables of $\Pro$,
and $\mathbf{ACT}$ the set of actions that resolve the nondeterministic choices.
For a variable $q \in \mathbf{VAR}$, the state space is $\h_q$.
The state space $\h_\mathbf{VAR}$ of all program variables is simply
the tensor product of individual state spaces $\h_q$,
i.\,e.\@ $\h_\mathbf{VAR}=\bigotimes_{q \in \mathbf{VAR}} \h_q$.

\begin{table}[ht]
    \centering  
    \caption{The operational semantics of nondeterministic quantum programs}\label{tab:semantics}
    \begin{tabular}{|c|c|}
       \hline\xrowht{12pt}
       \textbf{statements} & \textbf{operational semantics} \\ 
       \hline\xrowht{24pt}
        $\makecell[l]{\begin{aligned}
            \ell_i: & \quad q := \ket{0} \\ \ell_{i+1}: & \quad S_1
        \end{aligned}}$
        & \makecell[l]{$(\ell_i, \rho) \xrightarrow{\mathbf{ACT}}
        (\ell_{i+1}, \sum_{j=0}^{k-1} \ket{0}_q\!\bra{j} \rho \ket{j}_q\!\bra{0})$ \\
        where the superscript $\mathbf{ACT}$ is abbreviated for any action in it, \\
        the subscript $q$ indicates which variable is involved, \\
        and $k = \dim(\h_q)$} \\
        \hline\xrowht{24pt}
        $\makecell[l]{\begin{aligned}
            \ell_i: & \quad \Bar{q}:=U[\Bar{q}] \\ \ell_{i+1}: & \quad S_1
        \end{aligned}}$
        & \makecell[l]{$(\ell_i,\rho) \xrightarrow{\mathbf{ACT}} (\ell_{i+1},\mathbf{U} \rho \mathbf{U}^\dag)$} \\
        \hline\xrowht{36pt}
        $\makecell[l]{\begin{aligned}
            \ell_i: & \quad \square_{j=1}^m \\ 
            & \begin{aligned}
                \ell_{i+j}: & \quad S_j \\
            \end{aligned} \\
            \ell_{i+m+1}: & \quad S_{m+1}
        \end{aligned}}$
        & \makecell[l]{$(\ell_i,\rho) \xrightarrow{\alpha_j} (\ell_{i+j},\rho)$} \\
        \hline\xrowht{36pt}
        $\makecell[l]{\begin{aligned}
            \ell_i: & \quad \mathbf{if}\ \M[\Bar{q}] = \yes \ \mathbf{then} \\
            & \ell_{i+1}: \quad S_1 \\
            \ell_j: & \quad S_2
        \end{aligned}}$
        & \makecell[l]{$(\ell_i,\rho) \xrightarrow{\mathbf{ACT}} (\ell_{i+1}, \M_\yes \rho \M_\yes)$ \\
        $(\ell_i,\rho) \xrightarrow{\mathbf{ACT}} (\ell_j, \M_\no \rho \M_\no)$} \\
        \hline\xrowht{36pt}
        $\makecell[l]{\begin{aligned}
            \ell_i: & \quad \mathbf{while}\ \M[\Bar{q}] = \yes\ \mathbf{do} \\
            & \ell_{i+1}: \quad S_1 \\
        \ell_j: & \quad S_2
        \end{aligned}}$
        & \makecell[l]{$(\ell_i,\rho) \xrightarrow{\mathbf{ACT}} (\ell_{i+1}, \M_\yes \rho \M_\yes)$ \\
        $(\ell_i,\rho) \xrightarrow{\mathbf{ACT}} (\ell_j, \M_\no \rho \M_\no)$} \\
        \hline\xrowht{12pt}
        $\makecell[l]{\ell_n:}$
        & \makecell[l]{$(\ell_n,\rho) \xrightarrow{\mathbf{ACT}} (\ell_n, \rho)$} \\
        \hline
    \end{tabular}
\end{table}

\begin{example}\label{ex1-1}
    We here consider a quantum Bernoulli factory protocol~\cite{KeO94,JZS18}
    which serves as a running example of our method.
    Alice and Bob explore a quantum analogue of Bernoulli factory constructed on two quantum coins,
    named \emph{quoins},
    regardless of the probability of producing a head when tossing a single quoin.
    The process of quantum Bernoulli factory is wrapped into a black box,
    composed of tossing one quoin,
    flipping the status of the remaining one and observing the status of two quoins.
    Once Alice and Bob enter the initialized quoins into the box and trigger the box to start,
    the status of the quoins will be hidden from them until both agree to check.
    Finally, Alice and Bob bet according to the indicated state of the quoins.
    The detailed protocol is described as:
    \begin{enumerate}
        \item Two quoins are referred to as the left and the right ones.
        \item It nondeterministically chooses one of the two quoins to toss, and the other one is flipped.
        \item If the left quoin is head and the right is tail, then Alice wins;
        if the right quoin is head and the left is tail, then Bob wins;
        otherwise, they end in a draw.
    \end{enumerate}
    
    Intuitively,
    tossing a quoin would produce the result ``head'' or ``tail'' with equal probability,
	independent to the initial status of the two quoins,
	so that it makes the bet fair.
    However, after applying the quantum Bernoulli factory,
    Alice and Bob want to know whether the result is \emph{defective} 
    in the sense
    that neither Alice nor Bob eventually has a chance of winning.
    Let us take Alice's stand to check the defectiveness in quantum setting,
    which is similar from Bob's stand.
    
    In order to describe the process of the protocol,
    we design a nondeterministic quantum program $\Pro_1$
    with program variables $\mathbf{VAR}=\{q_1,q_2\}$
    and locations $\mathbf{LOC}=\{\ell_1,\ell_2,\ldots,\ell_9\}$ as follows.
    \begin{algorithmic}
        \State $\ell_1$:\quad $q_1:=\ket{0}$
        \State $\ell_2$:\quad $q_2:=\ket{0}$
        \State $\ell_3$:\quad $q_1:=U_1[q_1]$
        \State $\ell_4$:\quad $q_2:=U_2[q_2]$
		\State $\ell_5$:\quad $\mathbf{while}$ $\M[q_1;q_2]=\yes$ $\mathbf{do}$
		\State \quad $\ell_6$:\quad $\square_{j=1}^2$
        \State \quad \quad $\ell_7$:\quad $(X_2\circ H_1)[q_1;q_2]$
        \State \quad \quad $\ell_8$:\quad $(X_1\circ H_2)[q_1;q_2]$
        \State $\ell_9$:
    \end{algorithmic}
    Both $\h_{q_1}$ and $\h_{q_2}$ are the one-qubit Hilbert space with orthonormal basis $\{\ket{0},\ket{1}\}$
    where $\ket{0}$ and $\ket{1}$ denote ``head'' and ``tail'' respectively.
    The state space of the two program variables $\h_\mathbf{VAR}=\h_{q_1}\otimes\h_{q_2}$ is
    a two-qubit Hilbert space.
    The unitary transformations $U_i$ ($i=1,2$) transform the initial one-qubit state $\ket{0}$ in registers $q_i$
    into any one-qubit state $\ket{\psi}$ to be prepared.
    For instance, we choose $U_i=\{\mathbf{X}\}$ where $\mathbf{X}=\op{0}{1}+\op{1}{0}$ is the bit-flip,
    so that $\mathbf{X}\ket{0}=\ket{1}$ is prepared in our setting.
    The status of two quoins prior to the \emph{while}-loop can be viewed as
    the composite quantum state $q_1;q_2:=\ket{1,1}\in\h_\mathbf{VAR}$.

    For the while-loop, a projective measurement $\M=\{\M_\yes,\M_\no\}$ is designed as the guard condition,
	where $\M_\yes = \op{0,1}{0,1}$
	and $\M_\no = \id_\mathbf{VAR}-\M_\yes = \op{0,0}{0,0} + \op{1,0}{1,0} + \op{1,1}{1,1}$
	are referred to the events ``the left quoin is head and the right is tail''
	and the complement, respectively.
    Whenever we enter the while loop, i.\,e.\@ being at location $\ell_6$,
    a nondeterministic choice corresponding to tossing the left or the right quoin
    should be resolved by some action from the set $\mathbf{ACT}=\{\alpha_1,\alpha_2\}$,
    which leads to location $\ell_7$ or $\ell_8$
    and the program will return to location $\ell_5$ after that.
    Finally the program would be expected to terminate at location $\ell_9$.

    Tossing the quoin $q_i$ is modelled by
    applying the Hadamard gate $H_i=\{\mathbf{H}\}$ where $\mathbf{H}=\op{+}{0}+\op{-}{1}=\op{0}{+}+\op{1}{-}$
	with $\ket{\pm}=(\ket{0}\pm\ket{1})/\sqrt{2}$ on the $i$th qubit,
    which means that $\ket{0}$ is transformed into $\ket{+}$ and $\ket{1}$ is transformed into $\ket{-}$,
    resulting in the super-positions of ``head'' and ``tail'' with equal probability.
    Flipping the quoin $q_i$ is modelled by applying the bit-flip gate $X_i=\{\mathbf{X}\}$ on the $i$th qubit.
\end{example}

\subsection{Operational Semantics}
We have seen that during the execution of a program,
the program states should take program locations into consideration.
To this end,
we will interpret the operational semantics of a nondeterministic quantum program $\Pro$
by a model of quantum MDP on the Hilbert space $\h_\mathbf{VAR}$ (quantum information)
with program locations $\mathbf{LOC}$ (classical information).
Let us review the model of quantum MDP first.

\begin{definition}\label{QMDP}
	A quantum Markov decision process (quantum MDP for short) on Hilbert space $\h$
	is a quadruple $(S, \Sigma, \E, \M)$, in which
    \begin{itemize}
		\item $S=\{s_i:i=1,2,\dots,n\}$ is a finite set of classical states;
		\item $\Sigma=\{\alpha_j:j=1,2,\dots,m\}$ is a finite set of actions;
		\item $\E: (S \times \Sigma \times S) \to \mathcal{S}^{\lesssim \I}$ gives rise to
        the super-operators $\E_{i,j,k}$ on $\h$
        that characterize the transitions from state $s_i$ to $s_k$ by taking action $\alpha_j$,
		satisfying that $\sum_{s_k\in S} \E_{i,j,k}\eqsim\I$ holds
		for each $s_i\in S$ and each $\alpha_j\in \Sigma$;
		\item $\M$ is a projective measurement on $\h_\CQ=\mathcal{C}\otimes\h$
        with $\mathcal{C}=\spn(\{\ket{s_i}:i=1,2,\dots,n\})$.
    \end{itemize}
\end{definition}
Note that in the classical model of MDP,
there is a probability-allocation function attached to state transitions,
which is generalised to the density operator-allocation function
by those super-operators $\E_{i,j,k}$ in Definition~\ref{QMDP}.
Additionally, to extract classical information from quantum states,
the projective measurement $\M$ is adopted here.

\begin{definition}\label{II-QMDP}
    For a nondeterministic quantum program $\Pro$ with
    program variables $\mathbf{VAR}$, actions $\mathbf{ACT}$ and locations $\mathbf{LOC}$,
    the quantum MDP that interprets $\Pro$ 
	is a quadruple $(\mathbf{LOC},\mathbf{ACT},\linebreak[0]\rightarrow,\{\M_\TT,\M_\NT\})$
    on Hilbert space $\h_\mathbf{VAR}$, where
    \begin{itemize}
		\item the transition relation $\rightarrow$,
        whose entries $(\ell_i,\rho) \xrightarrow{\alpha_j} (\ell_k,\rho')$
        characterize the transitions from location $\ell_i$ to $\ell_k$ by taking action $\alpha_j$
        while changing quantum states $\rho$ to $\rho'$,
        is given by the quantum operations in the right column of Table~\ref{tab:semantics}, and
        \item $\{\M_\TT,\M_\NT\}$ is a projective measurement on $\spn(\{\ket{\ell}: \ell \in \mathbf{LOC}\})$,
        in which $\M_\TT=\op{\ell_n}{\ell_n}$ refers to the end location of $\Pro$
        and $\M_\NT=\sum_{\ell\in \mathbf{LOC}\setminus \{\ell_n\}} \op{\ell}{\ell}$ refers to the complement.
    \end{itemize}
\end{definition}
Here, the projectors $\M_\TT$ and $\M_\NT$ on the space $\spn(\{\ket{\ell}: \ell \in \mathbf{LOC}\})$
are short for the trivial ones $\M_\TT \otimes \id_\mathbf{VAR}$ and $\M_\NT\otimes \id_\mathbf{VAR}$
on the product Hilbert space $\spn(\{\ket{\ell}: \ell \in \mathbf{LOC}\}) \otimes \h_\mathbf{VAR}$.

When a nondeterministic program $\Pro$ executes,
it has finitely many actions $\alpha_1, \alpha_2,\ldots,\alpha_m$
to choose at each location $\ell_i$ of nondeterministic statements.
Each action $\alpha_j\in\Sigma$ is attached with a series of super-operators $\E_{i,j,k}$,
where $\ell_k$ ranges over $\mathbf{LOC}$,
satisfying $\sum_{s_k\in S} \E_{i,j,k}\eqsim\I$.
The nondeterminism is resolved by a sequence of actions.
An infinite sequence $\sigma=\sigma(1)\,\sigma(2)\cdots$ with $\sigma(i)\in\Sigma$
is called an \emph{infinite scheduler} (scheduler for short),
and a finite sequence $\varsigma=\sigma(1)\,\sigma(2)\cdots\sigma(k)$ is a \emph{finite scheduler}.

\begin{example}\label{ex1-2}
    Consider the nondeterministic quantum program $\Pro_1$
    with actions $\mathbf{ACT}=\{\alpha_1,\alpha_2\}$
    and locations $\mathbf{LOC}=\{\ell_1,\ell_2,\ldots,\ell_9\}$
    in Example~\ref{ex1-1}.
    Since the program terminates at the location $\ell_9$,
    we can obtain a projective measurement $\{\M_\TT,\M_\NT\}$
    with $\M_\TT = \op{\ell_9}{\ell_9}$
    and $\M_\NT = \sum_{\ell\in \mathbf{LOC}\setminus \{\ell_9\}}{\op{\ell}{\ell}}$.
    Thus we construct a quantum MDP $\mathcal{M}_1$ interpreting $\Pro_1$,
    which is characterized by the quadruple $(\mathbf{LOC},\mathbf{ACT},\rightarrow,\{\M_\TT,\M_\NT\})$ 
    with the transition relation $\rightarrow$ given by the middle column of Table~\ref{tab:ex2}.
    Here $\rho_1=(\id\otimes\bra{0}) \rho (\id\otimes\ket{0}) + (\id\otimes\bra{1}) \rho (\id\otimes\ket{1})$
    and $\rho_2=(\bra{0}\otimes\id) \rho (\ket{0}\otimes\id) + (\bra{1}\otimes\id) \rho (\ket{1}\otimes\id)$
    are the reduced density operators of $\rho$ that trace out the states on $\h_{q_2}$ and $\h_{q_1}$, respectively.
    
    \begin{table}[ht]
    \centering
    \caption{Translating the nondeterministic quantum program to the quantum MDPs}\label{tab:ex2}
    \scalebox{0.96}{
    \begin{tabular}{|l|l|l|}
        \hline\xrowht{36pt}
        \makecell[c]{\textbf{original statements} \\
        \textbf{in the quantum program $\Pro_1$} \\
        \textbf{from Example~\ref{ex1-1}}}
        & \makecell[c]{\textbf{transition relation $\rightarrow$} \\
        \textbf{in the quantum MDP $\mathcal{M}_1$} \\
        \textbf{as described in Definition~\ref{II-QMDP}}}
        & \makecell[c]{\textbf{transition super-operator $\E$} \\
        \textbf{in the quantum MDP $\mathcal{M}_1'$} \\
        \textbf{as described in Definition~\ref{I-QMDP}}} \\
        \hline\xrowht{36pt}
        \makecell[l]{$\begin{aligned}
            \ell_1: &\quad q_1:=\ket{0}\\
            \ell_2:&\quad\cdots
        \end{aligned}$} & $(\ell_1,\rho) \xrightarrow{\mathbf{ACT}} (\ell_2,\op{0}{0}\otimes\rho_2)$ 
        & \makecell[l]{add Kraus operators \\
            $\op{\ell_2}{\ell_1}\otimes\op{0}{0}\otimes\id$
            and $\op{\ell_2}{\ell_1}\otimes\op{0}{1}\otimes\id$ \\
            to $\E(\vec{\alpha})$ for all $\vec{\alpha}\in \Sigma$} \\
        \hline\xrowht{36pt}
        \makecell[l]{$\begin{aligned}
            \ell_2:&\quad q_2:=\ket{0}\\
            \ell_3:&\quad \cdots
        \end{aligned}$} & $(\ell_2,\rho) \xrightarrow{\mathbf{ACT}} (\ell_3,\rho_1\otimes\op{0}{0})$
        & \makecell[l]{add Kraus operators \\
            $\op{\ell_3}{\ell_2}\otimes\id\otimes\op{0}{0}$
            and $\op{\ell_3}{\ell_2}\otimes\id\otimes\op{0}{1}$ \\
            to $\E(\vec{\alpha})$ for all $\vec{\alpha}\in \Sigma$} \\
        \hline\xrowht{24pt}
        \makecell[l]{$\begin{aligned}
            \ell_3:&\quad q_1:=U_1[q_1]\\
            \ell_4:&\quad \cdots
        \end{aligned}$} & $(\ell_3,\rho) \xrightarrow{\mathbf{ACT}} (\ell_4,U_1(\rho))$
        & \makecell[l]{add Kraus operator $\op{\ell_4}{\ell_3}\otimes \mathbf{X}\otimes\id$ \\
            to $\E(\vec{\alpha})$ for all $\vec{\alpha}\in \Sigma$} \\
        \hline\xrowht{24pt}
        \makecell[l]{$\begin{aligned}
            \ell_4:&\quad q_2:=U_2[q_2]\\
            \ell_5:&\quad \cdots
        \end{aligned}$} & $(\ell_4,\rho) \xrightarrow{\mathbf{ACT}} (\ell_5,U_2(\rho))$
        & \makecell[l]{add Kraus operator $\op{\ell_5}{\ell_4}\otimes\id\otimes \mathbf{X}$ \\
        to $\E(\vec{\alpha})$ for all $\vec{\alpha}\in \Sigma$} \\
        \hline\xrowht{36pt}
        \makecell[l]{$\begin{aligned}
            \ell_5:&\quad \mathbf{while}\ \M[q_1;q_2]=\yes\ \mathbf{do}\\
            &\ell_6: \quad \cdots \\
            \ell_9:&\quad \cdots
        \end{aligned}$} & $\makecell[l]{\begin{aligned}
            &(\ell_5,\rho) \xrightarrow{\mathbf{ACT}} (\ell_6,\M_\yes\rho\M_\yes)\\
            &(\ell_5,\rho) \xrightarrow{\mathbf{ACT}} (\ell_9,\M_\no\rho\M_\no)
        \end{aligned}}$ &\makecell[l]{
        add Kraus operators $\op{\ell_6}{\ell_5}\otimes\op{0,1}{0,1}$ and\\
        $\op{\ell_9}{\ell_5}\otimes(\op{0,0}{0,0}+\op{1,0}{1,0}+\op{1,1}{1,1})$\\
        to $\E(\vec{\alpha})$ for all $\vec{\alpha}\in\Sigma$
        }\\
        \hline\xrowht{48pt}
        \makecell[l]{$\begin{aligned}
            \ell_6:&\quad \square_{j=1}^2\\
            &\ell_7:\quad \cdots\\
            &\ell_8:\quad \cdots
        \end{aligned}$}& $\makecell[l]{\begin{aligned}
            &(\ell_6,\rho) \xrightarrow{\alpha_1} (\ell_7,\rho)\\
            &(\ell_6,\rho) \xrightarrow{\alpha_2} (\ell_8,\rho) 
        \end{aligned}}$ & \makecell[l]{add Kraus operator $\op{\ell_7}{\ell_6}\otimes\id_{\mathbf{VAR}}$ \\
            to $\E(\vec{\alpha})$ for all $\vec{\alpha}\in \Sigma$ with $\vec{\alpha}(6)=\alpha_1$, \\
            add Kraus operator $\op{\ell_8}{\ell_6}\otimes\id_{\mathbf{VAR}}$ \\
            to $\E(\vec{\alpha})$ for all $\vec{\alpha}\in \Sigma$ with $\vec{\alpha}(6)=\alpha_2$} \\
        \hline\xrowht{24pt}
        \makecell[l]{$\begin{aligned}
            \ell_5:&\quad \mathbf{while}\ \cdots\ \mathbf{do} \\
            & \ell_7:\quad (X_2\circ H_1)[q_1;q_2]
         \end{aligned}$}
        & $(\ell_7,\rho) \xrightarrow{\mathbf{ACT}} (\ell_5,(X_2\circ H_1)(\rho))$
        & \makecell[l]{add Kraus operator 
        $\op{\ell_5}{\ell_7}\otimes\mathbf{H}\otimes\mathbf{X}$ \\
        to $\E(\vec{\alpha})$ for all $\vec{\alpha}\in\Sigma$} \\
        \hline\xrowht{24pt}
        \makecell[l]{$\begin{aligned}
            \ell_5:&\quad \mathbf{while}\ \cdots\ \mathbf{do} \\
            & \ell_8: \quad (X_1\circ H_2)[q_1;q_2]
        \end{aligned}$} 
        & $(\ell_8,\rho) \xrightarrow{\mathbf{ACT}} (\ell_5,(X_1\circ H_2)(\rho))$
        & \makecell[l]{add Kraus operator 
        $\op{\ell_5}{\ell_8}\otimes\mathbf{X}\otimes\mathbf{H}$ \\
        to $\E(\vec{\alpha})$ for all $\vec{\alpha}\in\Sigma$} \\
        \hline\xrowht{24pt}
        $\ell_9$: & $(\ell_9,\rho) \xrightarrow{\mathbf{ACT}} (\ell_9,\rho)$
        & \makecell[l]{add Kraus operator $\op{\ell_9}{\ell_9}\otimes\id_{\mathbf{VAR}}$ \\
        to $\E(\vec{\alpha})$ for all $\vec{\alpha}\in \Sigma$} \\
        \hline
    \end{tabular}}
\end{table}

    Starting at the location $\ell_1$ and
    given a finite action sequence $\varsigma = \alpha_1\,\alpha_1\,\alpha_1\,\alpha_1\,\alpha_1\,\alpha_2\,\alpha_1$,
    the run of the quantum MDP $\mathcal{M}_1$ generated by $\varsigma$ is
    \[
        \begin{aligned}
            (l_1,\rho)
            & \xrightarrow{\alpha_1} (\ell_2,\op{0}{0}\otimes\rho_2)
            \xrightarrow{\alpha_1} (\ell_3,\op{0,0}{0,0})
            \xrightarrow{\alpha_1} (\ell_4,\op{1,0}{1,0}) 
            \xrightarrow{\alpha_1} (\ell_5,\op{1,1}{1,1}) \\
            & \xrightarrow{\alpha_1} (\ell_6,\op{1,1}{1,1})
            \xrightarrow{\alpha_2} (\ell_8,\op{1,1}{1,1})
            \xrightarrow{\alpha_1} (\ell_5,\op{0,-}{0,-}).
        \end{aligned}
    \]
    The quantum state would evolve into $\op{0,-}{0,-}$
    while $\mathcal{M}_1$ runs into the location $\ell_5$. \qed
\end{example}

Sometimes,
we would focus on the structure of \emph{while}-loop that plays a central role in termination analysis.
It is a subclass of nondeterministic quantum programs,
which terminates when refuting the guard condition instead of entering the end location.
Thus the location information can be omitted for brevity.
We could interpret the while-loop by the model of quantum MDP as follows: 

\begin{definition}[{\cite[Definition~1]{LYY14}}]\label{I-QMDP}
    For a nondeterministic quantum while-loop $\Pro$
    with program variables $\mathbf{VAR}$ and actions $\mathbf{ACT}$,
    the quantum MDP that interprets $\Pro$ is
    a triple $(\Sigma,\E,\{\M_\yes,\M_\no\})$ on Hilbert space $\h_\mathbf{VAR}$, where:
    \begin{itemize}
        \item $\Sigma=\mathbf{ACT}$;
        \item $\E: \Sigma \to \mathcal{S}^{\eqsim \I}$ gives rise to
        the super-operators $\E_j$ on $\h_\mathbf{VAR}$ by taking action $\alpha_j$;
        \item $\{\M_\yes,\M_\no\}$ is a projective measurement on $\h_\mathbf{VAR}$
		with the outcomes $\yes$ and $\no$ referring to the nontermination and the termination, respectively.
    \end{itemize}
\end{definition}

It is worth noting that in this model,
a measurement is performed on the current quantum state
to determine whether the program terminates or not
before taking each action.
In case the program does not terminate,
an action $\alpha_j$ will be nondeterministically chosen
and the corresponding super-operator $\E_j$ will be applied to the current quantum state.
The program keeps running step by step like the above execution until it terminates,
but it is unnecessary to consider the change on the location after executing every step.

\begin{example}\label{ex1-3}
    Review the nondeterministic quantum program $\Pro_1$ in Example~\ref{ex1-1}.
    There is an embedded quantum while-loop $\Pro_2$ (from location $\ell_5$ to $\ell_9$)
    with actions $\mathbf{ACT}=\{\alpha_1,\alpha_2\}$ and guard condition $\{\M_\yes,\M_\no\}$.
    We can interpret it simply by the quantum MDP $\mathcal{M}_2=(\mathbf{ACT},\E,\{\M_\yes,\M_\no\})$
    with the input state $\rho_0=\op{1,1}{1,1}$,
    where
    \[
    \begin{aligned}
        \E(\alpha_1) =\E_1 &= X_2 \circ H_1 =\{\mathbf{H}\otimes\mathbf{X}\} \\
        \E(\alpha_2) =\E_2 &= X_1 \circ H_2 =\{\mathbf{X}\otimes\mathbf{H}\}. 
    \end{aligned}
    \]
    We define the Kraus operators $\EE_1=\mathbf{H}\otimes\mathbf{X}$
    and $\EE_2=\mathbf{X}\otimes\mathbf{H}$ on $\h_\mathbf{VAR}$ for use afterwards. \qed
\end{example}

Although the model in Definition~\ref{II-QMDP} seems much easier to manipulate
than that in Definition~\ref{I-QMDP},
they are of the same expressiveness indicated by the following lemma.
Hence, we can freely choose one of the two definitions for convenience.
In this paper,
we will mainly adopt the model in Definition~\ref{I-QMDP} for ease of verification.

\begin{lemma}
	The model in Definition~\ref{II-QMDP} has the same expressiveness as that in Definition~\ref{I-QMDP}.
\end{lemma}
\begin{proof}
Given a quantum MDP $(\Sigma,\E,\{\M_\yes,\M_\no\})$ on Hilbert space $\h_\mathbf{VAR}$
in Definition~\ref{I-QMDP},
we can obtain another quantum MDP $(\mathbf{LOC},\mathbf{ACT},\rightarrow,\{\M_\TT,\M_\NT\})$
on Hilbert space $\h_\mathbf{VAR}$,
as described in Definition~\ref{II-QMDP},
by introducing two locations $\ell_1$ and $\ell_2$.
Then every transition of $\E$ makes a self-loop at $\ell_1$ if it is not terminating,
otherwise is led to $\ell_2$. 
Formally it is constructed as:
\begin{itemize}
    \item setting the location set $\mathbf{LOC}=\{\ell_1,\ell_2\}$,
    so that $\M_\TT=\op{\ell_2}{\ell_2}$ and $\M_\NT=\op{\ell_1}{\ell_1}$,
    \item setting the action set $\mathbf{ACT}=\Sigma$,
    \item $(\ell_1,\rho) \xrightarrow{\alpha_j} (\ell_2,\M_\no\rho\M_\no)$
    and $(\ell_1,\rho) \xrightarrow{\alpha_j} (\ell_1,\rho')$
    follow from $\rho'=\E_j(\M_\yes\rho\M_\yes)$ where $\E_j=\E(\alpha_j)$ for $\alpha_j \in \Sigma$.
\end{itemize}

Conversely, given a quantum MDP $(\mathbf{LOC},\mathbf{ACT},\rightarrow,\{\M_\TT,\M_\NT\})$
with locations $\mathbf{LOC}=\{\ell_1,\ell_2,\ldots,\ell_n\}$
on Hilbert space $\h_\mathbf{VAR}$ in Definition~\ref{II-QMDP},
we can obtain another quantum MDP $(\Sigma,\E,\{\M_\yes,\M_\no\})$
on the product Hilbert space $\mathcal{C}\otimes\h_\mathbf{VAR}$
with $\mathcal{C} = \spn(\{\ket{\ell}: \ell\in \mathbf{LOC})$,
as described in Definition~\ref{I-QMDP},
by quantitizing the location information $\mathbf{LOC}$ into $\mathcal{C}$.
Formally it is constructed as:
\begin{itemize}
    \item setting the action set $\Sigma=\mathbf{ACT}^n$ that is the $n$-fold of $\mathbf{ACT}$,
    \item $\E(\vec{\alpha})(\sum_{i=1}^n \op{\ell_i}{\ell_i}\otimes\rho_i)
    =\sum_{i=1}^n \sum_{\ell_{k_i}\in\mathbf{LOC}} |\ell_{k_i}\rangle\langle\ell_{k_i}| \otimes\rho_{k_i}$,
    where $\E(\vec{\alpha})$ is the super-operator of $\E$ by
    taking an action $\vec{\alpha}=(\alpha_{j_1},\alpha_{j_2},\ldots,\alpha_{j_n})\in \Sigma$,
    follows from the series of $(\ell_i,\rho_i) \xrightarrow{\alpha_{j_i}} (\ell_{k_i},\rho_{k_i})$
    with $\ell_{k_i}$ ranging over $\mathbf{LOC}$,
    \item setting $\M_\no=\op{\ell_n}{\ell_n} \otimes \id_\mathbf{VAR}$
    and $\M_\yes=\sum_{\ell \in \mathbf{LOC}\setminus \{\ell_n\}} \op{\ell}{\ell} \otimes \id_\mathbf{VAR}$. \qedhere
\end{itemize}
\end{proof}

\begin{example}\label{ex1-4}
    For the quantum MDP $\mathcal{M}_1=(\mathbf{LOC},\mathbf{ACT},\rightarrow,\{\M_\TT,\M_\NT\})$ in Example~\ref{ex1-2},
    we can construct an equally expressive quantum MDP $\mathcal{M}_1'= (\Sigma,\E,\{\M_\yes,\M_\no\})$
    with the following components:
    \begin{itemize}
        \item the input state $\op{\ell_1}{\ell_1}\otimes\rho_0$ for the input state $\rho_0$ of $\mathcal{M}_1$,
        \item the action set $\Sigma=\mathbf{ACT}^9$ as $|\mathbf{LOC}|=9$,
        \item the transition super-operator $\E$ is constructed part by part
        in the right column of Table~\ref{tab:ex2}
        and summarized as
        \[
            \E(\vec{\alpha}) = \left\{\begin{array}{l}
                \op{\ell_2}{\ell_1}\otimes\op{0}{0}\otimes\id,\
                \op{\ell_2}{\ell_1}\otimes\op{0}{1}\otimes\id, \\
                \op{\ell_3}{\ell_2}\otimes\id\otimes\op{0}{0},\
                \op{\ell_3}{\ell_2}\otimes\id\otimes\op{0}{1}, \\
                \op{\ell_4}{\ell_3}\otimes \mathbf{X}\otimes\id,\
                \op{\ell_5}{\ell_4}\otimes\id\otimes \mathbf{X}, \\
                \op{\ell_6}{\ell_5}\otimes\op{0,1}{0,1},\
                \op{\ell_9}{\ell_5}\otimes(\op{0,0}{0,0}+\op{1,0}{1,0}+\op{1,1}{1,1}), \\
                \underline{\op{\ell_7}{\ell_6}\otimes\id_{\mathbf{VAR}}},\
                \op{\ell_5}{\ell_7}\otimes\mathbf{H}\otimes\mathbf{X},\
                \op{\ell_5}{\ell_8}\otimes\mathbf{X}\otimes\mathbf{H},\
                \op{\ell_9}{\ell_9}\otimes\id_{\mathbf{VAR}}
            \end{array}\right\}
        \]
        if the $6$th component $\vec{\alpha}(6)$ of the 9-tuple $\vec{\alpha} \in \Sigma$ is $\alpha_1$,
        while replacing the underlined Kraus operator $\op{\ell_7}{\ell_6}\otimes\id_{\mathbf{VAR}}$
        with $\op{\ell_8}{\ell_6}\otimes\id_{\mathbf{VAR}}$ if $\vec{\alpha}(6)=\alpha_2$, and
        \item the projectors $\M_\no=\op{\ell_9}{\ell_9}\otimes\id_\mathbf{VAR}$
        and $\M_\yes=\sum_{\ell\in\mathbf{LOC}\setminus \{\ell_9\}}{\op{\ell}{\ell}\otimes\id_\mathbf{VAR}}$. \qed
    \end{itemize}
\end{example}

An execution scheduler of a program can be represented as a sequence of actions as in Definition~\ref{I-QMDP}. 
We define the super-operator $\F_{\alpha_j}=\E_j \circ \{\M_\yes\}$ ($\alpha_j\in \Sigma$)
as the composite quantum operation upon the measure outcome of nontermination;
let $\varsigma \uparrow k$ be the finite prefix of $\varsigma$ with length $k$
for $k \le |\varsigma|$,
and $\varsigma \downarrow k$ the suffix obtained by removing the $k$-prefix from $\varsigma$.
Then we have the inductive construction of the super-operator over a sequence of actions
\[
    \F_\varsigma = \begin{cases}
    \I & \textup{if }|\varsigma|=0 \\
    \F_{\varsigma \downarrow 1} \circ \F_{\varsigma \uparrow 1}
    & \textup{if }|\varsigma| \ge 1.
    \end{cases}
\]
For example, for a finite schedule $\varsigma=\alpha_1\alpha_2\alpha_3$,
we have $\varsigma \uparrow 1=\alpha_1$, $\varsigma \downarrow 1=\alpha_2\alpha_3$,
and $\F_\varsigma =\F_{\alpha_1\alpha_2\alpha_3}
=\F_{\alpha_2\alpha_3}\circ\F_{\alpha_1}
=\F_{\varsigma \downarrow 1}\circ\F_{\varsigma \uparrow 1}$.
The construction of the super-operator over a sequence of actions can be extended to infinite schedulers $\sigma$.

From now on, we employ the model of quantum MDP $(\Sigma,\E,\{\M_\yes,\M_\no\})$ in Definition~\ref{I-QMDP}
to represent nondeterministic quantum programs $\Pro$.
The size of $\Pro$ is dominated by $\BigO(m \cdot d^4)$ where $m=|\Sigma|$ and $d=\dim(\h)$,
since $\E$ has $m$ super-operators $\E_j$ for $\alpha_j \in \Sigma$
and each super-operator $\E_j$ has at most $d^2$ Kraus operators that are $d$-by-$d$ matrices.
For brevity, we measure the size of $\Pro$ simply by the two parameters $m$ and $d$.
All the $m \cdot d^4$ numbers in $\E$ are supposed to be \emph{algebraic numbers}
that are roots of the polynomials with rational coefficients.
Algebraic numbers are widely used in quantum computing,
such as $1/\sqrt{2}$ appearing in the Hadamard gate
and the imaginary unit $\imath$ appearing in the Pauli gate $\imath\op{1}{0}-\imath\op{0}{1}$.
Arithmetic operations (addition, subtraction, multiplication and division) on algebraic numbers
are further supposed to be of unit cost as usual, i.\,e.\@ $\BigO(1)$.
These basics will build up our complexity analysis on later.

\subsection{Termiantion Problems}
We are to deliver the termination probabilities of nondeterministic quantum programs
and the termination problems.
\begin{definition}[Termination Probability]
	For a nondeterministic quantum program $\Pro=(\Sigma,\E,\{\M_\yes,\M_\no\})$ in Definition~\ref{I-QMDP}
	and an input state $\rho \in \den$,
	\begin{enumerate}
		\item the (accumulative) termination probability under a finite scheduler $\varsigma$ is
		\[
		\Prob_\varsigma(\rho)=\sum_{i=0}^{|\varsigma|}\tr(\M_\no\F_{\varsigma\uparrow i}(\rho));
        \]
		\item the termination probability under an infinite scheduler $\sigma$ is
		\[
		\Prob_\sigma(\rho)=\sum_{i=0}^\infty\tr(\M_\no\F_{\sigma\uparrow i}(\rho));
        \]
        \item the termination probability (for conservation) of $\Pro$ is
        $\Prob(\rho)=\inf_{\sigma\in\Sigma^\omega}\Prob_{\sigma}(\rho)$.
	\end{enumerate}
\end{definition}
It is not hard to see $\Prob_\varsigma(\rho)=\tr(\rho)-\tr(\M_\yes\F_\varsigma(\rho))$.

\begin{problem}[Termination]\label{P1}
Given a nondeterministic quantum program $\Pro=(\Sigma,\E,\{\M_\yes,\M_\no\})$
and an input state $\rho \in \den$,
is $\rho$ terminating with probability one under all schedulers $\sigma$,
i.\,e.\@ $\forall\,\sigma\in\Sigma^\omega: \Prob_\sigma(\rho)=1$?
If not, how can a scheduler $\sigma$ be synthesized to evidence the nontermination?
\end{problem}

\begin{problem}[Weak Termination]\label{P2}
	Given a nondeterministic quantum program $\Pro=(\Sigma,\E,\{\M_\yes,\M_\no\})$
    and an input state $\rho \in \den$,
	is $\rho$ terminating with probability one under some scheduler $\sigma$,
    i.\,e.\@ $\exists\,\sigma\in\Sigma^\omega: \Prob_\sigma(\rho)=1$?
\end{problem}
	
\begin{problem}[Optimal Termination]\label{P3}
	Given a nondeterministic quantum program $\Pro=(\Sigma,\E,\{\M_\yes,\M_\no\})$
    and an input state $\rho \in \den$,
	what is the angelic (resp.~demonic) scheduler $\sigma$
	that maximizes (resp.~minimizes) the termination probability,
    i.\,e.\@ $\arg\max_{\sigma\in\Sigma^\omega} \Prob_\sigma(\rho)$
    (resp.~$\arg\min_{\sigma\in\Sigma^\omega} \Prob_\sigma(\rho)$)?
\end{problem}

\begin{problem}[Universal Termination]\label{P4}
    Given a nondeterministic quantum program $\Pro=(\Sigma,\E,\{\M_\yes,\M_\no\})$,
	are all input states $\rho$ terminating with probability one
    under their respective schedulers $\sigma$, i.\,e.\@
    \begin{subequations}
    \begin{equation}
        \forall\,\rho \in \den\ \exists\,\sigma\in\Sigma^\omega: \Prob_\sigma(\rho)=1?
    \end{equation}
	If yes, a further question asks
    whether there is a scheduler $\sigma$
    that forces all input states $\rho$ to be terminating with probability one, i.\,e.\@
    \begin{equation}
        \exists\,\sigma\in\Sigma^\omega\ \forall\,\rho \in \den: \Prob_\sigma(\rho)=1,
    \end{equation}
    \end{subequations}
    together with how to synthesize such a scheduler $\sigma$;
    otherwise, how can an input $\rho$ be provided to refute the universal termination?
\end{problem}

The first three problems are specified with an input state,
while the last one is not,
since it concerns the termination on all input states that is a ``universal'' problem.
Specifically, Problem~\ref{P1} requires the (strong) termination under \emph{all} schedulers,
Problem~\ref{P2} requires the weak termination under \emph{some} scheduler,
both are concerned with \emph{qualitative} properties.
Problem~\ref{P3} is on \emph{quantitative} property,
which seems to be harder than Problems~\ref{P1} \&~\ref{P2},
since for a given input state $\rho$,
the program terminates with probability one under all schedulers
if $\Prob(\rho)=\inf_{\sigma\in\Sigma^\omega}\Prob_\sigma(\rho)=1$
while it terminates under some scheduler
if $\sup_{\sigma\in\Sigma^\omega}\Prob_\sigma(\rho)=1$.
However, as shown in~\cite[Example~1]{YiY18} that
such an optimal scheduler does not exist,
Problem~\ref{P3} is not well-posed sometimes.
We will solve Problems~\ref{P1} \&~\ref{P4} in the coming sections,
and left the remaining Problem~\ref{P2} as future work.

\section{Computing the Reachable Spaces}\label{S4}
In this section,
we introduce the reachable spaces for a nondeterministic quantum program
starting from an input state.
They over-approximate the set of reachable states in order to obtain an explicit algebraic structure,
which is crucial for an algorithmic analysis.
We review the notion of reachable space
together with the construction method presented in~\cite{LYY14}.
Then we propose a more precise notion of reachable space.
Two kinds of reachable spaces are said to be of types I and II respectively,
and both are computable in polynomial time
w.\,r.\,t.\@ the dimension of the state space $\h$ and the number of actions in $\mathbf{ACT}$
as the existing literature~\cite{LYY14}.

\begin{definition}[Reachable Set]\label{ReachStates}
	Given a nondeterministic quantum program $\Pro$ and an input state $\rho\in\den$,
	the set of reachable states of $\Pro$ starting from $\rho$ is
    $\Psi(\Pro,\rho)=\{\F_\varsigma(\rho): \varsigma \in \Sigma^*\}$.
\end{definition}

The elements $\gamma$ in the set $\Psi(\Pro,\rho)$ are certainly reachable states from $\rho$.
Here, the reachability is specified in a \emph{qualitative} sense
that requires a unit probability of reachability under some finite scheduler $\varsigma$. 
Can we specify the reachability in a \emph{quantitative} sense?
To answer it, we investigate what states are in a given state $\rho \in \den$.
Supposing that $\rho$ is the uniform distribution $\id/d$,
we know that:
\begin{itemize}
    \item any pure state $\ket{\psi} \in \h$ is in $\rho$,
    which is with probability $\tfrac{1}{d}$,
    \item any state $\gamma \in \den$ is also in $\rho$,
    which is with probability $1/(d\cdot\lambda_{\max})$
    where $\lambda_{\max}$ is the maximal eigenvalue of $\gamma$.
\end{itemize}
Overall, a state is in $\rho$,
provided that it has a positive probability in some probabilistic ensemble of $\rho$.
Developing this concept,
a state is said \emph{reachable} from $\rho$,
provided that it has a positive probability of reachability
under some finite scheduler $\varsigma$.
So the eigenstates $\ket{\lambda}$ with positive eigenvalue $\lambda$ of $\gamma \in \Psi(\Pro,\rho)$
are (pure) reachable states;
and even the elements $\ket{\psi}$ in the support of $\gamma$ are (pure) reachable states too,
since, by~\cite[Exercise~2.73]{NC00} (refer to Appendix~\ref{A1} for self-containedness),
there is a minimal probabilistic ensemble of $\gamma$ containing $\ket{\psi}$ with positive probability $p$,
i.\,e.\@ $\gamma=p\op{\psi}{\psi}+\sum_k p_k\op{\psi_k}{\psi_k}$
for some $\ket{\psi_k} \in \supp(\gamma)$ with the probability sum $p+\sum_k p_k=1$.

It is obvious to see that
$\Psi(\Pro,\rho)$ is a countable set without explicit algebraic structure in general,
which yields nontrivial hardness in verification.
To overcome it,
we would like to introduce the notion of \emph{reachable space}.

\begin{definition}[I-Reachable Space {\cite[Definition~3]{LYY14}}]\label{ReachSpace}
	Given a nondeterministic quantum program $\Pro$ and an input state $\rho\in\den$,
	the type I reachable space of $\Pro$ starting from $\rho$ is
    \[\Phi(\Pro,\rho)=\bigvee_{\gamma\in\Psi(\Pro,\rho)}\supp(\gamma).\]
\end{definition}

From the above definition, we see that
for two elements $\gamma_1$ and $\gamma_2$ of $\Psi(\Pro,\rho)$
that are reachable under finite schedulers $\varsigma_1$ and $\varsigma_2$ respectively
and for $\ket{\psi_i} \in \supp(\gamma_i)$ ($i\in\{1,2\}$),
all super-positions $\ket{\psi}$ of $\ket{\psi_1}$ and $\ket{\psi_2}$ are elements of $\Phi(\Pro,\rho)$,
but they are unnecessarily required to be reachable
since the construction does not ensure
the existence of a common finite scheduler $\varsigma$ that generates $\ket{\psi}$.
In this sense, the I-reachable space is known to be a superset of the reachable set.
More precisely, we have:
\begin{itemize}
\item $\Psi(\Pro,\rho) \subset \den(\h)$
since $\Psi(\Pro,\rho)$ is countable while $\den(\h)$ is a continuum that is uncountable,
\item $\Phi(\Pro,\rho) \subseteq \h$, and further
\item $\Psi(\Pro,\rho) \subset \den(\Phi(\Pro,\rho))$.
\end{itemize}
Thus, to show that a property holds on the reachable set $\Psi(\Pro,\rho)$,
it is sufficient to show that the property holds on
all density operators $\den(\Phi(\Pro,\rho))$ on the reachable space $\Phi(\Pro,\rho)$.
The latter has the nice algebraic structure of a finite-dimensional linear space,
which is promising to be effectively verified.

To get an explicit description of the reachable space,
we resort to the following program model that has only one action and thus resolves nondeterminism:
\begin{definition}[Average Quantum Program {\cite[Definition~4]{LYY14}}]\label{AQP}
	Let $\Pro=(\Sigma,\E,\{\M_\yes,\M_\no\})$
    with $\Sigma=\{\alpha_j:j=1,2,\dots,m\}$ and $\E(\alpha_j)=\E_j$
    be a nondeterministic quantum program.
    Then the average quantum program $\bar{\Pro}$ of $\Pro$ is
    the pair $(\bar{\E},\{\M_\yes,\M_\no\})$,
	where
    \begin{itemize}
    \item $\bar{\E}$ is the arithmetic average of $\E$,
	i.\,e.\@, for any program state $\rho\in\den$,
	the effect of the average super-operator $\bar{\E}$ performed on $\rho$ is
    $\tfrac{1}{m}\sum_{j=1}^m \E_j(\rho)$.
    \end{itemize}
\end{definition}

\begin{lemma}[{\cite[Lemma~1]{LYY14}}]\label{RSofAP}
	Given a nondeterministic quantum program $\Pro$ and an input state $\rho\in\den$,
	the I-reachable subspace of $\Pro$ starting from $\rho$ is
    that of the quantum program $\bar{\Pro}$ averaging $\Pro$ starting from $\rho$,
	i.\,e.\@ $\Phi(\Pro,\rho)=\Phi(\bar{\Pro},\rho)$. 
\end{lemma}

This lemma reveals that $\Pro$ agrees with $\bar{\Pro}$ on the I-reachable subspace,
despite $\Pro$ does not on the reachable set.
Using it, the I-reachable space of $\Pro$ can be obtained
as the least fixedpoint of the ascending chain of linear subspaces of $\h$:
\begin{equation}\label{eq:chain}
\begin{aligned}
\supp(\rho_0) & \subseteq \supp(\rho_0) \vee \supp(\rho_1) \\
& \subseteq \supp(\rho_0) \vee \supp(\rho_1) \vee \supp(\rho_2) \\
& \subseteq \cdots,
\end{aligned}
\end{equation}
where $\rho_i=\bar{\F}^i(\rho)$ with $\bar{\F}=\bar{\E}\circ\{\M_\yes\}$.
Namely, we denote this chain by $\B_0 \subseteq \B_1 \subseteq \B_2 \subseteq \cdots$,
in which each linear space $\B_i$ is computed upon the average quantum program $\bar{\Pro}$.
The following lemma gives an upper bound for
the occurrence of the least fixedpoint in the ascending chain,
thus establishing the computability.

\begin{lemma}\label{lem:lfp}
	Let $\B_0 \subseteq \B_1 \subseteq \B_2 \subseteq \cdots$
    be the ascending chain of nonnull linear subspaces $\B_i \subseteq \h$,
    as defined in~\eqref{eq:chain}.
	Then there is an integer $\ell \le \dim(\h)-1$ such that $\B_k=\B_\ell$ holds for all $k>\ell$. 
\end{lemma}
\begin{proof}
The function $F$ mapping from $\B_i$ to $\B_{i+1}$ ($i \ge 0$) can be formulated as
a monotonic function
\[
    F(\mathbb{X})=\mathbb{X} \vee \bigvee_{\ket{\psi}\in \mathbb{X}} \supp(\bar{\F}(\op{\psi}{\psi})).
\]
Meanwhile, all subspaces $\B^\textup{all}$ of $\h$ form
a complete lattice $(\B^\textup{all},\subseteq,\inf,\sup)$
by taking `$\inf$' as the meet $\bigwedge=\bigcap$ and `$\sup$' as the join $\bigvee$.
By Knaster--Tarski fixedpoint theorem~\cite{CoC77,MOS04},
we have that the least fixedpoint occurs upon $\B_\ell = \B_{\ell+1}$,
which $\ell$ is bounded by $\dim(\h)-1$ since $\B_i$ are nonnull subspaces of $\h$.
\end{proof}

The procedure of computing the I-reachable space $\Phi(\Pro,\rho)$ is stated in Algorithm~\ref{Algo:RS}
with complexity analysis below.
\begin{algorithm}[ht]
	\caption{Computing the I-Reachable Space~{\cite[Algorithm~1]{LYY14}}}\label{Algo:RS}
	\begin{algorithmic}
		\Require a nondeterministic quantum program $\Pro=(\Sigma,\E,\{\M_\yes,\M_\no\})$
        with $\Sigma=\{\alpha_j:j=1,2,\dots,m\}$ and $\E(\alpha_j)=\E_j$
        over $\h$ with dimension $d$,
        and an input state $\rho\in\den$;
		\Ensure an orthonormal basis $B$ of $\Phi(\Pro,\rho)$.
		\begin{algorithmic}[1]
            \State let $\bar{\F}=\tfrac{1}{m}\sum_{j=1}^m \E_j \circ \{\M_\yes\}$
            be the average super-operator;
            \State let $\{\FF_j: j=1,2,\dots,l\}$ be a Kraus representation of $\bar{\F}$;
            \State compute an orthonormal basis $B_0$ of $\supp(\rho)$,
            and $B_{-1} \gets \emptyset$;
			\For{$i\gets 1$ to $d-1$}
            \State $B_i \gets B_{i-1}$;
            \ForAll{$\ket{\psi} \in B_{i-1} \setminus B_{i-2}$}\label{ln:inner}
			\State $V \gets \{\FF_j\ket{\psi}: j=1,2,\dots,l \}$;
			\State compute an orthonormal basis $B'$ of $V$ complement to $B_i$;
            \State $B_i \gets B_i \cup B'$;
            \EndFor
            \If{$B_i=B_{i-1}$ or $|B_i|=d$} \textbf{break};
            \EndIf
			\EndFor
			\State \Return $B_i$.
		\end{algorithmic}
	\end{algorithmic}
\end{algorithm}
\paragraph{Complexity}
The Kraus representation of $\E_j$ are known as the input information of $\Pro$.
For convenience,
we do not compute the simplest Kraus representation of $\bar{\F}$
whose number of Kraus operators can be bounded by $d^2$ here,
but just use the average Kraus operators of $\E_j$,
since the simplest Kraus representation obtained
by quantum process tomography~\cite[Subsection~8.4.2]{NC00}
costs additionally $\BigO(d^{12})$ operations.
Note that there are less than $d$ times of entering the inner loop in Line~\ref{ln:inner}.
Each inner loop performs $l$ times of matrix-vector multiplication
and $l$ times of computing orthocomplement,
where $l$ is bounded by $m \cdot d^2$,
as the factor $m$ comes from the number of actions in $\Pro$
and the factor $d^2$ comes from the number of Kraus operators of the super-operators $\E_j$.
The matrix-vector multiplication $\FF_j\ket{\psi}$ is in $\BigO(d^2)$,
and computing orthocomplement of $\FF_j\ket{\psi}$ is also in $\BigO(d^2)$
by normalizing the difference
\[
    \FF_j\ket{\psi} - \sum_{\ket{\varphi} \in B_i} \bra{\varphi}\FF_j\ket{\psi} \ket{\varphi}
\]
if it is nonzero as the standard Gram--Schmit procedure.
Hence Algorithm~\ref{Algo:RS} is in time $\BigO(m \cdot d^5)$. \qed

\begin{example}\label{ex2-1}
    Continue to consider the nondeterministic quantum program $\Pro_1$ in Example~\ref{ex1-3},
    the average super-operator is $\bar{\F} = \tfrac{1}{2}(\F_{\alpha_1}+\F_{\alpha_2})$.
    Since $\tfrac{1}{2}\F_{\alpha_i}(\rho)=\tfrac{1}{2} (\EE_i \M_\yes) \rho (\EE_i \M_\yes)^\dag$
    for $i\in\{1,2\}$,
    the Kraus representation of $\bar{\F}$ can be $\{\FF_1,\FF_2\}$,
    where
    \[
        \begin{aligned}
        \FF_1 &= \tfrac{1}{\sqrt{2}} \EE_1 \M_\yes
        = \tfrac{1}{\sqrt{2}}(\op{+,1}{0,0}+\op{-,1}{1,0}+\op{-,0}{1,1}), \\
        \FF_2 &= \tfrac{1}{\sqrt{2}} \EE_2 \M_\yes
        = \tfrac{1}{\sqrt{2}}(\op{1,+}{0,0}+\op{0,+}{1,0}+\op{0,-}{1,1}).
        \end{aligned}
    \]
    By Algorithm~\ref{Algo:RS},
    for the given input state $\rho_0 = \op{q_1,q_2}{q_1,q_2} = \op{1,1}{1,1}$,
    the I-reachable space can be inductively computed as follows.
    \begin{enumerate}
        \item Initially, we have $\B_0 = \supp(\rho_0)=\spn(\{\ket{1,1}\})$.
        \item To get the next subspace $\B_1$ along the ascending chain,
        for the basis element $\ket{1,1}$ in $\B_0$,
        we compute
        \[
            \begin{aligned}
            \FF_1 \ket{1,1} &= \tfrac{1}{\sqrt{2}}\ket{-,0}, \\
            \FF_2 \ket{1,1} &= \tfrac{1}{\sqrt{2}}\ket{0,-}.
            \end{aligned}
        \]
        The former operator $\FF_1 \ket{1,1}$ is already orthogonal to $\B_0$
        and can be normalized to $\ket{-,0}$;
        the latter operator $\FF_2 \ket{1,1}$ is also orthogonal to $\B_0$
        but gives another orthogonal element $(\ket{+,0}-\sqrt{2}\ket{0,1})/\sqrt{3}$ by normalizing
        $\tfrac{1}{\sqrt{2}}\ket{0,-} - \tfrac{1}{\sqrt{2}}\ip{-,0}{0,-} \ket{-,0}$.
        Thus the orthonormal basis complement to $\B_0$
        is $\{\ket{-,0}, (\ket{+,0} -\sqrt{2}\ket{0,1})/\sqrt{3}\}$,
        and we get
        $\B_1 = \spn(\{\ket{1,1},\ket{-,0},(\ket{+,0}-\sqrt{2}\ket{0,1})/\sqrt{3}\})$.
        \item To get the subspace $\B_2$,
        for the newly-produced basis elements $\ket{-,0}$
        and $(\ket{+,0}-\sqrt{2}\ket{0,1})/\sqrt{3}$ in $\B_1$,
        we have
        \[
        \begin{aligned}
            \FF_1 \ket{-,0} &= \tfrac{1}{\sqrt{2}}\ket{1,1}, \\
            \FF_2 \ket{-,0} &= -\tfrac{1}{2}\ket{-,+}, \\
            \FF_1 \tfrac{1}{\sqrt{3}}(\ket{+,0} - \sqrt{2}\ket{0,1}) &= \tfrac{1}{\sqrt{6}}\ket{0,1}, \\
            \FF_2 \tfrac{1}{\sqrt{3}}(\ket{+,0} - \sqrt{2}\ket{0,1}) &= \tfrac{1}{\sqrt{6}}\ket{+,+}.
        \end{aligned}
        \]
        Thus an orthonormal basis complement to $\B_1$
        is $\{(-\sqrt{2}\ket{+,0} - \ket{0,1})/\sqrt{3}\}$,
        and we get
        $\B_2 = \spn(\{\ket{1,1}, \ket{-,0},\linebreak[0] (\ket{+,0}-\sqrt{2}\ket{0,1})/\sqrt{3},
        (-\sqrt{2}\ket{+,0}-\ket{0,1})/\sqrt{3}\})$.
        Since $\dim(\h_{\mathbf{VAR}})=d=4=\dim(\B_2)$, we have $\B_2 = \h_{\mathbf{VAR}}$.
    \end{enumerate}
    Hence the least fixedpoint of the ascending chain occurs,
    which yields the I-reachable space $\Phi(\Pro_1,\rho_0) = \h_{\mathbf{VAR}}$. \qed
\end{example}

In the following,
we will have a deeper observation of the reachable set and the reachable space.
Since the former is a countable set and the latter is a continuum,
the latter is possibly a much larger superset of the former.
We are to narrow the over-approximation of the reachable set using other algebraic structures,
instead of the I-reachable space.
One promising way is to use the linearly independent basis of Hermitian operators on $\h$,
say
\begin{equation}\label{eq:base}
    \begin{aligned}
	\{\op{i}{i}:\,1 \le i \le d\}
    &\cup \{(\op{i}{j}+\op{j}{i})/\sqrt{2}:\,1 \le i<j \le d\} \\
	&\cup \{(\imath\op{i}{j}-\imath\op{j}{i})/\sqrt{2}:\,1 \le i<j \le d\}.
    \end{aligned}
\end{equation}
(When the Hilbert space $\h$ in consideration is exactly on the $k$-qubit system,
i.\,e.\@ $d=2^k$ for some integer $k$,
an alternative way is to use the $4^k$ linearly independent elements
$\bigotimes_{i=1}^k \gamma_i$,
where $\gamma_i$ is one of Pauli operators $\id$, $\mathbf{X}=\op{0}{1}+\op{1}{0}$,
$\mathbf{Y}=\imath\op{0}{1}-\imath\op{0}{1}$ and $\mathbf{Z}=\op{0}{0}-\op{1}{1}$.)
Although the general state is expressed by all $d^2$ basis elements in~\eqref{eq:base},
all reachable states might be expressed by only a part of these basis elements.
So, using as few as possible basis elements to express all pure reachable states
yields a more precise notion of reachable space.
In the setting of reachability analysis,
at most $d^2$ pure reachable states could be served as
the linearly independent basis of $\her(\h)$ that we require.
To this end, we resort to the following operator-level program
that characterizes the operations between pure reachable states.

\begin{definition}[Operator-level Program]\label{def:operator}
Let $\Pro=(\Sigma,\E,\{\M_\yes,\M_\no\})$
be a nondeterministic quantum program with $\E_j=\{\EE_{j,k}: k=1,2,\dots,K_j\}$.
Then the operator-level program $\hat{\Pro}$ of $\Pro$ is
the triple $(\hat{\Sigma},\EE,\{\M_\yes,\M_\no\})$,
where
    \begin{itemize}
        \item $\hat{\Sigma}=\{\alpha_{j,k}:j=1,2,\dots,m \wedge k=1,2,\dots,K_j\}$
        is a finite set of actions;
        \item $\EE: \hat{\Sigma} \to \mathcal{L}$ gives rise to the linear operators $\EE_{j,k}$
        taken action $\alpha_{j,k}$,
        which are obtained from the Kraus representation $\bigcup_{j=1}^m \{\EE_{j,k}: k=1,2,\dots,K_j\}$
        of $\sum_{j=1}^m \E_j$.
    \end{itemize}
\end{definition}
Rigorously speaking,
the operator-level program $\hat{\Pro}$ is not a nondeterministic quantum program described in Definition~\ref{I-QMDP},
since it does not meet the trace-preserving restriction generally,
i.\,e.\@, $\{\EE_{j,k}\} \eqsim \I$ holds for all actions $\alpha_{j,k} \in \hat{\Sigma}$
where $\{\EE_{j,k}\}$ denotes the super-operator that has the unique Kraus operator $\EE_{j,k}$.
However, dropping this restriction does not affect the qualitative termination $\Prob(\rho)=1$
considered in the paper,
and we would study the qualitative termination of the operator-level program afterwards.
For convenience, the notation $\FF_\varsigma$ is adapted to $\F_\varsigma$,
e.\,g.\@ $\FF_{\alpha_{j,k}}=\EE_{j,k}\M_\yes$
and $\FF_\varsigma=\FF_{\varsigma \downarrow 1} \FF_{\varsigma \uparrow 1}$.

\begin{definition}[II-Reachable Space]\label{FRS}
	Given a nondeterministic quantum program $\Pro$
    and an input pure state $\rho=\op{\lambda}{\lambda}\in\den$,
	the type II reachable space of $\Pro$ starting from $\rho$ is
    $\Upsilon(\Pro,\rho)=\spn(\Psi(\hat{\Pro},\rho))$,
    where $\hat{\Pro}$ is the operator-level program of $\Pro$ as in Definition~\ref{def:operator}.
\end{definition}

It is not hard to see that
the reachable set $\Psi(\Pro,\rho)$ is over-approximated by the II-reachable space $\Upsilon(\Pro,\rho)$,
since i) all elements $\gamma \in \Psi(\Pro,\rho)$ can be linearly expressed
by those elements in $\Psi(\hat{\Pro},\rho)$
and ii) $\Upsilon(\Pro,\rho)=\spn(\Psi(\hat{\Pro},\rho))$.

For an input pure state $\rho=\op{\lambda}{\lambda}$,
we compute the II-reachable space as the least fixedpoint of the ascending chain of linear subspaces of $\her(\h)$: 
\begin{equation}\label{eq:chain1}
\begin{aligned}
\spn(\{\{\FF_\varsigma\}(\rho): \varsigma \in \hat{\Sigma}^* \wedge |\varsigma|=0\})
& \subseteq \spn(\{\{\FF_\varsigma\}(\rho): \varsigma \in \hat{\Sigma}^* \wedge |\varsigma|\le 1\}) \\
& \subseteq \spn(\{\{\FF_\varsigma\}(\rho): \varsigma \in \hat{\Sigma}^* \wedge |\varsigma|\le 2\}) \\
& \subseteq \cdots,
\end{aligned}
\end{equation}
where the notation $\{\FF_\varsigma\}$ in bracket denotes a super-operator.
The following lemma gives an upper bound for the occurrence of the least fixedpoint in the ascending chain.

\begin{lemma}\label{FRSC}
	Let $\Theta_0 \subseteq \Theta_1 \subseteq \Theta_2 \subseteq \cdots$
    be the ascending chain of nonnull linear subspaces $\Theta_i \subseteq \her(\h)$,
    as defined in~\eqref{eq:chain1}.
	Then there is an integer $\ell \le \dim(\h)^2-1$
    such that $\Theta_k=\Theta_\ell$ holds for all $k>\ell$. 
\end{lemma}
\begin{proof}
    The proof is similar to that of Lemma~\ref{lem:lfp}.
    The function $G$ from $\Theta_i$ to $\Theta_{i+1}$ ($i \ge 0$) can be formulated as
    a monotonic function
    \[
        G(\mathbb{Y})=
        \spn(\mathbb{Y} \cup \{\{\FF_\alpha\}(\gamma):\gamma\in \mathbb{Y} \wedge \alpha\in\hat{\Sigma}\}).
    \]
    Meanwhile, all subspaces $\Theta^\textup{all}$ of $\her(\h)$ form
    a complete lattice $(\Theta^\textup{all},\subseteq,\inf,\sup)$
    by taking `$\inf$' as the meet $\bigwedge=\bigcap$ and `$\sup$' as the join $\bigvee$.
    By Knaster--Tarski fixedpoint theorem~\cite{CoC77,MOS04},
    we have that the least fixedpoint occurs upon $\Theta_\ell = \Theta_{\ell+1}$,
    where $\ell$ is bounded by $\dim(\h)^2-1$
    since $\Theta_i$ are nonnull subspaces of $\her(\h)$.
\end{proof}

The procedure of computing the II-reachable space $\Upsilon(\Pro,\rho_0)$
is stated in Algorithm~\ref{Algo:FRS} with complexity analysis below.
\begin{algorithm}[ht]
	\caption{Computing the II-Reachable Space}\label{Algo:FRS}
	\begin{algorithmic}
		\Require a nondeterministic quantum program $\Pro=(\Sigma,\E,\{\M_\yes,\M_\no\})$
        with $\Sigma=\{\alpha_j:j=1,2,\dots,m\}$, $\E(\alpha_j)=\E_j$
        and $\E_j=\{\EE_{j,k}:k=1,2,\dots,K_j\}$
        over $\h$ with dimension $d$,
        and an input pure state $\rho_0=\op{\lambda}{\lambda}\in\den$;
		\Ensure a linearly independent basis $\theta$ of $\Upsilon(\Pro,\rho_0)$
        whose elements are pure states.
		\begin{algorithmic}[1]
            \State let $\hat{\Sigma}=\{\alpha_{j,k}: j=1,2,\dots,m \wedge k=1,2,\dots,K_j\}$,
            and $\EE(\alpha_{j,k})=\EE_{j,k}$;
            \State let $\hat{\Pro}=(\hat{\Sigma},\EE,\{\M_\yes,\M_\no\})$
            be the operator-level program of $\Pro$;
            \State $\FF_{\alpha_{j,k}} \gets \EE_{j,k} \M_\yes$ with $j=1,2,\dots,m$ and $k=1,2,\dots,K_j$;
            \State $B_0 \gets \{\ket{\lambda}\}$, $B_{-1} \gets \emptyset$,
            and $\theta_0 \gets \{\rho_0\}$;
            \For{$i\gets 1$ to $d^2-1$}\label{ln:it}
            \State $B_i \gets B_{i-1}$ and $\theta_i \gets \theta_{i-1}$;
            \ForAll{$\ket{\psi} \in B_{i-1} \setminus B_{i-2}$}\label{ln:inner1}
			\State $V \gets \{\FF_{\alpha_{j,k}}\ket{\psi}/\|\FF_{\alpha_{j,k}}\ket{\psi}\|:
            j=1,2,\dots,m \wedge k=1,2,\dots,K_j\}$;
			\State find a maximal subset $B'$ of $V$,
            such that $\theta'=\{\op{\psi'}{\psi'}: \ket{\psi'}\in B'\}$ is a linearly independent basis
            complement to $\theta_i$;
			\State $B_i \gets B_i \cup B'$ and $\theta_i \gets \theta_i \cup \theta'$;
            \EndFor
            \If{$B_i=B_{i-1}$ or $|B_i|=d^2$} \textbf{break};
            \EndIf
			\EndFor
			\State \Return $\theta_i$.
		\end{algorithmic}
	\end{algorithmic}
\end{algorithm}
\paragraph{Complexity}
Note that there are less than $d^2$ times of entering the inner loop in Line~\ref{ln:inner1}.
Each inner loop performs at most $m \cdot d^2$ times of matrix-vector multiplication
together with normalization
and at most $m \cdot d^2$ times of checking the linear independence,
as the factor $m$ comes from the number of actions in $\Pro$
and the factor $d^2$ comes from the number of Kraus operators of $\E_j$.
The matrix-vector multiplication is in $\BigO(d^2)$,
the normalization is in $\BigO(d)$,
and checking the linear independence can be in $\BigO(d^4)$
with embedding into the orthonormalization of the linearly independent basis.
That is, $\theta_i$ is a linearly independent basis
if and only if there is an orthonormal basis $\vartheta_i$ such that $\spn(\theta_i)=\spn(\vartheta_i)$,
in which each element can be obtained in $\BigO(d^4)$ by the Gram--Schmit procedure.
Hence Algorithm~\ref{Algo:FRS} is in time $\BigO(m \cdot d^8)$. \qed

\begin{example}\label{ex2-2}
    Reconsider the program $\Pro_2$ in Example~\ref{ex1-3},
    the operator-level program $\hat{\Pro}=(\hat{\Sigma},\EE,\{\M_\yes,\M_\no\})$ of $\Pro_2$ provides
    \begin{itemize}
        \item the set of actions $\hat{\Sigma}=\{\alpha_{1,1},\alpha_{2,1}\}$;
        \item linear operators $\EE(\alpha_{1,1})=\EE_{1,1} = \mathbf{H} \otimes \mathbf{X}$
        and $\EE(\alpha_{2,1})=\EE_{2,1}=\mathbf{X} \otimes \mathbf{H}$.
    \end{itemize}
    We define $\FF_{\alpha_{1,1}}=\EE_{1,1} \M_\yes$ and $\FF_{\alpha_{2,1}}=\EE_{2,1} \M_\yes$.
    By Algorithm~\ref{Algo:FRS},
    for the input pure state $\rho = \op{1,1}{1,1}$,
    the II-reachable space can be computed as follows.
    \begin{enumerate}
        \item Initially, we have $B_0 = \{\ket{1,1}\}$ and $\theta_0 = \{\op{1,1}{1,1}\}$.
        \item Then, we compute
        \[
            \begin{aligned}
            \FF_{\alpha_{1,1}} \ket{1,1}/\|\FF_{\alpha_{1,1}} \ket{1,1}\| &=\ket{-,0}, \\
            \FF_{\alpha_{2,1}} \ket{1,1}/\|\FF_{\alpha_{2,1}} \ket{1,1}\| &=\ket{0,-}.
            \end{aligned}
        \]
        So we have $V = \{\ket{-,0},\ket{0,-}\}$.
        Since the two pure states in $V$ have density operators
        that form a linearly independent basis complement to $\theta_0$,
        we obtain $B_1 = B_0 \cup V = \{\ket{1,1},\ket{-,0},\ket{0,-}\}$ and
        $\theta_1 = \{\op{\psi}{\psi}:\psi\in B_1\} = \{\op{1,1}{1,1},\op{-,0}{-,0},\op{0,-}{0,-}\}$.
        \item Repeating this process, we have
            \[
                \begin{aligned}
                B_2 &= \{\ket{1,1},\ket{-,0},\ket{0,-},\ket{-,+},\ket{+,1},\ket{1,+}\}, \\
                B_3 &= B_2 \cup \{(\ket{-,0}-\sqrt{2}\ket{1,1})/\sqrt{3},(\sqrt{2}\ket{0,0}-\ket{1,+})/\sqrt{3}\}, \\
                B_4 &= B_3.
                \end{aligned}
            \]
    In detail, we name the eight elements in $B_3$ by
    $\ket{\psi_1}=\ket{1,1}$,
    $\ket{\psi_2}=\ket{-,0}$,
    $\ket{\psi_3}=\ket{0,-}$,
    $\ket{\psi_4}=\ket{-,+}$,
    $\ket{\psi_5}=\ket{+,1}$,
    $\ket{\psi_6}=\ket{1,+}$,
    $\ket{\psi_7}=(\ket{-,0}-\sqrt{2}\ket{1,1})/\sqrt{3}$
    and $\ket{\psi_8}=(\sqrt{2}\ket{0,0}-\ket{1,+})/\sqrt{3}$,
    whose outer product form $\op{\psi_i}{\psi_i}$ are reachable states of the operator-level program $\hat{\Pro}$.
    The eight outer products $\op{\psi_i}{\psi_i}$ make up the set $\theta_3$,
    which is sufficient to linearly express any reachable pure state of $\hat{\Pro}$.
    For instance, $\FF_{\alpha_{2,1}} \ket{1,+}=\ket{0,0}$
    and its outer product form $\op{0,0}{0,0}$ is a reachable state of $\hat{\Pro}$,
    which can be linearly expressed as
    \[
        \op{0,0}{0,0}=
        \op{\psi_1}{\psi_1}-\op{\psi_2}{\psi_2}+\op{\psi_6}{\psi_6}-3\op{\psi_7}{\psi_7}+3\op{\psi_8}{\psi_8}.
    \]
    So we do not necessarily put $\ket{0,0}$ into $B_3$,
    nor necessarily put $\op{0,0}{0,0}$ into $\theta_3$,
    since $\op{0,0}{0,0}$ is in $\spn(\theta_3)$.
    Overall, the closure of $B_3$ under the operators $\FF_{\alpha_{1,1}}$ and $\FF_{\alpha_{2,1}}$
    is shown as in Table~\ref{tab:closure},
    implying that all linear combinations of the eight outer products $\op{\psi_i}{\psi_i}$
    under the operators $\FF_{\alpha_{1,1}}$ and $\FF_{\alpha_{2,1}}$
    are also in $\spn(\theta_3)$.
    \end{enumerate}
    Thus the least fixedpoint of the ascending chain occurs,
    which yields the II-reachable space
    $\Upsilon(\Pro_2,\rho_0) = \spn(\{\op{\psi}{\psi}: \ket{\psi} \in B_4\})$.
    
    \begin{table}[htp]
    \renewcommand\arraystretch{1.5}
    \centering
    \caption{The closure of $B_3$ under the operators $\FF_{\alpha_{1,1}}$ and $\FF_{\alpha_{2,1}}$}\label{tab:closure}
    \begin{tabular}{l|l}
        \cline{1-2}
        $\FF_{\alpha_{1,1}} \ket{\psi_1} =\ket{-,0} =\ket{\psi_2}$
        & $\FF_{\alpha_{2,1}} \ket{\psi_1} =\ket{0,-} =\ket{\psi_3}$ \\
        \cline{1-2}
                
        $\FF_{\alpha_{1,1}} \ket{\psi_2} =\ket{1,1} =\ket{\psi_1}$
        & $\FF_{\alpha_{2,1}} \ket{\psi_2} =-\ket{-,+} =-\ket{\psi_4}$ \\
        \cline{1-2}
                
        $\FF_{\alpha_{1,1}} \ket{\psi_3}/\|\FF_{\alpha_{1,1}} \ket{\psi_3}\| =\ket{+,1} =\ket{\psi_5}$
        & $\FF_{\alpha_{2,1}} \ket{\psi_3}/\|\FF_{\alpha_{2,1}} \ket{\psi_3}\| =\ket{1,+} =\ket{\psi_6}$ \\
        \cline{1-2}
                
        $\FF_{\alpha_{1,1}} \ket{\psi_4}/\|\FF_{\alpha_{1,1} \ket{\psi_4}\|}
        =\frac{\sqrt{2}\ket{1,1}-\ket{-,0}}{\sqrt{3}} =-\ket{\psi_7}$
        & $\FF_{\alpha_{2,1}} \ket{\psi_4}/\|\FF_{\alpha_{2,1}} \ket{\psi_4}\|
        =\frac{\ket{1,+}-\sqrt{2}\ket{0,0}}{\sqrt{3}} =-\ket{\psi_8}$ \\
        \cline{1-2}
                
        $\FF_{\alpha_{1,1}} \ket{\psi_5}/\|\FF_{\alpha_{1,1}} \ket{\psi_5}\|
        =\ket{-,0} =\ket{\psi_2}$
        & $\FF_{\alpha_{2,1}} \ket{\psi_5}/\|\FF_{\alpha_{2,1}} \ket{\psi_5}\|
        =\ket{0,-} =\ket{\psi_3}$ \\
        \cline{1-2}
                
        $\FF_{\alpha_{1,1}} \ket{\psi_6} =\ket{-,+} =\ket{\psi_4}$
        & $\FF_{\alpha_{2,1}} \ket{\psi_6} =\ket{0,0}$ \\
        \cline{1-2}
                
        $\FF_{\alpha_{1,1}} \ket{\psi_7}/\|\FF_{\alpha_{1,1}} \ket{\psi_7}\|
        =\frac{\ket{1,1}-\sqrt{2}\ket{-,0}}{\sqrt{3}} =\ket{\varphi_{7,1}}$
        & $\op{\varphi_{7,1}}{\varphi_{7,1}}
        =-\tfrac{1}{3}\op{\psi_1}{\psi_1}+\tfrac{1}{3}\op{\psi_2}{\psi_2}+\op{\psi_7}{\psi_7}$ \\
        \cline{1-2}
                
        $\FF_{\alpha_{2,1}} \ket{\psi_7}/\|\FF_{\alpha_{2,1}} \ket{\psi_7}\|
        = \frac{-\ket{-,+} - \sqrt{2}\ket{0,-}}{\sqrt{3}} =\ket{\varphi_{7,2}}$
        & $\op{\varphi_{7,2}}{\varphi_{7,2}}
        =\tfrac{1}{3}\op{\psi_3}{\psi_3}-\tfrac{1}{3}\op{\psi_4}{\psi_4}+\op{\psi_8}{\psi_8}$ \\
        \cline{1-2}
                
        $\FF_{\alpha_{1,1}} \ket{\psi_8}/\|\FF_{\alpha_{1,1}} \ket{\psi_8}\|
        =\frac{\sqrt{2}\ket{+,1}-\ket{-,+}}{\sqrt{3}} =\ket{\varphi_{8,1}}$
        & $\op{\varphi_{8,1}}{\varphi_{8,1}}
        =-\tfrac{1}{3}\op{\psi_4}{\psi_4}+\tfrac{1}{3}\op{\psi_5}{\psi_5}+\op{\psi_7}{\psi_7}$ \\
        \cline{1-2}
                
        $\FF_{\alpha_{2,1}} \ket{\psi_8}/\|\FF_{\alpha_{2,1}} \ket{\psi_8}\|
        =\frac{\sqrt{2}\ket{1,+}-\ket{0,0}}{\sqrt{3}} =\ket{\varphi_{8,2}}$
        & $\op{\varphi_{8,2}}{\varphi_{8,2}}
        =-\tfrac{1}{3}\op{\psi_1}{\psi_1}+\tfrac{1}{3}\op{\psi_2}{\psi_2}+\op{\psi_7}{\psi_7}$\\
        \cline{1-2}
        \end{tabular}
    \end{table}

    It is not hard to see that $\Phi(\Pro_2,\rho_0)$ contains all pure states in $\h_{\mathbf{VAR}}$
    while $\Upsilon(\Pro_2,\rho_0)$ has dimension $8$ that is less than $\dim(\her(\h_{\mathbf{VAR}}))=16$.
    Hence there are many pure states in $\Phi(\Pro_2,\rho_0)$
    whose density operators are not in $\Upsilon(\Pro_2,\rho_0)$,
    e.\,g.\@,
    the pure state $\op{\varphi}{\varphi}$ with $\ket{\varphi}=\tfrac{1}{2}(\ket{0,0}+\ket{0,1}+\ket{1,0}+\ket{1,1})$
    in $\Phi(\Pro_2,\rho_0)$ cannot be linearly expressed by the basis of $\Upsilon(\Pro_2,\rho_0)$.
    The II-reachable space $\Upsilon(\Pro_2,\rho_0)$ gives
    an over-approximation of $\Psi(\Pro_2,\rho_0)$ more precise than $\Phi(\Pro_2,\rho_0)$ in this example. \qed 
\end{example}

\begin{remark}
    The ascending chain $\Theta_0 \subseteq \Theta_1 \subseteq \Theta_2 \subseteq \cdots$
    as defined in~\eqref{eq:chain1}
    is finer than the ascending chain
    $\B_0 \subseteq \B_1 \subseteq \B_2 \subseteq \cdots$ as defined in~\eqref{eq:chain}
    in such a sense:
    \begin{itemize}
    \item For each linear subspace $\Theta_i \subseteq \her(\h)$,
    there is a unique index $j$ such that
    $\Theta_i \subseteq \her(\B_j)$ and $\Theta_i \not\subseteq \her(\B_{j-1})$.
    \item For each linear subspace $\B_j \subseteq \h$,
    there is one index $i$ or more such that
    $\Theta_i \subseteq \her(\B_j)$ and $\Theta_i \not\subseteq \her(\B_{j-1})$.
    \item By the construction in Algorithm~\ref{Algo:FRS}
    that the basis elements in $\Theta_i$ are pure states,
    all ensembles of elements in $\Theta_i$ are elements of $\den(\B_j)$.
    \end{itemize}
    In a nutshell, each increment in $\B_j$ corresponds to one or more increments in $\Theta_i$. \qed
\end{remark}

By Algorithms~\ref{Algo:RS} and~\ref{Algo:FRS}, we obtain the result:
\begin{theorem}
    Both I-reachable space and II-reachable space are computable in polynomial time.
\end{theorem}

\section{Computing the Divergent Set}\label{S5}
In this section, we show how the set of \emph{divergent} states can be computed
from which a given nondeterministic quantum program terminates with probability zero under some scheduler,
and synthesize the corresponding divergence schedulers.
The procedure turns out to be in exponential time,
which as far as we know is reported for the first time.

\begin{definition}\label{PD}
	Given a nondeterministic quantum program $\Pro$ with the quantum state space $\h$,
	\begin{itemize}
		\item the set $D(\Pro)$ of divergent states is
		$\{\rho\in\den(\h): \lim_{i \to \infty}
		\tr(\M_\yes\F_{\sigma\uparrow i}(\rho))=1 \wedge \sigma\in\Sigma^\omega\}$;
		\item the set $PD(\Pro)$ of pure divergent states is
		$\{\ket{\psi}\in \h:\lim_{i \to \infty}
		\tr(\M_\yes\F_{\sigma\uparrow i}(\op{\psi}{\psi}))=1 \wedge \sigma\in\Sigma^\omega\}$.
	\end{itemize}
	The parameters $\Pro$ in $D(\Pro)$ and $PD(\Pro)$ can be omitted
	if they are clear from the context.
\end{definition}

The divergence requires that under some infinite scheduler $\sigma$,
all eigenstates $\ket{\lambda}$ of $\rho$ are terminating with probability zero,
i.\,e.\@ $\bigwedge_{i=0}^\infty \tr(\M_\no\F_{\sigma\uparrow i}(\op{\lambda}{\lambda}))=0$.
It is not hard to see that
an element $\rho$ in the divergent set $D$ is a probabilistic ensemble $\{(\ket{\psi_k},p_k):k=1,2,\ldots \}$
of some elements $\ket{\psi_k}$ in the pure divergent set $PD$,
and conversely an element $\ket{\psi}$ in $PD$ is
a pure state $\op{\psi}{\psi}$ in $D$, i.\,e.\@ $PD=\{\ket{\psi}\in\h: \op{\psi}{\psi} \in D(\Pro)\}$.
Once the pure divergent set $PD$ is determined,
the divergent set $D$ is also determined.
We will focus on how to compute the pure divergent set $PD$ afterwards.

For convenience, we introduce some auxiliary notions and notations:
\begin{itemize}
	\item $PD^\sigma$ denotes the set of all pure divergent states $\ket{\psi}$
	under the infinite scheduler $\sigma$,
	i.\,e.\@ $$PD^\sigma=\left\{\ket{\psi}\in\h:
	\lim_{i\to \infty}\tr(\M_\yes\F_{\sigma\uparrow i}(\op{\psi}{\psi}))=1\right\};$$
	\item $PD_i^\sigma$ denotes the set of all pure states $\ket{\psi}$ that are terminating with probability zero
	under the $i$-fragment of the infinite scheduler $\sigma$,
	i.\,e.\@ $$PD_i^\sigma=PD^{\sigma\uparrow i}=\{\ket{\psi}\in\h:
	\tr(\M_\yes\F_{\sigma\uparrow i}(\op{\psi}{\psi}))=1\};$$
	\item $PD_i$ denotes the set of all pure states $\ket{\psi}$ that are terminating with probability zero
	under the $i$-fragment of some infinite scheduler $\sigma$,
	i.\,e.\@ $PD_i =\bigcup_{\sigma\in\Sigma^\omega} PD_i^\sigma
	=\bigcup_{\varsigma \in \Sigma^i} PD^\varsigma$.
\end{itemize}
It is not hard to see:
\begin{itemize}
	\item for any infinite scheduler $\sigma$ and any integer $i$,
	$PD_i^\sigma \supseteq PD_{i+1}^\sigma$,
	as the latter requires to be terminating with probability zero for one more step,
	i.\,e.\@ $\tr(\M_\no\F_{\sigma\uparrow (i+1)}(\op{\psi}{\psi}))=0$;
	\item for any infinite scheduler $\sigma$,
	$PD^\sigma =\bigcap_{i=0}^\infty PD_i^\sigma =\lim_{i \to \infty} PD_i^\sigma$;
	\item for any integer $i$,
	$PD_i = \bigcup_{\sigma \in \Sigma^\omega} PD_i^\sigma$ amounts to a finite union of $PD_i^\sigma$,
	as there are only finitely many distinct $i$-fragments $\varsigma\in\Sigma^i$
	of all infinite schedulers $\sigma\in\Sigma^\omega$;
	\item $PD =\bigcap_{i=0}^\infty PD_i =\lim_{i \to \infty} PD_i$.
\end{itemize}

From the definition $PD_i=\bigcup_{\varsigma \in \Sigma^i} PD^\varsigma$,
the derivation of those pure divergent sets $PD_i$ can be described by an infinite $m$-branching tree,
named \emph{derivation tree} (see Fig.~\ref{fig:tree}).
Its nodes are all pure divergent sets $PD^\varsigma$,
organized by the prefix relationship on strings $\varsigma \in \Sigma^*$.
Particularly,
\begin{itemize}
\item the root of the tree is $PD^\epsilon=PD_0$
	that is the pure divergent set under the empty scheduler $\epsilon$;
\item each intermediate node $PD^\varsigma$
	has $m$ children $PD^{\varsigma\cdot\alpha}$
	that are the sets of pure divergent states derived by one more step $\alpha\in \Sigma$.
\end{itemize}
Thus, the union of $PD^\varsigma$ in the $i$th layer is actually $PD_i$.
By the nice property $PD^{\sigma \uparrow i} \supseteq PD^{\sigma \uparrow (i+1)}$,
we have that an intermediate node $PD^\varsigma$ is
a common superset of its $m$ children $PD^{\varsigma \cdot \alpha}$ with $\alpha$ ranging over $\Sigma$.
The derivation tree is said \emph{stabilized} at the $\ell$th layers
if $PD_k = PD_\ell$ holds for all $k>\ell$.
So the pure divergent set $PD$ could be collected at that layer.

\begin{figure}[htp]
	\centering
	\begin{tikzpicture}[grow = down,
	level 1/.style = {sibling distance = 26mm},
	level distance = 20mm,
	edge from parent/.style = {draw, -latex},
	child anchor=north,
	edge from parent macro=\myedgefromparent,
	sloped,
	every node/.style = {font=\small}]
	\def\myedgefromparent#1#2{
		[style=edge from parent,#1]
		(\tikzparentnode\tikzparentanchor) to #2 (\tikzchildnode\tikzchildanchor)
	}
	\node (a) {$PD^\epsilon$}
	child [grow=down] {node (pdone) [left=50mm] {$PD^{\alpha_1}$}
		child [grow=down] {node (pdoo) [left = 20mm]  {$PD^{\alpha_1\alpha_1}$}
			child[grow=down,level distance = 10mm] {node (xx)  {$\vdots$}}
			child[grow=down,level distance = 10mm]{node (aa) [base left = 6mm of xx] {$\vdots$}}			
		}
		child [grow=down] {node (pdpd) [left=10mm]  {$\cdots$} edge from parent [draw=none]}
		child [grow=down] {node (pdom) {$PD^{\alpha_1\alpha_m}$}
			child[grow=down,level distance = 10mm] {node (xy)   {$\vdots$}
			child[grow=down] {node (maxwidleft) [above=2mm] {} edge from parent [draw=none]}
			}
			child[grow=down,level distance = 10mm]{node (ab) [base left = 6mm of xy] {$\vdots$}}
		}
	}
	child[grow=down] {node (pd) [left=22mm]  {$\cdots$} edge from parent [draw=none]
		child[grow=down] {node [left=12mm] (cc) {$\cdots$} edge from parent [draw=none]
		} 
	}
	child[grow=down] {node  (pdm){$PD^{\alpha_m}$}
		child [grow=down]{node (pdmo) [left=20mm] {$PD^{\alpha_m\alpha_1}$}
			child[grow=down,level distance = 10mm] {node (xz)  {$\vdots$}}
			child[grow=down,level distance = 10mm]{node (ac) [base left = 6mm of xz] {$\vdots$}}
		}
		child [grow=down]{node [left=10mm] (cdd) {$\cdots$} edge from parent [draw=none]}
		child [grow=down] {node (pdmm) {$PD^{\alpha_m\alpha_m}$}
			child[grow=down] {node (maxwidright) [above=2mm] {} edge from parent [draw=none]}
			child[grow=down,level distance = 10mm] {node (yz)   {$\vdots$}}
			child[grow=down,level distance = 10mm]{node (bc) [base left = 6mm of yz] {$\vdots$}}
		}
		child [grow=right] {node [left=10mm] {$=$} edge from parent [draw=none]
			child [grow=up] {node {$=$} edge from parent [draw=none]}
			child [grow=down] {node {$=$} edge from parent [draw=none]}
			child [grow=right] {node[left=11mm] (PD1)  {$PD_1$} edge from parent [draw=none]
				child [grow=up] {node (PD0) {$PD_0$} edge from parent [draw=none]
					child [grow=right] {node [left=14mm]{} edge from parent [draw=none]
						child [grow=down] {node [below=40mm]{}
							edge from parent node [above]{\textbf{Stabilization within depth $d$}}
						}
					}
				}
				child [grow=down] {node (PD2) {$PD_2$} edge from parent [draw=none]
					child [grow=down, level distance=10mm] {node (cdots) {$\vdots$} edge from parent [draw=none]}
				}
			}
		}
	};
	\path (pdone) -- (pd) node [midway] {$\cup$};
	\path (pd) -- (pdm) node [midway] {$\cup$};
	\path (pdoo) -- (pdpd) node [midway] {$\cup$};
	\path (pdpd) -- (pdom) node [midway] {$\cup$};
	\path (pdom) -- (cc) node [midway] {$\cup$};
	\path (cc) -- (pdmo) node [midway] {$\cup$};
	\path (pdmo) -- (cdd) node [midway] {$\cup$};
	\path (cdd) -- (pdmm) node [midway] {$\cup$};
	\path (xx) -- (aa) node [midway] {$\cdots$};
	\path (aa) -- (xy) node [midway] {$\cdots$};
	\path (xy) -- (ab) node [midway] {$\cdots$};
	\path (ab) -- (xz) node [midway] {$\cdots$};
	\path (xz) -- (ac) node [midway] {$\cdots$};
	\path (ac) -- (yz) node [midway] {$\cdots$};
	\path (yz) -- (bc) node [midway] {$\cdots$};
	\path (PD0) -- (PD1) node [midway] {$\supseteq$};
	\path (PD1) -- (PD2) node [midway] {$\supseteq$};
	\path (PD2) -- (cdots) node [midway,right=-1mm] {$\supseteq$};
	\draw[-latex] (maxwidright) -- node[below]{\textbf{Stabilization within width $m^d$}} (maxwidleft);
	\end{tikzpicture}	
	\caption{Derivation of $PD_i$ by a tree construction}\label{fig:tree}
  \Description[<short description>]{<long description>}
\end{figure}
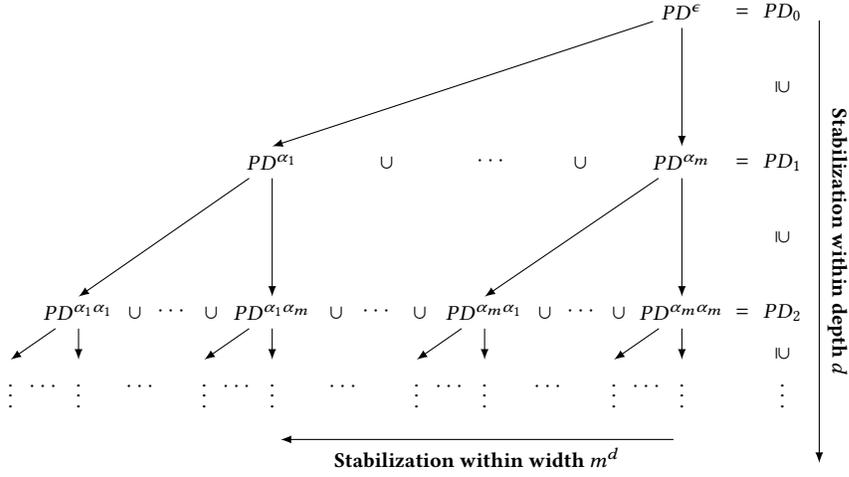

Besides the definition of $PD_i$,
there is an alternative approach to calculate $PD_i$.
Before stating it, we need to introduce the notion of \emph{pure divergent space} $PDS^\varsigma$ 
that is the closure of $PD^\varsigma$ under scalar multiplication,
i.\,e.\@ $PDS^\varsigma=\{c\ket{\psi}: \ket{\psi}\in PD^\varsigma \wedge c\in\mathbb{C}\}$.
The notions $PDS_i$ and $PDS$ are defined similarly.
For a finite scheduler $\varsigma$,
$PDS^\varsigma$ is a subspace of $\h$~\cite[Lemma~4]{LYY14}\footnote{%
To see why it is the case, we note that:
$PDS^\varsigma$ is equivalently defined as
$\{ c\ket{\psi}: \bigwedge_{i=0}^{|\varsigma|} \tr(\M_\no\F_{\varsigma\uparrow i}(\op{\psi}{\psi}))=0
\wedge \ket{\psi}\in \h \wedge c\in\mathbb{C}\}$,
where $\tr(\M_\no\F_{\varsigma\uparrow i}(\op{\psi}{\psi}))=0$ holds if and only if
$\supp(\M_\no)$ is orthogonal to $\supp(\F_{\varsigma\uparrow i}(\op{\psi}{\psi}))$.
By the inclusion~\eqref{eq:support},
if $\supp(\M_\no)$ is orthogonal to
each $\supp(\F_{\varsigma\uparrow i}(\op{\psi_k}{\psi_k}))$ ($k=1,2,\ldots,K$),
it is also orthogonal to $\supp(\F_{\varsigma\uparrow i}(\op{\psi}{\psi}))$
for any $\ket{\psi} \in \spn(\{\ket{\psi_k}:k=1,2,\ldots,K\})$,
i.\,e.\@ $\tr(\M_\no\F_{\varsigma\uparrow i}(\op{\psi}{\psi}))=0$,
yielding the linearity $\ket{\psi} \in PD^\varsigma \subset PDS^\varsigma$.},
$PD^\varsigma$ is the unit sphere of $PDS^\varsigma$,
and they can be mutually determined.
Particularly, if $PDS^\varsigma$ is the null space $\{0\}$,
there is no element $\ket{\psi} \in PD^\varsigma$,
i.\,e.\@, $PD^\varsigma$ is the empty set;
vice versa.
It entails
\[
	PD^\varsigma=\emptyset \Longleftrightarrow PDS^\varsigma=\{0\}
	\Longleftrightarrow \dim(PDS^\varsigma)=0.
\]

\begin{proposition}\label{prop:pd}
	The pure divergent sets $PD_i$ can be calculated inductively as
	\[
		PD_i= \begin{cases}
			\{\ket{\psi}\in\h: \M_\no\ket{\psi}=0\} & \textup{if }i=0, \\
			\bigcup_{\alpha \in \Sigma} \{\ket{\psi}\in PD_0: \supp(\F_\alpha(\op{\psi}{\psi})) \subseteq PDS_{i-1}\}
			& \textup{if }i>0.
		\end{cases}
	\]
\end{proposition}
\begin{proof}
If $i=0$, we have $PD_0=PD^\epsilon=\{\ket{\psi}\in\h: \M_\no\ket{\psi}=0\}$ plainly.
Otherwise, for a sphere $PD^\varsigma$ in the union $PD_{i-1}$ and an action $\alpha \in \Sigma$,
we compute:
\begin{equation}\label{eq:pd0}
	\begin{aligned}
	PD^{\alpha \cdot \varsigma}
	& = \{\ket{\psi}\in PD_0: \F_\alpha(\op{\psi}{\psi}) \in \den(PDS^\varsigma) \} \\
	& = \{\ket{\psi}\in PD_0: \supp(\F_\alpha(\op{\psi}{\psi})) \subseteq PDS^\varsigma \}.
	\end{aligned}
\end{equation}
We collect all spheres $PD^{\alpha \cdot \varsigma}$
with $\alpha$ ranging over $\Sigma$ and $\varsigma$ ranging over $\Sigma^{i-1}$ as $PD_i$,
i.\,e.\@,
\[
	\begin{aligned}
	PD_i &= \bigcup_{\alpha \in \Sigma} \bigcup_{\varsigma\in\Sigma^{i-1}}
	\{\ket{\psi}\in PD_0: \supp(\F_\alpha(\op{\psi}{\psi})) \subseteq PDS^\varsigma \} \\
	&= \bigcup_{\alpha \in \Sigma}
	\{\ket{\psi}\in PD_0: \supp(\F_\alpha(\op{\psi}{\psi})) \subseteq PDS_{i-1}\},
	\end{aligned}
\]
where the second equality comes from the fact that
the linear subspace $\supp(\F_\alpha(\op{\psi}{\psi}))$ is covered by $PDS_{i-1}$ if and only if
it is covered by some subspace $PDS^\varsigma$ in the union $PDS_{i-1}$.
\end{proof}

As an immediate corollary,
from any subtree rooted at $PD^\varsigma$ with $\varsigma \in \Sigma^*$,
we can get
\[
	PD^\varsigma \cap PD_{|\varsigma|+i+1} = \bigcup_{\alpha \in \Sigma}
	\{\ket{\psi}\in PD^\varsigma:
	\supp(\F_\alpha(\op{\psi}{\psi})) \subseteq (PDS^\varsigma \cap PDS_{|\varsigma|+i})\},
\]
where $PD^\varsigma \cap PD_{|\varsigma|+i}$ denotes
the union of $PD^{\varsigma \cdot \varsigma'}$ over $\varsigma'\in\Sigma^i$.

By the alternative approach,
the set $PD_i$ is calculated from the prior set $PD_{i-1}$,
which will be used to establish the stabilization of derivation tree.
That is, an upper bound is given below for the occurrence of the least fixedpoint
in the descending chain of finite unions $PD_i$,
one-to-one corresponding to $PDS_i$.
\begin{lemma}
	Let $PDS_0 \supseteq PDS_1 \supseteq PDS_2 \supseteq \cdots$ be a descending chain of
	finite unions of nonempty subspaces $PDS_i \subseteq \h$,
	as calculated in Proposition~\ref{prop:pd}.
	Then there is an integer $\ell \le d$ such that $PDS_k=PDS_\ell$ holds for all $k>\ell$.
\end{lemma}
\begin{proof}
	The proof improves that of~\cite[Lemma~6]{LYY14}
	by giving the explicit bound $d$.
	It is crucial in establishing the complexity of computing the pure divergent set,
	which is left as an open problem in~\cite[Subsection~7.3]{LYY14}.
	We complete the proof by an induction on the dimension $d_0$ of $PDS_0$.
\begin{itemize}
	\item Basically, when $d_0=0$, we have $PDS_0=\{0\}$.
	It is plainly the fixedpoint of the chain,
	implying the pure divergent set $PD$ is empty then.
	\item Inductively, when $d_0>0$,
	we assume $PDS_0 \supset PDS_1$,
	otherwise it would be much easier.
	Let $PDS_1=\bigcup_{\alpha\in\Sigma} PDS^\alpha$
	where $PDS^\alpha$ are subspaces in the union.
	Define $Z_{\alpha,i}=PDS^\alpha \cap PDS_{1+i}$ for $i \ge 0$.
	We have $PDS_{1+i}=\bigcup_{\alpha\in\Sigma} Z_{\alpha,i}$
	and the following $m$ descending chains:
	\[
		PDS^\alpha =Z_{\alpha,0} \supseteq Z_{\alpha,1} \supseteq Z_{\alpha,2} \supseteq \cdots
		\quad\textup{ for }\alpha \in \Sigma.
	\]
	Since $PDS_0$ is a single subspace, it follows $\dim(PDS^\alpha)<d_0$ by $PDS_0 \supset PDS_1$.
	By induction hypothesis,
	there is an $\ell_\alpha \le \dim(PDS^\alpha)$ in the respective descending chain,
	such that $Z_{\alpha,k}=Z_{\alpha,\ell_\alpha}$ holds for all $k > \ell_\alpha$.
	Thereby, we obtain that $PDS_\ell=PDS_{\ell+1}$ holds
	with $\ell=1+\max_{\alpha\in\Sigma} \dim(PDS^\alpha) \le d_0 \le d$.
	(Once $PDS_0 = PDS_1$, we directly reach this point.)
	We further claim that
	the least fixedpoint of the original descending chain occurs upon $PDS_\ell=PDS_{\ell+1}$, since
	\[
	\begin{aligned}
		PD_{\ell+2}
		& = \bigcup_{\alpha \in \Sigma}
		\{\ket{\psi}\in PD_0: \supp(\F_\alpha(\op{\psi}{\psi})) \subseteq PDS_{\ell+1}\} \\
		& = \bigcup_{\alpha \in \Sigma}
		\{\ket{\psi}\in PD_0: \supp(\F_\alpha(\op{\psi}{\psi})) \subseteq PDS_\ell\} \\
		& =PD_{\ell+1}=PD_\ell
	\end{aligned}
	\]
	and $PD_k=PD_\ell$ follows for all $k>\ell+2$ similarly. \qedhere
\end{itemize}
\end{proof}

The above lemma indicates that
the derivation tree can be stabilized within height $d$ and width $m^d$.
The procedure of computing the pure divergent set $PD$ is stated in Algorithm~\ref{Algo:PD}.
In detail, $\Sch_{i-1}$ in Line~\ref{ln:inner2} stores those finite schedulers $\varsigma$
corresponding to the nodes $PD^\varsigma$ to be derived.
We attempt to derive the node $PD^\varsigma$ in Line~\ref{ln:derive}, provided that
the condition $\bigcup_{\varsigma'\in \Sch'}PD^{\varsigma'}=PD^\varsigma$ in Line~\ref{ln:condition} is not met.
Otherwise, the derivation is unnecessary since the subtree rooted at $PD^\varsigma$ is stabilized then.
The output is a union of various $\Sch_i$ consisting of schedulers
that generate $PD^\varsigma$ stabilized at the $i$th layer.
The complexity of Algorithm~\ref{Algo:PD} is provided below.
\begin{algorithm}[ht]
	\caption{Computing the Pure Divergent Set}\label{Algo:PD}
	\begin{algorithmic}
		\Require a nondeterministic quantum program $\Pro=(\Sigma,\E,\{\M_\yes,\M_\no\})$
        over $\h$ with dimension $d$;
		\Ensure a set $\Sch$ of finite schedulers that generate the pure divergent set $PD$ of $\Pro$.
		\begin{algorithmic}[1]
			\State $\Sch_0 \gets \{\epsilon\}$, and compute $PD^\epsilon$;
			\For{$i \gets 1$ to $d-1$}
			\State $\Sch_i \gets \emptyset$ and $\Sch \gets \emptyset$;
			\While{$\Sch_{i-1} \setminus \Sch \ne \emptyset$}\label{ln:inner2}
			\State let $\varsigma$ be an element of $\Sch_{i-1} \setminus \Sch$;
			\State\label{ln:basis} $\Sch' \gets \{\varsigma\cdot\alpha:\alpha\in\Sigma\}$,
			and compute $PD^{\varsigma'}$ for each $\varsigma'\in \Sch'$;
			\If{$\bigcup_{\varsigma'\in \Sch'}PD^{\varsigma'}=PD^\varsigma$}\label{ln:condition}
			$\Sch \gets \Sch \cup \{\varsigma\}$;
			\Else\label{ln:derive}\
			$\Sch_i \gets \Sch_i \cup \Sch'$
			and $\Sch_{i-1} \gets \Sch_{i-1} \setminus \{\varsigma\}$;
			\EndIf
			\EndWhile
			\If{$\Sch_i=\emptyset$} \textbf{break};
			\EndIf
			\EndFor
			\State \Return $\Sch=\Sch_0 \cup \Sch_1 \cup \cdots \cup \Sch_{i-1}$.
		\end{algorithmic}
	\end{algorithmic}
\end{algorithm}

\paragraph{Complexity}
Note that there are at most $1+m+\cdots+m^{d-1}$ times of entering the inner loop
in Line~\ref{ln:inner2}.
Each inner loop computes $m$ spheres $PD^{\varsigma'}$ in Line~\ref{ln:basis},
which is finished in such a way:
\begin{enumerate}
	\item Write $\varsigma'\in\Sigma^i$ as the form $\alpha\cdot \varsigma''$
	for some $\alpha\in\Sigma$ and $\varsigma''\in\Sigma^{i-1}$.
	\item By the last loop, we have determined the linear subspace $PDS^{\varsigma''}$ as well as
	the orthonormal basis $\{|\psi_j^\perp\rangle:j=1,2,\ldots,J\}$ of its complement where $J<d-i+1$.
	\item Let $\{\FF_k:k=1,2,\ldots,K\}$ be the Kraus representation of $\F_\alpha$ where $K \le d^2$.
	\item By Eq.~\eqref{eq:pd0}, $PD^{\varsigma'}$ is obtained
	as the solution space of $\supp(\F_\alpha(\op{\psi}{\psi})) \subseteq PDS^{\varsigma''}$,
	i.\,e.\@,
	\[
		\bigwedge_{j=1}^J \bigwedge_{k=1}^K \langle\psi_j^\perp|\,\FF_k \ket{\psi}=0.
	\]
	It performs $J\cdot K$ times of matrix-vector multiplication $\langle\psi_j^\perp|\,\FF_k$,
	each of which is done in $\BigO(d^2)$,
	and solves $J\cdot K$ linear equations,
	w.\,r.\,t.\@ $d$ complex variables introduced to encode the pure state $\ket{\psi}$,
	in $\BigO(J\cdot K\cdot d^2) \subseteq \BigO(d^5)$.
\end{enumerate}
Hence Algorithm~\ref{Algo:PD} is in exponential time $\BigO(m^d \cdot d^5)$,
where the growth in the derivation tree is the bottleneck. \qed

\paragraph{Scheduler Synthesis}
For each finite scheduler $\varsigma$ in the output $\Sch$ of Algorithm~\ref{Algo:PD},
we know there is an action $\alpha \in \Sigma$
satisfying $PD^\varsigma=PD^{\varsigma\cdot\alpha}$.
Hence the $\omega$-regular scheduler $\sigma=\varsigma\cdot\alpha^\omega$ is a \emph{divergence scheduler},
under which all states on $PDS^\varsigma$ are terminating with probability zero. \qed

\begin{example}\label{ex3}
Here we will compute the set $\Sch$ of finite schedulers
that generate the pure divergent set $PD$ of the program $\Pro_2$ in Example~\ref{ex1-3}.
Algorithm~\ref{Algo:PD} delivers the inductive process.
\begin{enumerate}
	\item Initially, in the $0$th layer of the derivation tree,
	we have $\Sch_0=\{\epsilon\}$ and
	\[
		PDS_0 = PDS^\epsilon = \spn(\{\ket{0,0},\ket{1,0},\ket{1,1}\})
	\]
	to be derived.
	\item In the first layer,
	we derive $PDS^\epsilon$ for actions $\alpha_1$ and $\alpha_2$,
	and get
	\[
	\begin{aligned}
 		PDS^{\alpha_1} &= \spn(\{\ket{1,1},\ket{-,0}\}), \\
 		PDS^{\alpha_2} &= \spn(\{\ket{0,0},\ket{1,+}\}),
	\end{aligned}
	\]
	which are both proper subspaces of $PDS^\epsilon$.
	So we update $\Sch_0$ to $\emptyset$, and set $\Sch_1=\{\alpha_1,\alpha_2\}$
	and $PDS_1 = PDS^{\alpha_1} \cup PDS^{\alpha_2}$ to be derived.
	\item In the second layer,
	we derive $PDS^{\alpha_1}$ and $PDS^{\alpha_2}$ for actions $\alpha_1$ and $\alpha_2$,
	and get
	\[
	\begin{aligned}
		PDS^{\alpha_1\alpha_1} &= \spn(\{\ket{1,1},\ket{-,0}\}), \\
        PDS^{\alpha_2\alpha1} & = \spn(\{\tfrac{1}{\sqrt{3}}(-\sqrt{2}\ket{0,0}+\ket{1,+})\}), \\
        PDS^{\alpha_1\alpha2} & = \spn(\{\tfrac{1}{\sqrt{3}}(\sqrt{2}\ket{1,1}-\ket{-,0})\}),\\
		PDS^{\alpha_2\alpha_2} &= \spn(\{\ket{0,0},\ket{1,+}\}).
	\end{aligned}
	\]
	Since $PDS^{\alpha_1}=PDS^{\alpha_1\alpha_1}$ and $PDS^{\alpha_2}=PDS^{\alpha_2\alpha_2}$,
	the derivation subtrees rooted at them are stabilized then,
	as well as the whole derivation tree,
	i.\,e.\@ $PDS_2=PDS^{\alpha_1\alpha_1} \cup PDS^{\alpha_2\alpha_1}
	\cup PDS^{\alpha_1\alpha_2} \cup PDS^{\alpha_2\alpha_2}=PDS_1$.
	\end{enumerate}
	Hence, $PD=PD_1$ is the least fixedpoint of the descending chain.
	We report it by the set of finite schedulers
	$\Sch=\Sch_0\cup \Sch_1=\emptyset \cup\{\alpha_1,\alpha_2\}=\{\alpha_1,\alpha_2\}$.
    Additionally,
    we have the divergence schedulers $\alpha_1^\omega$ for those states on $PDS^{\alpha_1}$
	and $\alpha_2^\omega$ for those states on $PDS^{\alpha_2}$. \qed
\end{example}

By Algorithm~\ref{Algo:PD} and the transformation $D=\bigcup_{PD^\varsigma \in PD} \den(PDS^\varsigma)$,
we obtain the result:
\begin{theorem}
    Both pure divergent set and divergent set are computable in exponential time.
\end{theorem}

\section{Deciding the Termination Problem}\label{S6}
Combining the (pure) divergent sets with the reachable spaces obtained in the previous sections,
we are able to decide the termination of the nondeterministic quantum programs.
Although the reachable spaces are supersets of the reachable set,
they could still be utilized to yield a necessary and sufficient condition to the termination
as the following result.

\begin{lemma}\label{lem:terminate}
	Given a nondeterministic quantum program $\Pro$ and an input state $\rho_0\in\den$,
	$\Pro$ terminates with probability less than one under some scheduler
	if and only if
	the I-reachable space $\Phi(\Pro,\rho_0)$ and the pure divergent set $PD(\Pro)$ are not disjoint.
\end{lemma}
\begin{proof}
	We first prove the ``if'' direction.
	Let $\ket{\psi}$ be a pure divergent state in the I-reachable space $\Phi(\Pro,\rho_0)$,
	which is terminating with probability zero under some scheduler $\sigma$, i.\,e.\@,
	\[
		\lim_{i\to \infty}\tr(\M_\yes\F_{\sigma\uparrow i}(\op{\psi}{\psi}))=1.
	\]
	Let $\sigma=\varsigma \cdot \alpha^\omega$ be an $\omega$-regular scheduler
	as the output of Algorithm~\ref{Algo:PD}.
	Since $\Phi(\Pro,\rho_0)=\bigvee_{\gamma\in \Psi(\Pro,\rho_0)}\supp(\gamma)$
    where $\Psi(\Pro,\rho_0)$ is the reachable set,
	there is a finite set of pure reachable states $\ket{\psi_k}$ ($k=1,2,\ldots,K$)
	respectively in the supports of reachable states $\gamma_k \in \Psi(\Pro,\rho_0)$
	reached from $\rho_0$ under finite schedulers $\varsigma_k$,
	such that $\ket{\psi}=\sum_{k=1}^K c_k \ket{\psi_k}$ holds for some $c_k\in\mathbb{C}$.
	We claim that
	at least one, saying $|\psi_{k_0}\rangle$,
	among these $\ket{\psi_k}$ ($k=1,2,\ldots,K$) is terminating with probability less than one
	under the scheduler $\sigma$.
	(Otherwise all $\ket{\psi_k}$ are terminating with probability one under $\sigma$,
	as well as the mixture $\gamma=\tfrac{1}{K}\sum_{k=1}^K \op{\psi_k}{\psi_k}$.
	Since $\ket{\psi}$ is in the support $\spn(\{\ket{\psi_k}:k=1,2,\ldots,K\})$ of $\gamma$,
	by~\cite[Exercise~2.73]{NC00},
	there is a minimal probabilistic ensemble of $\gamma$ containing $\ket{\psi}$ with positive probability.
	Then we reaches the contradiction that $\ket{\psi}$ is terminating with probability one under $\sigma$.)
	Therefore $\Pro$ terminates with probability less than one
	under the nontermination scheduler $\varsigma_{k_0}\cdot\sigma$.
	The workflow is shown in Fig.~\ref{fig:workflow}.
	
	\begin{figure}[ht]
    \centering
    \tikzset{photon/.style={thick, decorate, purple, decoration={snake,segment length=0.8cm,amplitude=5pt}},
	arrow3/.style={thick, decorate, purple, decoration={snake,segment length=0.8cm,amplitude=0pt}},
	global scale/.style={scale=#1, every node/.append style={scale=#1}}}

    \begin{tikzpicture}[global scale=0.8,arrow1/.style = {
	draw = black, dashed, line width=1pt, -{Latex[length = 2mm, width = 2.5mm]}},arrow2/.style = {
	draw = black, thick, -{Latex[length = 2mm, width = 1.5mm]},},
	]

    \draw[fill=purple,fill opacity=0.1] (-5,4) ellipse(1.6 and 1);
    \draw[fill=blue,fill opacity=0.1] (-3,4) ellipse(1.6 and 1);
    \node at (-5.5,4) {$\Phi(\Pro,\rho_0)$};
    \node at (-2.5,4) {$PD(\Pro)$};
    \node at (-4,4) {$\ket{\psi}$};
    \draw[photon,->] (-4,3.5) -- (-4,1);
    \node at (-5.2,2) {$\sigma=\varsigma\cdot\alpha^\omega$};
    \draw[thick] (-5.5,1) -- (-2.5,1) -- (-2.5,0) -- (-5.5,0) -- cycle;
    \node at (-4,0.5) {divergence};
    \path[arrow1] (-1,4) -- node[above] {$\ket{\psi}\in\spn(\{\ket{\psi_1},\ldots,\ket{\psi_K}\})$} (4,4) ;
    \draw[fill=green,fill opacity=0.1] (5,4) circle(0.5);
    \draw[fill=green,fill opacity=0.1] (8,4) circle(0.5);
    \draw[fill=blue, fill opacity=0.2,dashed] (8.6,4.6) rectangle(4.4,3.4);
    \node at (5,4) {$\ket{\psi_1}$};
    \node at (8,4) {$\ket{\psi_K}$};
    \node at (6.5,4) {$\cdots$};
    \draw[fill=red,fill opacity=0.2] (6.5,7) circle(0.5);
    \node at (6.5,7) {$\rho_0$};
    \draw[photon,->] (6.2,6.6) -- (5,4.5);
    \draw[photon,->] (6.8,6.6) -- (8,4.5);
    \node at (5,5.5) {$\varsigma_1$};
    \node at (8,5.5) {$\varsigma_K$};
    \node at (3,2) {$\sigma$};
    \node at (7.5,2) {$\sigma$};
    \path[arrow2] (-2.2,0.5) -- node[above] {imply} (1,0.5) ;
    \draw[thick] (1.3,1) -- (4.3,1) -- (4.3,0) -- (1.3,0) -- cycle;
    \node at (2.8,0.5) {nontermination};
    \draw[thick] (6,1) -- (9,1) -- (9,0) -- (6,0) -- cycle;
    \node at (7.5,0.5) {nontermination};
    \draw[photon,->] (4.8,3.54) -- (2.8,1);
    \draw[photon,->] (8,3.5) -- (8,1);
    \node at (5.15,0.5) {$\cdots$};
    \node at (4.6,0.5) {or};
    \node at (5.65,0.5) {or};
    \end{tikzpicture}
    \caption{The workflow of ``if'' direction}\label{fig:workflow}
    \Description[<short description>]{<long description>}
\end{figure}
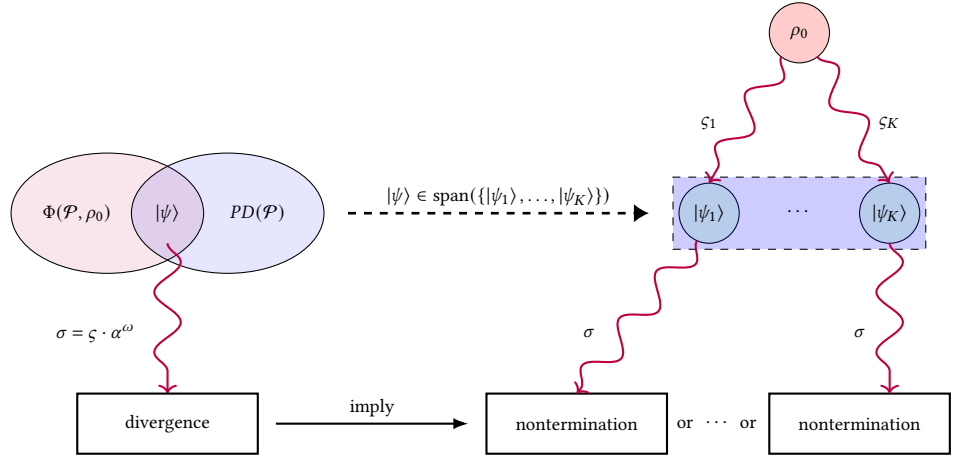

	For the ``only if'' direction, we assume that $\sigma$ is the nontermination scheduler,
	under which $\Pro$ does not terminate with probability one on $\rho_0$.
	From the input state $\rho_0$,
    $\Pro$ terminates with probability less than one.
	Then we will construct a sequence of pure reachable states as:
	\begin{itemize}
		\item fixed a spectral decomposition of $\rho_0$,
		there is an eigenstate $\ket{\lambda_0}$ among eigenstates in the decomposition
		that maximizes the nontermination probability
		\[
			p_0=\lim_{i \to \infty} \tr(\M_\yes \cdot \F_{\sigma\uparrow i}(\op{\lambda_0}{\lambda_0}));
		\]
		\item fixed a spectral decomposition of
		$\F_{\sigma\uparrow 1}(\op{\lambda_0}{\lambda_0})$,
		there is an eigenstate $\ket{\lambda_1}$ that maximizes the nontermination probability
		\[
			p_1=\lim_{i \to \infty}
			\tr(\M_\yes\F_{(\sigma\downarrow 1) \uparrow i}(\op{\lambda_1}{\lambda_1}));
		\]
		\item fixed a spectral decomposition of
		$\F_{\sigma\uparrow 1}(\op{\lambda_1}{\lambda_1})$,
		there is an eigenstate $\ket{\lambda_2}$ that maximizes the nontermination probability
		\[
			p_2=\lim_{i \to \infty}
			\tr(\M_\yes\F_{(\sigma\downarrow 2) \uparrow i}(\op{\lambda_2}{\lambda_2}));
		\]
		\item and so on.
	\end{itemize}
	The nontermination probabilities $p_0,p_1,p_2,\dots$ are monotonously increasing
	and convergent to some limit value $p^*$.
	We proceed to show $p^*=1$.
	For any $\ket{\lambda_i}$ with nontermination probability $p_i<1$,
	we know the termination probability of $\ket{\lambda_i}$
	under the infinite scheduler $\sigma\downarrow i$ is $1-p_i$,
	and there is a finite fragment $(\sigma\downarrow i) \uparrow j$ of $\sigma\downarrow i$
	under which the termination probability of $\ket{\lambda_i}$ is at least $\tfrac{1}{2}(1-p_i)$,
	i.\,e.\@ $\Prob_{(\sigma\downarrow i) \uparrow j}(\op{\lambda_i}{\lambda_i}) \ge \tfrac{1}{2}(1-p_i)$.
	By choosing the eigenstate $|\lambda_{i+j}\rangle$ in that sequence,
	we know that the nontermination probability $p_{i+j}$ of $|\lambda_{i+j}\rangle$
	is not less than the average nontermination probability of
	$\F_{(\sigma\downarrow i) \uparrow j}(\op{\lambda_i}{\lambda_i})$,
	i.\,e.\@ the nontermination probability of the normalized
	$\F_{(\sigma\downarrow i) \uparrow j}(\op{\lambda_i}{\lambda_i})/
	\tr(\F_{(\sigma\downarrow i) \uparrow j}(\op{\lambda_i}{\lambda_i}))$.
	The nontermination probability of
	$\F_{(\sigma\downarrow i) \uparrow j}(\op{\lambda_i}{\lambda_i})$ is still $p_i$,
	while the trace
	$\tr(\F_{(\sigma\downarrow i) \uparrow j}(\op{\lambda_i}{\lambda_i}))$ is
	$1-\Prob_{(\sigma\downarrow i) \uparrow j}(\op{\lambda_i}{\lambda_i}) \le \tfrac{1}{2}(1+p_i)$.
	The average nontermination probability is at least $2p_i/(1+p_i)$,
	which is also a lower bound of the nontermination probability $p_{i+j}$ of $|\lambda_{i+j}\rangle$.
	So we have $p_{i+1} \ge 2p_i/(1+p_i)$.
	Taking the limit, we get $p^* \ge 2p^*/(1+p^*)$, which entails $p^*=1$.

	Those eigenstates $\ket{\lambda_0},\ket{\lambda_1},\ket{\lambda_2},\ldots$
	are unit vectors in the I-reachable subspace $\Phi(\Pro,\rho)$ of $\h$.
	We have that there is a convergent subsequence $|\lambda_1'\rangle,|\lambda_2'\rangle,|\lambda_3'\rangle,\ldots$
	also in the I-reachable subspace $\Phi(\Pro,\rho)$.
	By the completeness of finite-dimensional Hilbert space
	that the limit of a convergent sequence is contained in that space,
	the limit $\ket{\lambda'}$ of the subsequence $|\lambda_1'\rangle,|\lambda_2'\rangle,|\lambda_3'\rangle,\ldots$
    is in the finite-dimensional Hilbert space $\Phi(\Pro,\rho)$,
	which is a pure divergent state
	as $\ket{\lambda'}$ has the nontermination probability $p^*=1$.
\end{proof}

The above proof only tells us that
at least one pure reachable state $|\psi_{k_0}\rangle$ among finitely many ones $\ket{\psi_k}$ ($k=1,2,\ldots,K$)
is terminating with probability less than one,
but does not recognize it.
In the following,
we will recognize this $|\psi_{k_0}\rangle$
by one-by-one checking whether $\ket{\psi_k}$ is terminating with probability less than one.
Conditioning on the nontermination
under the $\omega$-regular scheduler $\sigma=\varsigma\cdot\alpha^\omega$,
we get the following equivalent statements:
\begin{enumerate}
	\item $\ket{\psi_k}$ is terminating with probability less than one.
	\item $\rho=\F_\varsigma(\op{\psi_k}{\psi_k})$ is terminating with probability less than one.
	\item Let $\Pro^{\alpha^\omega}$ be the program $\Pro$ under the scheduler $\alpha^\omega$,
	and $\B$ the I-reachable space $\Phi(\Pro^{\alpha^\omega},\rho)$.
	Then there is a Hermitian operator $\gamma$ on $\B$ such that $\F_\alpha(\gamma)=\gamma$.
\end{enumerate}
The first two statements are equivalent
since $\op{\psi_k}{\psi_k}$ and $\rho$ have the same nontermination probability.
The necessity of the last statement follows
from Brouwer's fixedpoint theorem~\cite[Chapter~4]{Ist01},
since $\F_\alpha$ is a continuous function from the divergent set on $\B$ to itself,
where the divergent set on $\B$ is convex and compact in the viewpoint of probabilistic ensemble.
The sufficiency follows from the fact that
for any pure state $\ket{\psi'}$ in the support of $\gamma$
satisfying $\F_\alpha(\gamma)=\gamma$,
$\supp(\F_\alpha(\op{\psi'}{\psi'}))$ is contained in that $\supp(\gamma)$,
implying $\ket{\psi'}$ is a pure divergent state;
$\ket{\psi'}$ can be linearly expressed by
finitely many pure states $|\psi_k'\rangle$ ($k=1,2,\ldots,K'$) reachable from $\rho$,
at least one among which is terminating with probability less than one.
The workflow is shown in Fig.~\ref{fig:my_label}.

\begin{figure}[H]
    \centering
    \tikzset{photon/.style={thick, decorate, purple, decoration={snake,segment length=0.8cm,amplitude=5pt}},
	arrow3/.style={thick, decorate, purple, decoration={snake,segment length=0.8cm,amplitude=0pt}},
	global scale/.style={scale=#1, every node/.append style={scale=#1}}}

    \begin{tikzpicture}[global scale=0.9,arrow1/.style = {
	draw = black, dashed, line width=1pt, -{Latex[length = 2mm, width = 2.5mm]}},arrow2/.style = {
	draw = black, thick, {Latex[length = 2mm, width = 1.5mm]}-{Latex[length = 2mm, width = 1.5mm]},},
	]
        \draw[fill=yellow,fill opacity=0.1] (12,4) circle(0.5);
        \draw[fill=green,fill opacity=0.1] (8,4) circle(0.5);
        \draw[photon,->] (8.5,4) -- (11.5,4);
        \node at (12,4) {$\rho$};
        \node at (8,4) {$\ket{\psi_k}$};
        \node at (10,4.5) {$\varsigma$};
        \draw[thick] (10.5,1.5) -- (13.5,1.5) -- (13.5,0.5) -- (10.5,0.5) -- cycle;
        \node at (12,1) {nontermination};
        \draw[photon,->] (12,3.5) -- (12,1.5);
        \node at (11.5,2.5) {$\alpha^\omega$};
        \path[arrow1] (13,4) -- node[above] {$\gamma\in\her(\Phi(\Pro^{\alpha^\omega},\rho))$} (16.5,4) ;
        \draw (17.4,4) circle(0.5);
        \draw (17.4,1) circle(0.5);
        \node at (17.4,4) {$\gamma$};
        \node at (17.4,1) {$\gamma$};
        \draw[arrow3,->] (17.4,3.5) -- (17.4,1.5) ;
        \node at (17.1,2.5) {$\alpha$};
        \path[arrow2] (13.8,1) -- node[above] {equivalent} (16.6,1) ;
    \end{tikzpicture}
    \caption{The workflow of checking the nontermination of $\ket{\psi_k}$}\label{fig:my_label}
    \Description[<short description>]{<long description>}
\end{figure}
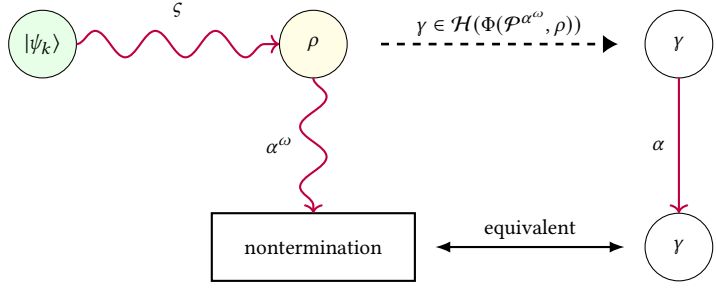

We summarize the procedure of synthesizing a nontermination scheduler as Algorithm~\ref{Algo:scheduler},
whose complexity analysis is provided below.

\begin{algorithm}[ht]
	\caption{Synthesizing a Scheduler for Nontermination}\label{Algo:scheduler}
	\begin{algorithmic}
		\Require a nondeterministic quantum program $\Pro=(\Sigma,\E,\{\M_\yes,\M_\no\})$
        over $\h$ with dimension $d$,
        and an input pure state $\rho_0\in\den(\h)$;
		\Ensure a scheduler under which $\Pro$ terminates with probability less than one on $\rho_0$ if exists.
		\begin{algorithmic}[1]
			\State compute the I-reachable space $\Phi(\Pro,\rho_0)$ by Algorithm~\ref{Algo:RS};
			\State compute the pure divergent set $PD(\Pro)$ by Algorithm~\ref{Algo:PD};
            \If{$\Phi(\Pro,\rho_0) \cap PD(\Pro) \ne \emptyset$}\label{ln:empty}
				\State let $\ket{\psi}$ be an element in $\Phi(\Pro,\rho_0) \cap PD^\varsigma$
				for some $PD^\varsigma\in PD(\Pro)$;
				\State let $\sigma=\varsigma\cdot\alpha^\omega$ be a divergence scheduler of $\ket{\psi}$;
			\Else\ \Return $\epsilon$;
			\Comment{report no nontermination scheduler}
			\EndIf
			\State let $\{\ket{\psi_k}:k=1,2,\ldots,K\}$ be a minimal set of pure reachable states
			under schedulers $\varsigma_k$
			that linearly express $\ket{\psi}$;
            \For{$k \gets 1$ to $K$}\label{ln:loop}
            \Comment{one-by-one checking $\ket{\psi_k}$ for nontermination}
				\State $\rho \gets \F_\varsigma(\op{\psi_k}{\psi_k})$;
				\State compute the I-reachable space $\B=\Phi(\Pro^{\alpha^\omega},\rho)$ by Algorithm~\ref{Algo:RS};
				\If{$\F_\alpha(\gamma)=\gamma$ has some nonzero solution $\gamma\in\her(\B)$}
            	\Return $\varsigma_k\cdot\sigma$.
            	\EndIf
            \EndFor
		\end{algorithmic}
	\end{algorithmic}
\end{algorithm}

\paragraph{Complexity}
Computing $\Phi(\Pro,\rho_0)$ is in $\BigO(m\cdot d^5)$,
and computing $PD(\Pro)$ is in $\BigO(m^d \cdot d^5)$.
The emptiness of $\Phi(\Pro,\rho) \cap PD(\Pro)$ in Line~\ref{ln:empty} can be checked
by computing whether the intersection of $\Phi(\Pro,\rho)$ and $PDS^\varsigma$ is null
for each individual sphere $PD^\varsigma$ in the union $PD(\Pro)$,
which is in at most $m^d \times \BigO(d^3)$.
Once an element $\ket{\psi}$ in $\Phi(\Pro,\rho) \cap PD(\Pro)$ is obtained,
we can find finitely many pure states $\ket{\psi_k}$ ($k=1,2,\ldots,K$)
to linearly express $\ket{\psi}$,
which has been embedded into the computation of $\Phi(\Pro,\rho_0)$.
There are at most $K \le d$ times of entering the loop in Line~\ref{ln:loop}.
Each loop
\begin{enumerate}
	\item performs $\F_\varsigma(\op{\psi_k}{\psi_k})$ in $\BigO(d^5)$,
	since it is $|\varsigma|\le d$ times of performing quantum operations on density operators;
	\item computes $\Phi(\Pro^{\alpha^\omega},\rho)$ which is in $\BigO(d^5)$
	since the action set of $\Pro^{\alpha^\omega}$ is a singleton set $\{\alpha\}$;
	\item solves $\F_\alpha(\gamma)=\gamma$
	which is in $\BigO(d^6)$ since it is a system of linear equations in $d^2$ real variables
	for encoding the Hermitian operator $\gamma\in\her(\B)$.
\end{enumerate}
Hence Algorithm~\ref{Algo:scheduler} is in exponential time $\BigO(m^d \cdot d^5+d^7)$,
whose bottleneck lies in the computation of $PD(\Pro)$. \qed

\begin{example}\label{ex4-1}
	For the while-loop $\Pro_2$ in Example~\ref{ex1-3} with input state $\rho_0=\op{1,1}{1,1}$,
	we have obtained the I-reachable subspace $\Phi(\Pro_2,\rho_0)=\h_{\mathbf{VAR}}$
	and the pure divergent set $PD(\Pro_2)=PD^{\alpha_1} \cup PD^{\alpha_2}$
	with $PDS^{\alpha_1}=\spn(\{\ket{1,1},\ket{-,0}\})$ and $PDS^{\alpha_2}=\spn(\{\ket{0,0},\ket{1,+}\})$,
	on which the divergence schedulers are $\sigma_i=\alpha_i^\omega$ respectively,
	in the previous examples.
  
	The intersection of $\Phi(\Pro_2,\rho_0)$ and $PD(\Pro_2)$ is not empty,
	as it has elements $\ket{1,1}$, $\ket{-,0}$, $\ket{0,0}$ and $\ket{1,+}$.
	It is clear that $\ket{1,1}$ and $\ket{-,0}$ are pure reachable states respectively
	in the supports of $\rho_0$ and $\F_{\alpha_1}(\rho_0)=\frac{1}{2}\op{-,0}{-,0}$.
	To demonstrate the generality of our method,
	we exemplify it with $\ket{0,0}$ to find out a pure reachable state
	that is terminating with probability less than one.
	Since $\ket{0,0}=\ket{1,1}-\sqrt{2}\ket{-,0}+\sqrt{2}\ket{0,-}+2\ket{-,+}$ is linearly expressed by
	the pure reachable states
	$\ket{1,1}$ under the finite scheduler $\varsigma_1=\epsilon$,
	$\ket{-,0}$ under  $\varsigma_2=\alpha_1$,
	$\ket{0,-}$ under $\varsigma_3=\alpha_2$ and $\ket{-,+}$ under $\varsigma_4=\alpha_1\alpha_2$,
	we know that at least one among $\ket{1,1}$,
	$\ket{-,0}$, $\ket{0,-}$ and $\ket{-,+}$ is terminating with probability less than one.

	The two pure states $\ket{1,1}$ and $\ket{-,0}$ are divergent,
	thus they are terminating with probability zero.
	Again, to demonstrate our method,
	we will check whether $\ket{-,+}$ is terminating with probability less than one as follows.
	The I-reachable subspace $\Phi(\Pro_2^{\sigma_2},\op{-,+}{-,+})$ is $\B=\spn(\{\ket{-,+},\ket{0,-},\ket{1,+}\})$.
	Solving $\F_{\alpha_2}(\gamma)=\gamma$ with $\gamma \in \her(\B)$,
	we get a nonzero solution $\op{\phi}{\phi}$
	with $\ket{\phi}=\ket{-,+}+\ket{0,-}/\sqrt{2}+(1+\sqrt{2})\ket{1,+}/\sqrt{2}$.
	Hence the nontermination scheduler $\varsigma_4 \cdot \sigma_2$ is synthesized
	to force $\Pro_2$ to terminate with probability less than one on $\rho_0$,
	which entails the protocol is defective. \qed
\end{example}

By a similar analysis on the the II-reachable space and the divergent set, we get:
\begin{corollary}\label{lem:terminate1}
	Given a nondeterministic quantum program $\Pro$
	and an input state $\rho_0=\op{\lambda_0}{\lambda_0}\in\den$,
	$\Pro$ terminates with probability less than one on $\rho_0$ under some scheduler
	if and only if
	the II-reachable space $\Upsilon(\Pro,\rho_0)$ and the divergent set $D(\Pro)$ are not disjoint.
\end{corollary}
\begin{proof}
	We only prove the ``if'' direction,
while the ``only if'' direction is the same as that of Lemma~\ref{lem:terminate}.
Let $\rho$ be a divergent state in the II-reachable space $\Upsilon(\Pro,\rho_0)$,
which has no probability of termination under some scheduler, saying $\sigma$,
i.\,e.\@, $\lim_{i\to \infty}\tr(\M_\NT\F_{\sigma\uparrow i}(\rho))=1$.
Since $\Upsilon(\Pro,\rho_0)=\bigvee_{\gamma\in \Psi(\hat{\Pro},\rho_0)}\supp(\gamma)$
where $\hat{\Pro}$ is the operator-level program of $\Pro$
and $\Psi(\hat{\Pro},\rho_0)$ is the reachable set of $\hat{\Pro}$,
there is a finite set of pure reachable states $\op{\psi_k}{\psi_k}$ ($k=1,2,\ldots,K$),
respectively reached from $\rho_0$ under finite scheduler $\varsigma_k$,
such that $\rho=\sum_{k=1}^K c_k \op{\psi_k}{\psi_k}$
holds for some $c_k\in\mathbb{R}$.
We claim that
at least one, saying $|\psi_{k_0}\rangle\langle\psi_{k_0}|$,
among these $\op{\psi_k}{\psi_k}$ ($k=1,2,\ldots,K$) has a positive probability of nontermination
under the scheduler $\sigma$,
since otherwise all $\op{\psi_k}{\psi_k}$ have no probability of nontermination under $\sigma$,
as well as their linear combination $\rho$.
Therefore $\Pro$ does not terminate with probability one under the scheduler $\varsigma_{k_0}\cdot\sigma$.
\end{proof}

Algorithm~\ref{Algo:scheduler} could be amended to Lemma~\ref{lem:terminate1}
by checking the emptiness of $\Upsilon(\Pro,\rho_0) \cap D(\Pro)$.
To this end, for each individual sphere $PD^\varsigma$ in $PD(\Pro)$,
we have to
\begin{enumerate}
	\item introduce $\dim(PDS^\varsigma) \le d$ complex variables
        to encode a pure state $\ket{\psi} \in PD^\varsigma$,
    \item introduce $\dim(\Upsilon(\Pro,\rho_0)) \le d^2$ real variables
        to encode an element in $\Upsilon(\Pro,\rho_0)$,
    \item $\op{\psi}{\psi} \in \Upsilon(\Pro,\rho_0)$ results in a polynomial formula in those variables,
		whose coefficients are algebraic numbers.
        It can be solved in $2^{\BigO(d^2)}$
        by the existential theory of the reals~\cite[Theorem~13.13]{BPR06}.
\end{enumerate}
Hence it would contribute an additional factor $2^{\BigO(d^2)}$ to the complexity of the procedure.

\begin{example}\label{ex4-2}
	For the while-loop $\Pro_2$ in Example~\ref{ex1-3} with input state $\rho_0=\op{1,1}{1,1}$,
	we have obtained the II-reachable subspace
	\[
	    \Upsilon(\Pro_2,\rho_0)=\spn\left(\left\{\op{\psi}{\psi}:\ket{\psi}\in\left\{\begin{array}{l}
        \ket{1,1},\ket{-,0},\ket{0,-},\ket{-,+},\ket{+,1},\ket{1,+}, \\
		(\sqrt{2}\ket{0,0}-\ket{1,+})/\sqrt{3},(\ket{-,0}-\sqrt{2}\ket{1,1})/\sqrt{3}
	    \end{array}\right\}\right\}\right)
	\]
	and the divergent set $D(\Pro_2)=\den(PDS^{\alpha_1}) \cup \den(PDS^{\alpha_2})$
	with $PDS^{\alpha_1}=\spn(\{\ket{1,1},\ket{-,0}\})$
	and $PDS^{\alpha_2}=\spn(\{\ket{0,0},\linebreak[0]\ket{1,+}\})$
	in the previous examples.
	They have common elements such as $\op{1,1}{1,1}$,
	which also refutes the termination. \qed
\end{example}

\begin{theorem}
    The termination problem described in Problem~\ref{P1} can be solved in exponential time.
\end{theorem}

\section{Synthesizing a Universal Scheduler}\label{S7}
In this section, we study the universal termination problem,
which asks whether all input states of a program are terminating with probability one under their respective schedulers.
We first decide the universal termination
by detecting the existence of invariant subspace of $\h$.
If the answer is affirmative, we could further synthesize a universal scheduler,
which forces all input states to be terminating with probability one.
The procedure turns out to be in polynomial time,
which is also reported for the first time.

For a nondeterministic quantum program $\Pro=(\Sigma,\E,\{\M_\yes,\M_\no\})$,
the states to be analyzed are those density operators on the subspace $\M_\yes(\h)$,
the null space of $\M_\no$.
Thus we propose:
\begin{definition}[Invariant Space]
    Given a nondeterministic quantum program $\Pro=(\Sigma,\E,\{\M_\yes,\linebreak[0]\M_\no\})$
    with $\Sigma=\{\alpha_j:j=1,2,\dots,m\}$ and $\E(\alpha_j)=\E_j$,
	an invariant space of $\Pro$ is a nonnull subspace $\inv$ of $\M_\yes(\h)$,
    satisfying that $\E_j(\rho) \in \den(\inv)$ holds
    for all input states $\rho \in \den(\inv)$ and all actions $\alpha_j \in \Sigma$.
\end{definition}

From the above definition,
we can see that the invariant subspace $\inv$ of $\M_\yes(\h)$
has the joint semi-lattice structure with ascending chain condition.
That is, for two invariant subspaces $\inv_1$ and $\inv_2$ of $\h$,
the join $\inv_1 \bigvee \inv_2$ is also an invariant subspace;
there is no infinite times of increment in the ascending chain
due to the finite-dimensional $\h$.
Additionally,
the invariant space $\inv$ requires
$\E_j(\den(\inv)) \subseteq \den(\inv)$ holds for all $\alpha_j \in \Sigma$,
entailing $\bigvee_{\rho \in \den(\inv)} \bigvee_{j=1}^m \supp(\E_j(\rho)) \subseteq \inv$.
Define a function $E$ on linear subspaces $\B$
of $\h$ as:
\begin{equation}\label{eq:gfp}
	E(\B):=\bigvee_{\rho \in \den(\B)} \bigvee_{j=1}^m \supp(\E_j(\rho)).
\end{equation}
It is a monotonic function.
For any invariant space $\inv$,
thank to Knaster--Tarski fixedpoint theorem~\cite{CoC77,MOS04},
we know there is a greatest fixedpoint $\inv_0 \subseteq \inv$ such that $E(\inv_0)=\inv_0$.
So we would refer the invariant space $\inv$
as the greatest fixedpoint $\inv_0$ of the function $E$ afterwards.

The existence of invariant space $\inv$ implies that
$\Pro$ terminates on those states $\rho \in \den(\inv)$ with probability zero,
no matter which scheduler is taken.
What is more important is the converse:
\begin{lemma}[{\cite[Theorem~7]{YiY18}}]
    If there is an input state
    on which the program $\Pro$ terminates with probability less than one under any scheduler,
    $\Pro$ has an invariant space $\inv$.
\end{lemma}

To efficiently compute the invariant space $\inv$ defined as the greatest fixedpoint of~\eqref{eq:gfp},
we will derive a series of necessary conditions to characterize $\inv$.
Firstly, we notice there is a density operator $\rho'\in\den(\inv)$ such that
\begin{subequations}
\begin{equation}\label{eq:char1}
	\bigvee_{j=1}^m \supp(\E_j(\rho')) = \bigvee_{\rho \in \den(\inv)} \bigvee_{j=1}^m \supp(\E_j(\rho)) = \inv,
\end{equation}
since $\inv$ is a linear space of finite dimension
and $\rho'$ can be chosen to be a mixture $\sum_k p_k \op{\psi_k}{\psi_k}$
of finitely many pure states $\ket{\psi_k}$,
each contributing at least one linearly independent element in $\inv$.
We further mix $\rho'$ to $\rho''=\tfrac{1}{m+1}[\rho'+\sum_{j=1}^m \E_j(\rho')]$, so that
\begin{equation}\label{eq:char2}
	\supp(\rho'') = \bigvee_{j=1}^m \supp(\E_j(\rho'')).
\end{equation}
On the other hand, we have
\begin{equation}\label{eq:char3}
	\supp(\bar{\E}(\rho'')) = \bigvee_{j=1}^m \supp(\E_j(\rho'')),
\end{equation}
where $\bar{\E}$ is the arithmetic average of $\E$,
i.\,e.\@ $\bar{\E}(\rho'')=\tfrac{1}{m}\sum_{j=1}^m \E_j(\rho'')$.
Combining~\eqref{eq:char2} and~\eqref{eq:char3}, we get
\begin{equation}
	\supp(\rho'') = \supp(\bar{\E}(\rho'')) = \inv.
\end{equation}
\end{subequations}
It yields the nice property that the supports of $\rho''$ and $\bar{\E}(\rho'')$ are both $\inv$.
We collect those density operators $\rho''\in\den(\inv)$ satisfying that property into the set $\Gamma$,
which is convex and compact in the viewpoint of probabilistic ensemble.
Since $\bar{\E}$ is a continuous function from $\Gamma$ to itself,
it follows from Brouwer's fixedpoint theorem~\cite[Chapter~4]{Ist01} that
there exists a fixedpoint $\gamma \in \Gamma$ of $\bar{\E}$ characterized by the stationary equation
\begin{equation}\label{eq:stationary}
    \bar{\E}(\gamma)=\gamma,
\end{equation}
where $\gamma$ is a Hermitian matrix of variables
and $\bar{\E}$ gives rise to coefficients.
The stationary equation is a system of linear equations that can be efficiently solved.
Here we loose the restriction $\gamma \in \Gamma$ to $\gamma \in \her(\M_\yes(\h))$
for the consideration of efficiency,
since Hermitian operators are much easier to be encoded than positive ones.
How to recover from the restriction $\gamma \in \Gamma$ and even compute the invariant space $\inv$
is ensured by the following lemma.

\begin{lemma}[{\cite[Lemma~5.4 \& Algorithm~1]{XFM+21}}]\label{lem:invariant}
    Let $\gamma_0$ be a nonzero solution
	of the stationary equation~\eqref{eq:stationary}.
    Then $\supp(\gamma_0)$ is an invariant space $\inv$ of $\Pro$,
	which can be computed in time $\BigO(d_0^6)$ with $d_0=\dim(\M_\yes(\h))$.
\end{lemma}

\begin{example}\label{ex5-1}
    To explicitly illustrate the method,
    we reset $\E(\alpha_1)=\E_1'=\{\mathbf{X}\otimes \mathbf{X}\}$
	and $\E(\alpha_2)=\E_2'=\{\mathbf{H}\otimes \mathbf{H}\}$
    as the super-operators for the while-loop $\Pro_2$ in Example~\ref{ex1-3}.
	We denote this modified while-loop by $\Pro_2'$.
    We can compute the average super-operator of $\Pro_2'$ 
    as $\bar{\F}'=\{\FF_1',\FF_2'\}$ with Kraus operators
    \[
     \begin{aligned}
        \FF_1' &=  \tfrac{1}{\sqrt{2}}(\op{1,1}{0,0}+\op{0,1}{1,0}+\op{0,0}{1,1}), \\
        \FF_2' &=  \tfrac{1}{\sqrt{2}}(\op{+,+}{0,0}+\op{-,+}{1,0}+\op{-,-}{1,1}).
     \end{aligned}
    \]
    Solving the stationary equation $\bar{\F}'(\gamma)=\gamma$ with $\gamma\in\her(\M_\yes(\h_{\mathbf{VAR}}))$,
    we obtain the unique solution $\gamma_0=(\ket{0,0}+\ket{1,1})(\bra{0,0}+\bra{1,1})$.
    Hence the invariant space $\inv$ of $\Pro'$ is actually $\supp(\gamma_0)=\spn(\{\ket{0,0}+\ket{1,1}\})$,
	which entails that $\Pro_2'$ is not universally terminating. \qed
\end{example}

Whenever the program has no invariant space,
every input state has its own scheduler that achieves the termination with probability one.
In the following, we are to synthesize a universal scheduler
that forces all input states to be terminating with probability one.
The procedure of synthesizing such a universal scheduler is stated in Algorithm~\ref{Algo:universal}.
In detail, each inner loop (Line~\ref{ln:inner3}) attempts to find a pure state $\ket{\psi}$ in
the orthocomplement $\B^\perp$ of $\B$
that is terminating with positive probability under some finite scheduler $\varsigma$.
It is realizable per outer loop (Line~\ref{ln:outer3}), since otherwise
\[
\begin{aligned}
    \neg \exists \ket{\psi}\in\B^\perp \,\exists \alpha_j\in\Sigma:
    \supp(\E_j(\op{\psi}{\psi})) \not\perp \B 
    & \Longleftrightarrow
    \forall \ket{\psi}\in\B^\perp \,\forall \alpha_j\in\Sigma:
    \supp(\E_j(\op{\psi}{\psi})) \subseteq \B^\perp \\
    & \Longleftrightarrow
    \forall \rho\in\den(\B^\perp) \,\forall \alpha_j\in\Sigma:
    \supp(\E_j(\rho)) \subseteq \B^\perp
\end{aligned}
\]
entailing $\B^\perp$ is invariant,
which contradicts the assumption that there is no invariant space.
Utilizing this property,
we avoid the exponential-up enumerating all finite schedulers
with length not greater than $d$
in~\cite[Algorithm~1]{YiY18} for expanding $\B$ by one dimension or more,
which yields the desired polynomial-time efficiency.
The correctness of Algorithm~\ref{Algo:universal} is guaranteed by the following lemma,
and the complexity is provided below the statements of Algorithm~\ref{Algo:universal}.

\begin{lemma}\label{lem:repeat}
	If the finite scheduler $\varsigma$ forces all input states to be terminating with positive probability,
	the infinite scheduler $\varsigma^\omega$ forces all input states to be terminating with probability one.
\end{lemma}
\begin{proof}
	For each density operator $\rho$ on $\B^\perp$,
we know that it has a positive probability $p(\rho)$ of termination under finite scheduler $\varsigma$.
Thus there is an open disk $\delta(\rho)$ around $\rho$,
in which each density operator has termination probability at least $\tfrac{1}{2}p(\rho)$.
Since $\den(\B^\perp)$ is a compact set,
the open cover $\{\delta(\rho): \rho\in\den(\B^\perp)\}$ of $\den(\B^\perp)$
has a subcover $\{\delta(\rho_i): i\in \mathit{IDX}\}$ with some finite index set $\mathit{IDX}$,
i.\,e.\@ $|\mathit{IDX}|<\infty$.
Let $p=\min_{i\in \mathit{IDX}} p(\rho_i)$, which is clearly nonzero.
Performing the finite scheduler $\varsigma$ once,
all density operators $\rho$ on $\B^\perp$ have termination probability at least $\tfrac{p}{2}$;
performing $\varsigma$ infinitely many times,
all density operators $\rho$ on $\B^\perp$ achieve the termination probability one.
\end{proof}

\begin{algorithm}[ht]
	\caption{Synthesizing a Universal Scheduler for Termination}\label{Algo:universal}
	\begin{algorithmic}
		\Require a nondeterministic quantum program $\Pro=(\Sigma,\E,\{\M_\yes,\M_\no\})$
        with $\Sigma=\{\alpha_j:j=1,2,\dots,m\}$ and $\E(\alpha_j)=\E_j$
        over $\h$ with dimension $d$
        that has no invariant space;
		\Ensure a universal scheduler under which $\Pro$ terminates with probability one on all input states.
		\begin{algorithmic}[1]
			\State $\B \gets \M_\no(\h)$ and $\varsigma \gets \epsilon$;
            \While{$\B \ne \h$}\label{ln:outer3}
            \ForAll{$\alpha_j\in\Sigma$}\label{ln:inner3}
			\If{there is a $\ket{\psi} \in \B^\perp$ such that
            $\supp(\E_j(\op{\psi}{\psi}))$ is not orthogonal to $\B$}
            \State let $\ket{\psi}$ be such an element in $\B^\perp$;
			\State $\B \gets \B \vee \spn(\{\ket{\psi}\})$ and $\varsigma \gets \varsigma \cdot \alpha_j$;
            \State \textbf{break};
            \EndIf
            \EndFor
			\EndWhile
			\State \Return $\varsigma^\omega$.
		\end{algorithmic}
	\end{algorithmic}
\end{algorithm}

\paragraph{Complexity}
Note that there are less than $m\cdot d$ times of entering the inner loop in Line~\ref{ln:inner3}.
Each inner loop seeks a pure state $\ket{\psi} \in \B^\perp$,
satisfying that $\EE_k \ket{\psi}$ is not orthogonal to $\B$
for some Kraus operator $\EE_k$ in the Kraus representation $\E_j=\{\EE_k:k=1,2,\ldots,K\}$ where $K \le d^2$.
Let $\{\ket{\psi_l}:l=1,2,\ldots,L\}$ with $L<d$ be the orthonormal basis of $\B$.
For a fixed pair of $\EE_k$ and $\ket{\psi_l}$,
determining whether $\EE_k \ket{\psi}$ is orthogonal to $\ket{\psi_l}$
amounts to solving the linear equation $\bra{\psi_l} \EE_k \ket{\psi}=0$,
which costs $\BigO(d^3)$ operations.
Hence Algorithm~\ref{Algo:universal} is in polynomial time
$\BigO(K\cdot L\cdot m\cdot d^3) \subseteq \BigO(m\cdot d^6)$. \qed

\begin{example}\label{ex5-2}
    Consider the original while-loop $\Pro_2$ attached
	with two nondeterministic super-operators $\E(\alpha_1)=\E_1=\{\mathbf{H} \otimes \mathbf{X}\}$
	and $\E(\alpha_2)=\E_2=\{\mathbf{X} \otimes \mathbf{H}\}$ in Example~\ref{ex1-3},
    the average super-operator $\bar{\F}=\{\FF_1,\FF_2\}$ is given by the Kraus operators
    \[
        \begin{aligned}
        \FF_1 &=  \tfrac{1}{\sqrt{2}}(\op{+,1}{0,0}+\op{-,1}{1,0}+\op{-,0}{1,1}), \\
        \FF_2 &=  \tfrac{1}{\sqrt{2}}(\op{1,+}{0,0}+\op{0,+}{1,0}+\op{0,-}{1,1}).
        \end{aligned}
    \]
    Since the stationary equation $\bar{\F}(\gamma)=\gamma$ with $\gamma\in\her(\M_\yes(\h_{\mathbf{VAR}}))$
	has no nonzero solution,
    $\Pro_2$ has no invariant space and thus is universally terminating.
    By Algorithm~\ref{Algo:universal},
	we can synthesize a universal scheduler that forces all input states to be terminating with probability one
	as follows.
    \begin{enumerate}
    \item Initially, we have $\B_0=\spn(\{\ket{0,1}\})$ and $\varsigma_0=\epsilon$.
    \item For $\B_0^\perp=\spn(\{\ket{0,0},\ket{1,0},\ket{1,1}\})$,
    we can find a pure state $\ket{\psi_1}=\ket{+,0}\in \B_0^\perp$
	such that $\supp(\E_1(\op{\psi_1}{\psi_1}))=\spn(\{\ket{0,1}\})=\B_0$.
    Then we update
    \[
        \B_1 = \B_0\vee\spn(\{\ket{\psi_1}\}) = \spn(\{\ket{0,1},\ket{+,0}\})
		\quad \textup{and} \quad
        \varsigma_1 = \varsigma_0 \cdot \alpha_1 = \alpha_1.
    \]
    \item Next, for $\B_1^\perp=\spn(\{\ket{-,0},\ket{1,1}\})$,
    we can find a pure state $\ket{\psi_2}=\ket{1,1}\in\B_1^\perp$
	such that $\supp(\E_2(\op{\psi_2}{\psi_2}))=\spn(\{\ket{0,-}\})$ which is not orthogonal to $\B_1$.
    Then we update
    \[
        \B_2 = \B_1\vee\spn(\{\ket{\psi_2}\}) = \spn(\{\ket{0,1},\ket{+,0},\ket{1,1}\})
		\quad \textup{and} \quad
		\varsigma_2 = \varsigma_1 \cdot\alpha_2 = \alpha_1\alpha_2.
    \]
    \item Finally, for $\B_2^\perp=\spn(\{\ket{-,0}\})$,
    the pure state $\ket{\psi_3}=\ket{-,0}\in\B_2^\perp$ gives
	$\supp(\E_1(\op{\psi_3}{\psi_3})) \linebreak[0] =\spn(\{\ket{1,1}\})$ which is not orthogonal to $\B_3$,
    and we get
    \[
        \B_3 = \spn(\{\ket{0,1},\ket{+,0},\ket{1,1},\ket{-,0}\}) = \h_{\mathbf{VAR}}
		\quad \textup{and} \quad
		\varsigma_3 = \alpha_1\alpha_2\alpha_1.
    \]
    \end{enumerate}
    Thereby, $\varsigma_3 = \alpha_1\alpha_2\alpha_1$ is the finite scheduler
	that forces all input states to be terminating with positive probability,
    and $\varsigma_3^\omega$ is the infinite scheduler
	that forces all input states to be terminating with probability one. \qed
\end{example}

By Lemma~\ref{lem:invariant} and Algorithm~\ref{Algo:universal}, we obtain the result:
\begin{theorem}
    The universal termination problem described in Problem~\ref{P4} can be solved in polynomial time.
\end{theorem}
As an immediate corollary, we get that
it is in polynomial time to synthesize a scheduler for the termination if exists.

\section{Conclusion}\label{S9}
In this paper,
we have studied the model of nondeterministic quantum program
and the termination and the universal termination problems.
To decide the termination, we needed two ingredients.
One was computing the reachable space of a program fed with an input state,
that was a superset of the set of reachable states
but was of explicit algebraic structure.
A more precise over-approximation of the reachable set was proposed
and could be computed in polynomial time.
The other was computing the divergent set of a program,
which could be obtained in exponential time.
The termination follows from the necessary and sufficient condition that the two sets were disjoint.

For the universal termination,
the necessary and sufficient condition was the existence of invariant space,
which could be detected in polynomial time.
Once a program was decided to be universally terminating,
a universal scheduler would be synthesized in polynomial time
to force all input states to be terminating with probability one.
A case study of the quantum Bernoulli factory protocol was provided to demonstrate our methods.

For future work, we would like to:
\begin{itemize}
	\item consider the weak termination problem,
	as described in Problem~\ref{P2},
	over nondeterministic quantum programs;
	\item synthesize the optimal scheduler that minimizes the expected execution time
	for a specified input state (resp.~all input states with uniform distribution),
	whenever the termination (resp.~universal termination) is guaranteed.
\end{itemize}
Here, for a specified input state $\rho$,
the expected execution time under an infinite scheduler $\sigma$ is defined by
$\mathrm{TE}_\sigma(\rho)=\sum_{i=0}^\infty i\cdot\tr(\M_\yes\F_{\sigma\uparrow i}(\rho))$;
the expected execution time under an optimal infinite scheduler is
$\inf_{\sigma\in\Sigma^\omega}\mathrm{TE}_{\sigma}(\rho)$.
When the input state is unspecified,
we could choose the input state as the uniform distribution $\rho=\id/d$.

\appendix
\section{Exercise~2.73 in Reference~38}\label{A1}
Let $\rho$ be a density operator.
A \emph{minimal ensemble} for $\rho$ is an ensemble $\{(p_k,\ket{\psi_k}):k=1,2,\ldots,K\}$
containing a number $K$ of elements equal to the rank of $\rho$.
Let $\ket{\psi}$ be any state in the support of $\rho$.
Show that there is a minimal ensemble for $\rho$ that contains $\ket{\psi}$,
and moreover that in any such ensemble $\ket{\psi}$ must appear
with a constructive probability.

\begin{proof}
Let $\sum_{k=1}^K p_k\op{\psi_k}{\psi_k}$ be the spectral decomposition of $\rho$,
where $p_k$ are all nonzero eigenvalues of $\rho$,
$\ket{\psi_k}$ are the corresponding eigenvectors,
and $K$ is the rank of $\rho$.
Clearly, the support $\Xi$ of $\rho$ is $\spn(\{\ket{\psi_1},\ket{\psi_2},\ldots,\ket{\psi_K}\})$.
For any $\ket{\psi}\in\Xi$,
by the orthonormality of eigenvectors,
we can uniquely determine the subspace $\Xi' \subseteq \Xi$
that is spanned by those $\ket{\psi_k}$ satisfying $\ip{\psi_k}{\psi} \ne 0$.
Namely, this $\Xi'$ is spanned by $\ket{\psi_{k'}}$.
We claim there is a constructive positive value of $p$,
such that $\det_{\Xi'}(\rho-p\op{\psi}{\psi})=0$,
where $\det_{\Xi'}$ is the determinant on the subspace $\Xi'$ of $\h$.
Then, letting $p_0$ be the smallest value of $p$,
we obtain $\rho=p_0\op{\psi}{\psi}+\rho'$,
where $\rho'$ is a positive operator with less rank than $\rho$.
Under the spectral decomposition $\sum_{j=1}^{K-1} q_j |\varphi_j\rangle\langle\varphi_j|$,
we get
the minimal ensemble $\{(p_0,\ket{\psi})\} \cup \{(q_j,|\varphi_j\rangle):j=1,2,\ldots,K-1\}$
of $\rho$ as desired.

Now we turn to prove the aforementioned claim.
Utilizing the facts:
\begin{itemize}
  \item both $\rho$ and $\op{\psi}{\psi}$ are positive operators, and
  \item the difference $\rho-p\op{\psi}{\psi}$ would be no longer positive
  when $p$ is sufficiently large,
\end{itemize}
the existence of $p_0$ follows by the middle-value theorem,
since $f(p)=\det_{\Xi'}(\rho-p\op{\psi}{\psi})$ is a continuous function in $p$,
satisfying both $f(0)>0$ and $\lim_{p\to\infty}f(p)<0$.
Let $\mathbf{U}=\sum_{k'}\op{k'}{\psi_{k'}}/\sqrt{p_{k'}}$.
It is easy to see $\mathbf{U}\rho \mathbf{U}^\dag=\sum_{k'} \op{\psi_{k'}}{\psi_{k'}}
=\id_{\Xi'}=\sum_{k'} \op{k'}{k'}$
and $\mathbf{U}\op{\psi}{\psi}\mathbf{U}^\dag
=\sum_{k'}|\ip{\psi_{k'}}{\psi}|^2\op{k'}{k'}/p_{k'}$.
Furthermore, such values of $p$ should satisfy
the following equations:
\[
  \begin{aligned}
  \det_{\Xi'}(\rho-p\op{\psi}{\psi})=0
  & \Longleftrightarrow \det_{\Xi'}(\mathbf{U}(\rho-p\op{\psi}{\psi})\mathbf{U}^\dag)=0 \\
  & \Longleftrightarrow \det_{\Xi'}\left[\left(\sum_{k'} \op{k'}{k'}\right)
  -\left(\sum_{k'}\tfrac{p}{p_{k'}}|\ip{\psi_{k'}}{\psi}|^2\op{k'}{k'}\right)\right]=0 \\
  & \Longleftrightarrow \prod_{k'}
  \left(p-\frac{p_{k'}}{|\ip{\psi_{k'}}{\psi}|^2}\right)=0.
  \end{aligned}
\]
It entails that $p_0$ should be chosen as
\[
  \min_{k'}\left\{\frac{p_{k'}}{|\ip{\psi_{k'}}{\psi}|^2}\right\}. \qedhere
\]
\end{proof}

\section{Implementation}\label{B2}
Algorithms 1 through 5 have been implemented
in the Wolfram language on Mathematica 11.3 with an Intel Core i5-4590 CPU at 3.30GHz.
We list below the main functions
for analyzing the termination and the universal termination problems of nondeterministic quantum programs.
\begin{itemize}
    \item \texttt{ReachableSpaceI}(\textit{Operas\_}, \textit{Meas\_}, \textit{Inistate\_}, \textit{Dims\_})
    computes the I-reachable subspace w.\,r.\,t.\@ an input state
    and returns an orthonormal basis of that subspace.
    \item \texttt{ReachableSpaceII}(\textit{Operas\_}, \textit{Meas\_}, \textit{Inibasis\_},
    \textit{Inistate\_}, \textit{Dims\_})
    computes the II-reachable subspace w.\,r.\,t.\@ an input state
    and returns a linearly independent basis of that subspace.
    \item \texttt{Divergent}(\textit{Operas\_}, \textit{Meas\_}, \textit{Dims\_}, \textit{Sigma\_})
    computes the set of finite schedulers and the union of their corresponding generated pure divergent sets
    from which the program has a divergence scheduler,
    i.\,e.\@, the program terminates with probability zero.
    \item \texttt{NTScheduler}(\textit{Operas\_}, \textit{Meas\_}, \textit{Inibasis\_}, \textit{Dims\_},
    \textit{RSI\_}, \textit{PD\_}, \textit{Sigma\_}, \textit{pdAss\_}, \textit{actionAss\_})
    computes a nontermination scheduler under which the program does not terminate with probability one on the input state, 
    once the intersection of the I-reachable space and the pure divergent set is checked to be not empty.
    \item \texttt{UniScheduler}(\textit{Operas\_}, \textit{Meas\_}, \textit{Dims\_})
    computes a universal scheduler under which the program terminates with probability one on all input states.
\end{itemize}
After specifying the Hilbert space,
a nondeterministic quantum program,
and an input state in the format of Wolfram language,
the five algorithms can be performed by calling these functions.

\subsection{Quantum Bernoulli Factory Protocol}
For the nondeterministic quantum program
describing the quantum Bernoulli factory (QBF) protocol in Example~\ref{ex1-1},
by invoking the implemented algorithms,
we have validated the nontermination and the universal termination of the program.
The detailed performance of the five algorithms is shown in Table~\ref{tab:QBF}.

\begin{table}[htp]
    \centering
    \caption{Algorithmic performance on the QBF protocol}\label{tab:QBF}
    \renewcommand\arraystretch{1.2}
    \begin{tabular}{|c|c|c|c|c|c|}
        \hline
         QBF & ReachSpace-I & ReachSpace-II & Divergent & NTScheduler & UniScheduler \\
         \hline
         Example No. & \ref{ex2-1} & \ref{ex2-2} & \ref{ex3} & \ref{ex4-1} & \ref{ex5-2} \\
         \hline
         Time (s) & 0.016 & 0.015 & 1.422 & 0.203 & 2.593 \\
         \hline
         Memory (MB) & 132.605 & 205.833 & 181.868 & 133.217 & 205.833 \\
         \hline
    \end{tabular}
\end{table}

\subsection{Nondeterministic Quantum Walk}
Here we consider another example,
a quantum walk along a ring with three vertexes in a 3-dimensional Hilbert space~\cite{KGK21}.
The vertex set is $V=\{\ket{0},\ket{1},\ket{2}\}$ entailing $\h=\spn(V)$,
where $\ket{0}$ denotes the starting position
and $\ket{2}$ denotes the absorbing boundary.
A projective measurement $\{\M_\yes,\M_\no\}$
with $\M_\no=\op{2}{2}$ and $\M_\yes=\id-\M_\no=\op{0}{0}+\op{1}{1}$
is designed to observe whether the particle is trapped in the boundary after each move.
Each move of the particle is modelled by a quantum operation
which is nondeterministically chosen from actions $\{w_1,w_2\}$,
so that $\E(w_1)=\{\mathbf{W}_1\}$ and $\E(w_2)=\{\mathbf{W}_2\}$ with
\[
    \mathbf{W}_1=\tfrac{1}{\sqrt{3}}\begin{bmatrix}
         1 & 1 & 1 \\
         1 & \varpi & \varpi^2 \\
         1 & \varpi^2 & \varpi
    \end{bmatrix}
    \quad \textup{and} \quad
    \mathbf{W}_2=\tfrac{1}{\sqrt{3}}\begin{bmatrix}
         1 & 1 & 1 \\
         1 & \varpi^2 & \varpi \\
         1 & \varpi & \varpi^2
    \end{bmatrix}
\]
where $\varpi=\mathrm{e}^{2\pi\imath/3}$.
Then we can formally describe the quantum walk with nondeterministic moves
as the program $\Pro_3=(\{w_1,w_2\},\E,\{\M_\yes,\M_\no\})$.

By invoking the implemented algorithms,
we can obtain the following results. 
\begin{itemize}
    \item Starting from position $\ket{0}$,
    the I-reachable space $\Phi(\Pro_3,\op{0}{0})$ of the particle is 
    \[
        \spn(\{\ket{0},(\ket{1}+\ket{2})/\sqrt{2},\imath(\ket{1}-\ket{2})/\sqrt{2}\})=\h,
    \]
    and the II-reachable space $\Upsilon(\Pro_3,\op{0}{0})$ is
    \[
        \spn\left(\left\{
        \begin{array}{l}
            \begin{bmatrix}
                 1 & 0 & 0 \\
                 0 & 0 & 0 \\
                 0 & 0 & 0
            \end{bmatrix},
            \begin{bmatrix}
                 1 & 1 & 1 \\
                 1 & 1 & 1 \\
                 1 & 1 & 1
            \end{bmatrix},
            \begin{bmatrix}
                 8 & 2+2\sqrt{3}\imath & 2-2\sqrt{3}\imath \\
                 2-2\sqrt{3}\imath & 2 & -1-\sqrt{3}\imath \\
                 2+2\sqrt{3}\imath & -1+\sqrt{3}\imath & 2
            \end{bmatrix},
            \begin{bmatrix}
                 8 & 2-2\sqrt{3}\imath & 2+2\sqrt{3}\imath \\
                 2+2\sqrt{3}\imath & 2 & -1+\sqrt{3}\imath \\
                 2-2\sqrt{3}\imath & -1-\sqrt{3}\imath & 2
            \end{bmatrix}, \\
            \begin{bmatrix}
                 14 & 5-\sqrt{3}\imath & 11-5\sqrt{3}\imath \\
                 5+\sqrt{3}\imath & 2 & 5-\sqrt{3}\imath \\
                 11+5\sqrt{3}\imath & 5+\sqrt{3}\imath & 14
            \end{bmatrix},
            \begin{bmatrix}
                 14 & 11-5\sqrt{3}\imath & 5-\sqrt{3}\imath \\
                 11+5\sqrt{3}\imath & 14 & 5+\sqrt{3}\imath \\
                 5+\sqrt{3}\imath & 5-\sqrt{3}\imath & 2
            \end{bmatrix}, \\
            \begin{bmatrix}
                 26 & 17+5\sqrt{3}\imath & 14-2\sqrt{3}\imath \\
                 17-5\sqrt{3}\imath & 14 & 8-4\sqrt{3}\imath \\
                 14+2\sqrt{3}\imath & 8+4\sqrt{3}\imath & 8
            \end{bmatrix},
            \begin{bmatrix}
                26 & 14-2\sqrt{3}\imath & 17+5\sqrt{3}\imath \\
                14+2\sqrt{3}\imath & 8 & 8+4\sqrt{3}\imath \\
                17-5\sqrt{3}\imath & 8-4\sqrt{3}\imath & 14
            \end{bmatrix}      
        \end{array}
        \right\}\right).
    \]
    \item The divergent set is $PD(\Pro_3)=\emptyset$,
    implying no nontermination scheduler,
    thus the particle is proven to be surely absorbed
    no matter which move it takes in each step.
    \item There exists a universal scheduler $w_1^\omega$
    that can force the particle to reach the absorbing boundary
    regardless of its initial position.
\end{itemize}
The detailed performance of the five algorithms is shown in Table~\ref{tab:NQW}.

\begin{table}[htp]
    \centering
    \caption{Algorithmic performance on the nondeterministic quantum walk (NQW)}\label{tab:NQW}
    \renewcommand\arraystretch{1.2}
    \begin{tabular}{|c|c|c|c|c|c|}
        \hline
         NQW & ReachSpace-I & ReachSpace-II & Divergent & NTScheduler & UniScheduler \\
         \hline
         Time (s) & 0.032 & 1.766 & 0.203 & 0.000 & 0.015 \\
         \hline
         Memory (MB) & 94.770 & 205.833 & 98.328 & 100.629 & 176.579 \\
         \hline
    \end{tabular}
\end{table}

Generally speaking, all of the functions involved in Algorithms 1, 2 and 5
are much efficient as their theoretical complexity has an upper bound of \textbf{PTIME}.
Those in Algorithms 3 and 4 may be inefficient in the worst case,
due to the fact that
the derivation tree construction for the pure divergent set
is \textbf{EXPTIME}.
However, their running time is acceptable in our case studies.

\end{document}